\documentclass[12pt]{article}
\usepackage[english]{babel}
\usepackage[utf8]{inputenc}
\usepackage{lastpage}
\usepackage{graphicx}
\usepackage{epstopdf}
\usepackage{pgf,tikz}
\usepackage{mathrsfs}
\usepackage{subfigure}
\usepackage{algorithm}
\usepackage{mathtools}
\usepackage{amsthm}
\usepackage{amsmath}
\usepackage{amssymb}
\usepackage{amsfonts}
\usepackage{calrsfs}
\usepackage{caption}
\usepackage{hyperref}
\usepackage{mathrsfs}
\usepackage[autostyle=true]{csquotes}
\usepackage[noend]{algpseudocode}
\usepackage{algorithmicx}
\usepackage{pgfplots}
\usepackage{varwidth}

\usetikzlibrary{arrows}
\usepackage{orcidlink}

\title{On the Approximation Ratio of the $k$-Opt and Lin-Kernighan Algorithm\footnote{An extended abstract of this work appeared in the proceedings of the 28th Annual European Symposium on Algorithms (ESA 2020).}}
\author{Xianghui Zhong \orcidlink{0000-0003-3812-2903}\\
\small University of Bonn\\
\small Bonn, Germany\\
\small zhong@uni-bonn.de\\[5mm]}
\date{\today} 

\newtheorem{thm}{Theorem}[section]
\theoremstyle{plain}
\newtheorem{lemma}[thm]{Lemma}
\newtheorem{theorem}[thm]{Theorem}
\newtheorem{corollary}[thm]{Corollary}

\theoremstyle{remark}
\newtheorem{remark}[thm]{Remark}
\theoremstyle{definition}
\newtheorem{definition}[thm]{Definition}
\theoremstyle{claim}
\newtheorem{claim}[thm]{Claim}
\newcommand*{\R}{\ensuremath{\mathbb{R}}}

\newcommand*{\N}{\ensuremath{\mathbb{N}}}
\newcommand*{\Z}{\ensuremath{\mathbb{Z}}}

\DeclareMathOperator{\ex}{ex}

\begin{document}

\maketitle
\begin{abstract}
The $k$-Opt and Lin-Kernighan algorithm are two of the most important local search approaches for the \textsc{Metric TSP}. Both start with an arbitrary tour and make local improvements in each step to get a shorter tour. We show that for any fixed $k\geq 3$ the approximation ratio of the $k$-Opt algorithm for \textsc{Metric TSP} is $O(\sqrt[k]{n})$. Assuming the Erd{\H{o}}s girth conjecture, we prove a matching lower bound of $\Omega(\sqrt[k]{n})$. Unconditionally, we obtain matching bounds for $k=3,4,6$ and a lower bound of $\Omega(n^{\frac{2}{3k-3}})$. Our most general bounds depend on the values of a function from extremal graph theory and are tight up to a factor logarithmic in the number of vertices unconditionally. Moreover, all the upper bounds also apply to a parameterized generalization of the Lin-Kernighan algorithm with appropriate parameters. We also show that the approximation ratio of $k$-Opt for \textsc{Graph TSP} is $\Omega\left(\frac{\log(n)}{\log\log(n)}\right)$ and $O\left(\left(\frac{\log(n)}{\log\log(n)}\right)^{\log_2(9)+\epsilon}\right)$ for all $\epsilon>0$. For the \textsc{(1,2)-TSP} we give a lower bound of $\frac{11}{10}$ on the approximation ratio of the $k$-improv and $k$-Opt algorithm for arbitrary fixed $k$.
\end{abstract}

{\small\textbf{keywords:} traveling salesman problem; $k$-Opt algorithm; Lin-Kernighan algorithm; approximation algorithm; approximation ratio}

\section{Introduction}
The traveling salesman problem (TSP) is probably the best-known problem in discrete optimization. An instance consists of the pairwise distances of $n$ vertices and the task is to find a shortest Hamiltonian cycle, i.e.\ a tour visiting every vertex exactly once. The problem is known to be NP-hard \cite{DBLP:books/fm/GareyJ79}. A special case of the \textsc{TSP} is the \textsc{Metric TSP}. Here the distances satisfy the triangle inequality. This \textsc{TSP} variant is still NP-hard  \cite{TSPNPHARD}.

Since the problem is NP-hard, a polynomial-time algorithm is not expected to exist. In order to speed up the calculation of a good tour in practice, several approximation algorithms are considered. The approximation ratio is one way to compare approximation algorithms. It is the maximal ratio, taken over all instances, of the output of the algorithm divided by the optimum solution. For many years the best known approximation algorithm in terms of approximation ratio for \textsc{Metric TSP} was independently developed by Christofides and Serdjukov \cite{christofides,serdjukov} with an approximation ratio of $\frac{3}{2}$. Recently, Karlin, Klein and Oveis Gharan provide an approximation algorithm with approximation ratio of $\frac{3}{2}-\epsilon$ for some $\epsilon>10^{-36}$ \cite{DBLP:conf/stoc/KarlinKG21}. However, in practice other algorithms are easier to implement and have better performance and runtime \cite{DBLP:journals/informs/Bentley92,DBLP:conf/icalp/Johnson90,DBLP:books/sp/Reinelt94}. One natural approach is the $k$-Opt algorithm which is based on local search. It starts with an arbitrary tour and replaces at most $k$ edges by new edges such that the resulting tour is shorter. It stops if the procedure cannot be applied anymore. For the 2-Opt algorithm Plesn\'ik showed that there are infinitely many instances with approximation ratio $\sqrt{\frac{n}{8}}$, where $n$ is the number of vertices \cite{plesnik}. Chandra, Karloff and Tovey showed that the approximation ratio of 2-Opt is at most $4\sqrt{n}$ \cite{Chandra}. Levin and Yovel observed that the same proof yields an upper bound of $\sqrt{8n}$ \cite{Levin}. Recently, Hougardy, Zaiser and Zhong closed the gap and proved that the approximation ratio of the 2-Opt algorithm is at most $\sqrt{\frac{n}{2}}$ and that this bound is tight \cite{HOUGARDY2020401}. For general $k>2$ Chandra, Karloff and Tovey gave a lower bound of $\frac{1}{4}\sqrt[2k]{n}$ \cite{Chandra}. No non-trivial upper bound is known so far. In the case where the instances can be embedded into the normed space $\R^d$ the approximation ratio of 2-Opt is between $\Omega\left(\frac{\log(n)}{\log\log(n)}\right)$ and $O(\log(n))$ \cite{Chandra}. The upper bound was improved by Brodowsky and Hougardy for $\R^2$ to $O\left(\frac{\log(n)}{\log\log(n)}\right)$ which implies a tight asymptotic approximation ratio of $\Theta\left(\frac{\log(n)}{\log\log(n)}\right)$ for the 2-Opt algorithm \cite{brodowsky_et_al:LIPIcs.STACS.2021.18}. 

One of the best practical heuristics by Lin and Kernighan is based on $k$-Opt \cite{DBLP:journals/ior/LinK73}. The Lin-Kernighan algorithm, like the $k$-Opt algorithm, modifies the tour locally to obtain a new tour. Instead of letting all subsets of $k$ edges be eligible for replacement, which results in a high runtime for large $k$, it searches for specific changes: changes, where the edges to be added and deleted alternate in a closed walk, a so-called closed alternating walk. 
Since the Lin-Kernighan algorithm uses a superset of the modification rules of the 2-Opt algorithm, the same upper bound as for 2-Opt also applies. Apart from this, no other upper bound was known.

The \textsc{Graph TSP} is a special case of the \textsc{Metric TSP}. In this case an undirected unweighted graph is given and the distance between two vertices is the distance between them in the graph. Apart from the upper bounds for the \textsc{Metric TSP}, which also apply to the special case, only a lower bound of $2(1-\frac{1}{n})$ on the approximation ratio of the $k$-Opt algorithm is known so far: Rosenkrantz, Stearns and Lewis describe a \textsc{Metric TSP} instance with this ratio that is also a \textsc{Graph TSP} instance \cite{DBLP:journals/siamcomp/RosenkrantzSL77}. 

Another special case is the \textsc{(1,2)-TSP} where every distance is either one or two. The currently best approximation ratio for the \textsc{(1,2)-TSP} is achieved by the $k$-improv algorithm by Berman and Karpinski with an approximation ratio of $\frac{8}{7}$ \cite{DBLP:journals/eccc/ECCC-TR05-069}. The $k$-improv algorithm is an improved version of the $k$-Opt algorithm that is based on a local search approach. Adamaszek, Mnich and Paluch proposed another algorithm with approximation ratio $\frac{8}{7}$ \cite{DBLP:conf/icalp/AdamaszekMP18}. For the \textsc{(1,2)-TSP} it is known that the approximation ratio of the 2-Opt and 3-Opt algorithm are $\frac{3}{2}$ and $\frac{11}{8}$, respectively \cite{khanna1998syntactic,DBLP:journals/orl/Zhong21}.

Beyond the worst-case analysis there are also results about the average case behavior of the algorithm, for example the smoothed analysis of the 2-Opt algorithm by Englert, Röglin and Vöcking \cite{Englert2014}. In their model each vertex of the TSP instance is a random variable distributed in the $d$ dimensional unit cube by a given probability density function $f_i:[0,1]^d \to [0,\phi]$ bounded from above by a constant $1\leq \phi< \infty$ and the distances are given by the $p$-norm. They show that in this case the expected approximation ratio is bounded by $ O(\sqrt[d]{\phi})$ for all $p$. In the model where any instance is given in $[0,1]^d$ and perturbed by Gaussian noise with standard deviation $\sigma$ the approximation ratio was improved to $O(\log(\frac{1}{\sigma}))$ by Künnemann and Manthey \cite{kunnemann2015towards}.

\noindent\textbf{New results.}
For fixed $k\geq 3$, we show that the approximation ratio of the $k$-Opt algorithm is related to the extremal graph theoretic problem of maximizing the number of edges in a graph with fixed number of vertices and no short cycles. Let $\ex(n,2k)$ be the largest number of edges in a graph with $n$ vertices and girth at least $2k$, i.e.\ it contains no cycles with fewer than $2k$ edges. For instances with $n$ vertices we show for \textsc{Metric TSP} that: 

\begin{theorem} \label{intro upper}
For all fixed $k$ if $\ex(n,2k) = O(n^{c})$ for some $c>1$, the approximation ratio of $k$-Opt for \textsc{Metric TSP} is $O(n^{1-\frac{1}{c}})$ where $n$ is the number of vertices.
\end{theorem}

\begin{theorem} \label{intro lower}
For all fixed $k$ if $\ex(n,2k) = \Omega(n^{c})$ for some $c>1$, the approximation ratio of $k$-Opt for \textsc{Metric TSP} is $\Omega(n^{1-\frac{1}{c}})$ where $n$ is the number of vertices.
\end{theorem}

Using known upper bounds on $\ex(n,2k)$ in \cite{Alon} we can conclude:

\begin{corollary} \label{intro upper known}
The approximation ratio of $k$-Opt for \textsc{Metric TSP} is $O(\sqrt[k]{n})$ for all fixed $k$ where $n$ is the number of vertices.  
\end{corollary}

If we further assume the Erd\H{o}s girth conjecture \cite{Erdos}, i.e.\ $\ex(n,2k) = \Theta(n^{1+\frac{1}{k-1}})$, we have:
\begin{corollary} \label{intro lower erdos}
Assuming the Erd\H{o}s girth conjecture, the approximation ratio of $k$-Opt for \textsc{Metric TSP} is $\Omega(\sqrt[k]{n})$ for all fixed $k$ where $n$ is the number of vertices.
\end{corollary}

Using known lower bounds on $\ex(n,2k)$ from \cite{erdHos1966problem, erdos1962problem, Brown1966, benson_1966, singleton, wenger,Lazebnik} we obtain:

\begin{corollary} \label{intro small k}
The approximation ratio of $k$-Opt for \textsc{Metric TSP} is $\Omega(\sqrt[k]{n})$ for $k=3,4,6$ and in $\Omega(n^{\frac{2}{3k-4+o}})$ for all fixed $k$ where $o=0$ if $k$ is even and $o=1$ if $k$ is odd and $n$ is the number of vertices.
\end{corollary}

Comparing our upper and lower bounds we obtain:
\begin{theorem} \label{intro upper sharp}
Our most general upper bound on the approximation ratio of the $k$-Opt algorithm for \textsc{Metric TSP} depending on $\ex(n,2k)$ is tight up to a factor of $O(\log(n))$ where $n$ is the number of vertices.
\end{theorem}

The upper bounds can be carried over to a parameterized generalization of the Lin-Kernighan algorithm we will describe in detail later. In contrast to the original version of the algorithm proposed by Lin and Kernighan two parameters determine the depth that the algorithm searches for improvement.

\begin{theorem} \label{intro metric lin}
The same upper bounds from Theorem \ref{intro upper} of $O(n^{1-\frac{1}{c}})$ if $\ex(n,2k) = O(n^{c})$ and Theorem \ref{intro upper known} of $O(\sqrt[k]{n})$ hold for a parameterized generalization of the Lin-Kernighan algorithm with appropriate parameters.
\end{theorem}

Although the Lin-Kernighan algorithm only considers special changes, namely changes by augmenting a closed alternating walk, we are able to show the same upper bound as for the general $k$-Opt algorithm. For the original version of Lin-Kernighan we get an improved upper bound of $O(\sqrt[3]{n})$. Our results solve two of the four open questions in \cite{Chandra}, namely: 
\begin{itemize}
\item Can the upper bounds given in \cite{Chandra} be generalized to the $k$-Opt algorithm, i.e.\ for increasing $k$ the performance guarantee improves? 
\item Can we show better upper bounds for the Lin-Kernighan algorithm than the upper bound obtained from the 2-Opt algorithm?
\end{itemize}

We also bound the approximation ratio of the $k$-Opt algorithm for \textsc{Graph TSP}.

\begin{theorem}\label{intro graph lower}
The approximation ratio of the $k$-Opt algorithm with fixed $k\geq 2$ for \textsc{Graph TSP} is $\Omega\left(\frac{\log(n)}{\log\log(n)}\right)$ where $n$ is the number of vertices.
\end{theorem}

\begin{theorem}\label{intro graph upper}
The approximation ratio of the $2$-Opt algorithm for \textsc{Graph TSP} is $O\left(\left(\frac{\log(n)}{\log\log(n)}\right)^{\log_2(9)+\epsilon}\right)$ for all $\epsilon>0$ where $n$ is the number of vertices.
\end{theorem}

Note that the same upper bound also applies to the $k$-Opt algorithm and the Lin-Kernighan algorithm since they produce 2-optimal tours. Hence, up to a constant factor of at most $\log_2(9)$ in the exponent the $k$-Opt algorithm does not achieve asymptotically better performance than the 2-Opt algorithm in contrast to the metric case.

Furthermore, we show a lower bound on the approximation ratio of the $k$-Opt and $k$-improv algorithm for the \textsc{(1,2)-TSP}.

\begin{theorem} \label{lower Bound k-improv}
The $k$-Opt and $k$-improv algorithm with arbitrary fixed $k$ have an approximation ratio of at least $\frac{11}{10}$ for the \textsc{(1,2)-TSP}.
\end{theorem}

\noindent\textbf{Outline of the paper.}
First, we sum up some previous results by Chandra, Karloff and Tovey and results from extremal graph theory we need for this paper. In Section~\ref{sec outline of analysis} we give an outline of the main ideas of the analysis of the upper bounds. 

In Section~\ref{sec lower bound metric tsp} we improve the existing lower bound for the \textsc{Metric TSP} by weakening the condition for the existing construction of bad instances given in \cite{Chandra}. In Section~\ref{sec upper bound for metric tsp} we prove the upper bound of the approximation ratio for \textsc{Metric TSP}. For that we assume that the optimal tour and the output of the $k$-Opt or Lin-Kernighan algorithm with the largest approximation ratio are given. Our aim is to show that the output of the algorithm does not have too many long edges compared to the optimal tour. To achieve this, we first divide the edges into length classes, such that the longest edge from each class is at most a constant times longer than the shortest. Then, we construct with help of the optimal tour a graph containing at least $\frac{1}{4}$ of the edges in a length class. We show that this graph has a high girth and use results from extremal graph theory to bound the number of its edges, which implies that the length class does not contain too many edges. In Section~\ref{sec comparing lower and upper bound} we compare the lower and upper bound we got from the previous sections and show that they differ asymptotically only by a logarithmic factor even if the exact behavior of $\ex(n,2k)$ is unknown.

Then, in Sections~\ref{sec lower bound graph tsp} and \ref{sec upper bound graph tsp} we give lower and upper bounds on the $k$-Opt algorithm for \textsc{Graph TSP}. For the lower bound, we construct an instance and a $k$-optimal tour with the appropriate approximation factor, again using results from extremal graph theory. To show the upper bound, starting with a worst-case instance we iteratively decompose the current graph into smaller graphs with small diameter and contract these smaller graphs into single vertices. We show that a certain subset of the vertices, the so-called active vertices, shrinks by a factor exponential in the approximation ratio after a sufficient number of iterations. Moreover, we show that after that many iterations we still have at least one active vertex. We conclude that the number of active vertices and hence the number of vertices in the beginning depends exponentially on the approximation ratio.

Finally, in Section \ref{sec 1,2 TSP lower} we give a lower bound of $\frac{11}{10}$ on the approximation ratio of the $k$-improv and $k$-Opt algorithm for the \textsc{(1,2)-TSP} and arbitrary fixed $k$. For every fixed $k$ we construct a family of \textsc{(1,2)-TSP} instances together with $k$-optimal tours from regular graphs with high girth. We show that the ratios of the length of the constructed tours to that of the optimal tours converge to $\frac{11}{10}$.

\subsection{Preliminaries}
\subsubsection{TSP}
An instance of \textsc{Metric TSP} is given by a complete weighted graph $(K_n,c)$ where the costs are non-negative and satisfy the triangle inequality: $c(\{x,z\})+c(\{z,y\})\geq c(\{x,y\})$ for all $x,y,z\in V(K_n)$. A \emph{cycle} is a closed walk that visits every vertex at most once. A \emph{tour} is a cycle that visits every vertex exactly once. For a tour $T$, let the \emph{length} of the tour be defined as $c(T):=\sum_{e\in T}c(e)$. The task is to find a tour of minimal length. We fix an \emph{orientation} of the tour, i.e.\ we consider the edges of the tour as directed edges such that the tour is a directed cycle. From now on, let $n$ denote the number of vertices of the instance.

\textsc{Graph TSP} is a special case of the \textsc{Metric TSP}. Each instance arises from an unweighted, undirected connected graph $G$. To construct a TSP instance $(K_n,c)$, we set $V(K_n)=V(G)$. The cost $c(\{u,v\})$ of the edge connecting any two vertices $u,v\in V(G)$ is given by the length of the shortest $u$-$v$-path in $G$.

For the \textsc{(1,2)-TSP} the distances between the vertices are restricted to be equal to 1 or 2. Note that this variant of the TSP is metric since three edges of length 1 or 2 always satisfy the triangle inequality.

An algorithm $A$ for the traveling salesman problem has \emph{approximation ratio}
$\alpha(n)\ge 1$ if for every TSP instance with $n$ vertices it finds a tour that is at most $\alpha(n)$ times as long as a shortest tour and this ratio is achieved by an instance for every $n$. Note that we require here the sharpness of the approximation ratio, deviating from the standard definition in the literature to express the approximation ratio in terms of the Landau symbols. Nevertheless, the results also hold for the standard definition with more complicated notation.

\subsubsection{$k$-Opt Algorithm}
A \emph{$k$-move} replaces at most $k$ edges of a given tour by other edges to obtain a new tour. It is called \emph{improving} if the resulting tour is shorter than the original one. A tour is called \emph{$k$-optimal} if there is no improving $k$-move.
\begin{algorithm}[H]
\caption{$k$-Opt Algorithm}
 \hspace*{\algorithmicindent} \textbf{Input:} Instance of \textsc{TSP} $(K_n,c)$ \\
 \hspace*{\algorithmicindent} \textbf{Output:} Tour $T$
\begin{algorithmic}[1]
\State Start with an arbitrary tour $T$
\While{$\exists$ improving $k$-move for $T$}  
\State Perform an improving $k$-move on $T$
\EndWhile
\end{algorithmic}
\end{algorithm}

For the 2-Opt algorithm recall the following well-known fact: Given a tour $T$ with a fixed orientation, it stays connected if we replace two edges of $T$ by the edge connecting their heads and the edge connecting their tails, i.e.\ if we replace edges $(a,b),(c,d)\in T$ by $(a,c)$ and $(b,d)$.

\subsubsection{$k$-Improv Algorithm} \label{subsec k-improv}
In this section we describe the $k$-improv algorithm, introduced by Berman and Karpinski in \cite{DBLP:journals/eccc/ECCC-TR05-069}, which is an improved version of the $k$-Opt algorithm for the \textsc{(1,2)-TSP}. In the same paper it was shown that this algorithm has an approximation ratio of $\frac{8}{7}$ for $k=15$ which is the currently best approximation ratio for the \textsc{(1,2)-TSP}. 

A \emph{2-matching} is the union of disjoint paths and cycles. A \emph{$k$-improv-move} deletes and adds a total of at most $k$ edges of a 2-matching to obtain a new 2-matching. In contrast to a $k$-move the number of removed and added edges do not have to be equal. A $k$-improv-move is called \emph{improving} if the result $\widetilde{T'}$ after performing the $k$-improv-move on $\widetilde{T}$ satisfies the following conditions:
\begin{enumerate}
\item $\widetilde{T'}$ only contains edges with cost 1.
\item One of the following properties hold:
   \begin{itemize}
     \item $\widetilde{T'}$ contains fewer connected components than $\widetilde{T}$.
     \item $\widetilde{T'}$ contains the same number of connected components as $\widetilde{T}$, but more cycles than $\widetilde{T}$.
     \item $\widetilde{T'}$ contains the same number of connected components and cycles as $\widetilde{T}$, but fewer \emph{singletons}, i.e.\ vertices with degree 0, than $\widetilde{T}$. 
   \end{itemize}
\end{enumerate}

The algorithm is also a local search algorithm. It starts with an arbitrary tour and removes all edges with cost 2 to obtain a 2-matching consisting of edges with cost 1. During each iteration of the algorithm we perform an improving $k$-improv-move. Note that this way we maintain a 2-matching in every iteration. We call a 2-matching \emph{$k$-improv-optimal} if there are no improving $k$-improv-moves. If this is the case, we remove an arbitrary edge from every cycle in $\widetilde{T}$ and after that connect the paths in $\widetilde{T}$ arbitrarily to a tour $T$ (Algorithm \ref{k improv algorithm}).

\begin{algorithm}
\caption{$k$-Improv Algorithm}
\label{k improv algorithm}
 \hspace*{\algorithmicindent} \textbf{Input:} Instance of (1,2)-\textsc{TSP} $(K_n,c)$ \\
 \hspace*{\algorithmicindent} \textbf{Output:} Tour $T$
\begin{algorithmic}[1]
\State Start with an arbitrary tour $T$
\State Let $\widetilde{T}$ be the 2-matching we obtain by removing all edges of cost 2 from $T$
\While{$\exists$ improving $k$-improv-move for $\widetilde{T}$}  
\State Perform an improving $k$-improv-move on $\widetilde{T}$
\EndWhile
\State Remove an arbitrary edge from each cycle in $\widetilde{T}$
\State Connect the paths in $\widetilde{T}$ arbitrarily to a tour $T$
\State \Return $T$
\end{algorithmic}
\end{algorithm}

The $k$-improv algorithm for fixed $k$ runs in polynomial time as shown in \cite{DBLP:journals/eccc/ECCC-TR05-069}.

\subsubsection{Lin-Kernighan Algorithm}
We analyze a parameterized generalization of the Lin-Kernighan algorithm described in Section 21.3 of \cite{Korte:2007:COT:1564997}. In this version two parameters $p_1$ and $p_2$ specify the depth the algorithm is searching for improvement.

An \emph{alternating walk} of a tour $T$ is a walk starting with an edge in $T$ where exactly one of two consecutive edges is in $T$. An edge of the alternating walk is called \emph{tour edge} if it is contained in $T$, otherwise it is called \emph{non-tour edge}. A \emph{closed alternating walk} and \emph{alternating cycle} are alternating walks whose edges form a closed walk and cycle, respectively. The symmetrical difference of two sets $A$ and $B$ is the set $A\triangle B:=(A\cup B)\backslash (A\cap B)$. When we \emph{augment} $T$ by an augmenting cycle $C$ we get the result $T\triangle C$. Moreover, $C$ of a tour $T$ is called \emph{improving} if $T\triangle C$ is a shorter tour than $T$. By $(x_1,x_2,\dots,x_j):=\cup_{i=1}^{j-1}(x_i,x_{i+1})$, we denote the walk that visits the vertices $x_1,x_2\dots, x_j$ in this order. We define the \emph{gain} $g$ of an alternating walk by 
\begin{align*}
g((x_0,x_1,\dots, x_{2m})):=\sum_{i=0}^{m-1}c(x_{2i},x_{2i+1})-c(x_{2i+1},x_{2i+2}).
\end{align*}

An alternating walk $(x_0,x_1,\dots, x_{2m})$ is \emph{proper} if $g((x_0,x_1,\dots, x_{2i}))>0$ for all $i\leq m$. 

The following theorem by Lin and Kernighan allows performance improvements of the Lin-Kernighan algorithm by only looking for proper alternating walks without changing the quality of the result.

\begin{theorem}[\cite{DBLP:journals/ior/LinK73}]
For every improving closed alternating walk $P$ there exists a proper closed alternating walk $Q$ with $E(P)=E(Q)$.
\end{theorem}

Now, we state the generalized version of the Lin-Kernighan algorithm with parameters $p_1$ and $p_2$:
\begin{itemize}
\item The algorithm starts with an arbitrary tour and searches for an improving closed alternating walk in every iteration by a depth-first search. 
\item At depth zero the list of \emph{candidate vertices} consists of all vertices of the instance. 
\item At each depth it chooses a vertex from the list of candidate vertices, computes the list of candidate vertices for the next depth and increases the depth. 
\item The list of candidate vertices consists of all vertices forming with the vertices already chosen in previous iterations an alternating walk with a positive gain. 
\item At each depth it checks if connecting the endpoints of the alternating walk results in an improving closed alternating walk. 
\item When the depth is higher than $p_2$ and even, we further require the candidate vertices satisfying the following condition: After choosing any candidate vertex in the next iteration and connecting the endpoints of the resulting alternating walk, we get an improving closed alternating walk. 
\item When no candidates are available at depth $i$ anymore, it backtracks to the depth $\min\{p_1,i-1\}$ and chooses the next candidate at that depth. 
\item It terminates if no improving closed alternating walk is found. Otherwise, it improves the current tour by augmenting the improving closed alternating walk with the highest gain it found and repeats the process.
\end{itemize}

\begin{algorithm}[H]
\caption{Lin-Kernighan Algorithm}
 \hspace*{\algorithmicindent} \textbf{Input:} Instance of \textsc{TSP} $(K_n,c)$, Parameters $p_1, p_2\in \N$ \\
 \hspace*{\algorithmicindent} \textbf{Output:} Tour $T$
\begin{algorithmic}[1]
\State Start with an arbitrary tour $T$
\State Set $X_0:=V(K_n), i:=0$ and $g^*:=0$
\While{$i\geq 0$}
\If {$X_i=\emptyset$}
\If{$g^*>0$}
\State Set $T:=T\triangle P^*$
\State Set $X_0:=V(K_n), i:=0$ and $g^*:=0$
\Else
\State Set $i:=\min\{i-1,p_1\}$
\EndIf
\Else
\State Choose $x_i\in X_i$, set $X_i:=X_i\backslash \{x_i\}$
\State Set $P:=(x_0,x_1,\dots, x_i)$
\If {$i$ is odd}
\If{$i\geq 3, T \triangle (P\cup(x_i,x_0)) \text{ is a tour}, g(P \cup(x_i,x_0))>g^*$}
\State Set $P^*:=P\cup(x_i,x_0)$ and $g^*:=g(P^*)$
\EndIf
\State 
\begin{varwidth}[t]{\linewidth}
Set $X_{i+1}:= \{x\in V(K_n) \backslash\{x_0,x_i\}: \{x,x_0\}\not\in T\cup P,$ \par
\hskip\algorithmicindent $T\triangle (P\cup(x_i,x,x_0)) \text{ is a tour}, g(P\cup (x_i,x))>g^*\}$
\end{varwidth}
\EndIf
\If{$i$ is even}
\If{$i \leq p_2$}
\State Set $X_{i+1}:=\{x\in V(K_n):\{x_i,x\}\in T\backslash P\}$
\Else
\State 
\begin{varwidth}[t]{\linewidth}
Set $X_{i+1}:=\{x\in V(K_n):\{x_i,x\}\in T\backslash P,$ \par
\hskip\algorithmicindent $\{x,x_0\}\not\in T\cup P, T\triangle (P\cup(x_i,x,x_0)) \text{ is a tour}\}$
\end{varwidth}
\EndIf
\EndIf
\State $i:=i+1$
\EndIf
\EndWhile
\end{algorithmic}
\end{algorithm}

In the original paper Lin and Kernighan described the algorithm with fixed parameters $p_1=5,p_2=2$. 

\begin{definition}
We call the Lin-Kernighan algorithm with parameter $p_1= 2k-1$ and $p_2= 2k-4$ the \emph{$k$-Lin-Kernighan algorithm}. A tour is \emph{$k$-Lin-Kernighan optimal} if it is the output of the $k$-Lin-Kernighan algorithm for some initial tour. 
\end{definition}
Note that the original version of the Lin-Kernighan algorithm is the 3-Lin-Kernighan algorithm. By the description of the algorithm, it is easy to see that all local changes of the Lin-Kernighan algorithm are augmentations of an improving closed alternating walk and: 

\begin{lemma}
The length of any improving alternating cycle in a $k$-Lin-Kernighan optimal tour is at least $2k+1$.
\end{lemma}

Obviously, this property also holds for the output of the Lin-Kernighan algorithm with parameters $p_1\geq 2k-1, p_2\geq 2k-4$ and our results carry over to this case. 
 
\subsubsection{Girth and Ex}

\begin{definition}
The \emph{girth} of a graph is the length of the shortest cycle contained in the graph if it contains a cycle and infinity otherwise. Let $\ex(n,2k)$ be the maximum number of edges in a graph with $n$ vertices and girth at least $2k$. Moreover, define $\ex^{-1}(m,2k)$ as the minimum number of vertices of a graph with $m$ edges and girth at least $2k$. 
\end{definition}

There are results in extremal graph theory on the behavior of the function $\ex(n,2k)$.

\begin{theorem}[\cite{Alon}] \label{exupper}
We have
\begin{align*}
\ex(n,2k)<\frac{1}{2^{1+\frac{1}{k-1}}}n^{1+\frac{1}{k-1}}+\frac{1}{2}n.
\end{align*}
\end{theorem}

\begin{theorem}[\cite{Lazebnik}] \label{exlower}
We have
\begin{align*}
\ex(n,2k)= \Omega(n^{1+\frac{2}{3k-6+o}}),
\end{align*}
where $k\geq 3$ is fixed, $o= 0$ if $k$ is even, $o=1$ if $k$ is odd and $n\to \infty$.
\end{theorem}

\begin{theorem}[Polarity Graph in \cite{erdHos1966problem, erdos1962problem, Brown1966}; Construction by Benson and by Singleton \cite{benson_1966, singleton}; Construction by Benson and by Wenger \cite{benson_1966, wenger}] \label{omega small cases}
For $k=3,4,6$ we have
\begin{align*}
\ex(n,2k)= \Omega(n^{1+\frac{1}{k-1}}).
\end{align*}
\end{theorem}

\begin{theorem}[Theorem 1.4' in Section III of \cite{bollobas2004extremal}] \label{regular high girth}
Let $\delta,g \geq 3$ and
\begin{align*}
m\geq \frac{(\delta-1)^{g-1}-1}{\delta-2}
\end{align*}
be integers. Then, there exists a $\delta$-regular graph with $2m$ vertices and girth at least $g$.
\end{theorem}

\section{Outline of the Analysis} \label{sec outline of analysis}
In this section, we give an outline of the analysis for the upper bounds for the \textsc{Metric TSP} and \textsc{Graph TSP}.

\subsection{Outline of Upper Bound for Metric TSP}
In this subsection we briefly summarize the ideas for the analysis of the upper bound for the \textsc{Metric TSP} given by Theorem \ref{intro upper}. 

For a fixed $k$ assume that an instance is given with a $k$-optimal tour $T$. We fix an orientation of $T$ and assume w.l.o.g.\ that the length of the optimal tour is 1. To bound the approximation ratio it is enough to bound the length of $T$. Our general strategy is to construct an auxiliary graph depending on $T$ and bound its girth. More precisely, we show that if this graph has a short cycle this would imply the existence of an improving $k$-move contradicting the $k$-optimality of $T$. Moreover, the auxiliary graph contains many long edges of $T$ so the bound on its girth also bounds the number of long edges in the tour and hence the approximation ratio.

Let the graph $G$ consist of the vertices of the instance and the edges of $T$, i.e. $G:=(V(K_n),T)$. We first partition the edges of $T$ into length classes such that the $l$th length class consists of the edges with length between $c^{l+1}$ and $c^{l}$ for some constant $c<1$. We call these edges \emph{$l$-long} (Definition~\ref{llong}). For each $l\in \mathbb{N}_0$ we want to get an upper bound on the number of $l$-long edges that depends on the number of vertices. 

If we performed the complete analysis on $G$, we would get a bad bound on the number of $l$-long edges since $G$ contains too many vertices. To strengthen the result we first construct an auxiliary graph containing all $l$-long edges for some fixed $l$ but fewer vertices and bound the number of $l$-long edges in that graph: We partition $V(G)$ into classes with the help of the optimal tour. Using the metric property, we ensure that in each class any two vertices have small distance to each other (Definition~\ref{defnear}). We contract the vertices in each class to one vertex and delete self-loops to get the multigraph $G_1$ (Definition~\ref{defG1}). We can partition $V(G)$ in such a way that $G_1$ contains all the $l$-long edges. Note that we did not delete parallel edges in $G_1$ and hence every edge in $G_1$ has a unique preimage in $G$.

Unfortunately, we cannot directly bound the girth of $G_1$ since the existence of a short cycle would not necessarily imply an improving $k$-move for $T$. For that we need a property of the cycles in the graph: The common vertex of consecutive edges in any cycle has to be head of both or tail of both edges according to the orientation of $T$. Therefore, we construct the auxiliary graph $G_2$ from $G_1$ as follows: We start with $G_2$ as a copy of $G_1$ and color the vertices of $G_2$ red and blue. We only consider $l$-long edges in $G_2$ from a red vertex to a blue vertex according to the orientation of $T$ and delete all other edges (Definition~\ref{defG2}). We show (Lemma~\ref{coloring}) that the coloring can be done in such a way that at least $\frac{1}{4}$ of the $l$-long edges remain in $G_2$.

We claim that the underlying undirected graph of $G_2$ has girth at least $2k$. Note that by construction the graph is bipartite and hence all cycles have even length. Assume that there is a cycle $C$ with $2h<2k$ edges. We call the preimage of the edges of $C$ in $G$ the \emph{C-edges}. Our aim is to construct a tour $T'$ with the assistance of $C$ that arises from $T$ by an improving $h$-move (Theorem \ref{k-opt}). This contradicts the $k$-optimality of $T$ as we assumed that $h<k$.

For every common vertex $w$ of two consecutive edges $e_1,e_2$ of $C$ in $G_2$ we consider the preimage $e_1^{-1},e_2^{-1}$ of $e_1,e_2$ in $G$. Then there have to be endpoints $u\in e_1^{-1}$ and $v\in e_2^{-1}$ such that the images of $u$ and $v$ after the contraction in $G_2$ are both $w$. We will call the edge $\{u,v\}$ a \emph{short edge} (Definition~\ref{defShortEdge}). In fact since both endpoints of a short edge are mapped to the same vertex in $G_1$ after the contraction and we contracted vertices that have a small distance to each other, they are indeed short. Furthermore, we can show that the total length of all the short edges is shorter than that of any single $C$-edge (Lemma~\ref{number of blue edges}). The number of the short edges is equal to the number of $C$-edges which is $2h$. Now, observe that the cycle $C$ defines an alternating cycle in $G$ in a natural way: Let the preimages of $C$ in $G$ be the tour edges and that of the short edges be the non-tour edges (Lemma~\ref{short edges}).

To construct a new tour $T'$ from $T$ we start by augmenting the alternating cycle. Afterward, the tour may split into at most $2h$ connected components. A key property is that the coloring of the vertices in $G_2$ ensures that every connected component contains at least two short edges. Since there are $2h$ short edges, we know that after the augmentation we actually get at most $h$ connected components (Lemma~\ref{connected comp}). To reconnect and retain the degree condition we twice add a set $L$ of at most $h-1$ different $C$-edges, i.e. in total at most $2h-2$ edges. In the end we shortcut to the new tour $T'$ in a particular way without decreasing $\lvert T \cap T' \rvert$. 

Note that the number of $C$-edges in the original tour $T$ is $2h$, thus $T'$ contains at least 2 fewer $C$-edges than $T$. The additional short edges that $T'$ contains are cheap, therefore $T'$ is cheaper than $T$. Moreover, $T'$ arises from $T$ by replacing at most $2h-\lvert L \rvert$ $C$-edges since we deleted the $C$-edges and added twice the set $L$ consisting of $C$-edges. 

Therefore, we know that $T'$ arises from $T$ by a $2h-\lvert L\rvert\leq 2h$-move. By the $k$-optimality of $T$, we have $2h> k$ or $2h\geq k+1$. This already gives us a lower bound of $k+1$ for the girth of the graph $G_2$ as $C$ contains $2h$ edges (Remark \ref{rem2h}). 

In the next step we use the previous result to show that there is actually a cheaper tour $T'$ that arises by an $h+1$-move. This implies that $h+1>k$ or $2h\geq 2k$, i.e.\ the girth of $G_2$ is at least $2k$. As we have seen above the number of edges we have to replace to obtain $T'$ from $T$ depends on $\lvert L \rvert$, the number of $C$-edges $T'$ contains. Therefore, we modify $T'$ iteratively such that the number of $C$-edges in $T'$ increases by 1 after every iteration while still maintaining the property that $T'$ is cheaper than $T$. We stop when the number of $C$-edges in $T'$ is $h-1$ as then $T'$ would arise from $T$ by a $2h-(h-1)=h+1$-move. 

To achieve this we start with the constructed tour $T'$ and iteratively perform ``ambivalent'' (Definition \ref{def ambivalent 2-move}) 2-moves that are not necessarily improving but add at least one more $C$-edge to $T'$. In every iteration we consider $C$-edges $e$ not in the current tour $T'$. We can show that there is an edge in $T'\backslash T$ incident to each of the endpoints of $e$. Let the two edges be $f_1$ and $f_2$. We want to replace $f_1$ and $f_2$ in $T'$ by $e_1$ and the edge connecting the endpoints of $f_1$ and $f_2$ not incident to $e$. To ensure the connectivity after the 2-move we need to find edges $e$ such that the corresponding edges $f_1, f_2$ fulfill the following condition: Either both heads or both tails of $f_1$ and $f_2$ have to be endpoints of $e$. It turns out that we can find such edges $e$ in enough iterations to construct $T'$ with the desired properties (Lemma~\ref{ambivalent-2-move}). This contradicts the $k$-optimality of $T$. Hence a cycle $C$ in $G_2$ of length less than $2k$ cannot exist, and so $G_2$ does have girth at least $2k$.

To conclude, the lower bound on the girth of $G_2$ gives us an upper bound on the number of edges in $G_2$ by known results from extremal graph theory. This implies an upper bound on the number of $l$-long edges as $G_2$ contains at least $\frac{1}{4}$ of the $l$-long edges in $T$. That gives us an upper bound on the length of $T$ and thus also an upper bound on the approximation ratio as we assumed that the optimal tour has length 1.

\subsection{Outline of Upper Bound for Graph TSP}
This subsection comprises a sketch of the proof of Theorem \ref{intro graph upper}. Assume that an instance of \textsc{Graph TSP} $(K_n,c)$ is given where $c$ arises from the unweighted graph $G$. Let a $2$-optimal tour $T$ be given for the instance and fix an orientation.

First, note that every edge with length $l$ corresponds to shortest paths with $l$ edges in $G$ between the endpoints of the edges. 
Now, if the corresponding shortest paths of two edges share a common directed edge, we see that there is an improving 2-move contradicting the assumed 2-optimality of $T$ (Figure \ref{intersectPath}). Hence, the directed edges of the corresponding shortest paths are disjoint. Note that the optimal tour contains $n$ edges and hence has length at least $n$. Thus, if the approximation ratio is high, we must have many edges in the union of the shortest paths corresponding to the edges in $T$ and hence also in $G$. The main challenge now is to exploit this fact in a good way since a simple bound of $n(n-1)$ on the number of directed edges in $G$ would only give an upper bound of $O(n)$ on the approximation ratio, which is worse than the upper bound of $O(\sqrt{n})$ for \textsc{Metric TSP}.

\begin{figure}[!htb]
        \centering
 \definecolor{qqqqff}{rgb}{1,0.,0}
\begin{tikzpicture}[line cap=round,line join=round,>=triangle 45,x=1.0cm,y=1.0cm]
\draw [line width=1.pt,color=qqqqff,->] (5.,2.)-- (5.61607378939762,1.2123115552343946);
\draw [line width=1.pt,->] (5.61607378939762,1.2123115552343946)-- (5.418395726741486,0.2320445591845851);
\draw [line width=1.pt,color=qqqqff,->] (5.418395726741486,0.2320445591845851)-- (5.101228968403289,-0.7163252249214036);
\draw [line width=1.pt,dash pattern=on 3pt off 3pt,color=qqqqff,->] (5.418395726741486,0.2320445591845851)-- (6.233642227337841,-0.3470695429277686);
\draw [line width=1.pt,dash pattern=on 3pt off 3pt,color=qqqqff,->] (6.233642227337841,-0.3470695429277686)-- (6.488667149341046,-1.3140040232092711);
\draw [line width=1.pt,dash pattern=on 3pt off 3pt,color=qqqqff,<-] (5.61607378939762,1.2123115552343946)-- (6.52,1.64);
\draw [shift={(5.294199778298621,0.7671547042692544)},line width=1.pt,dash pattern=on 3pt off 3pt,<-]  plot[domain=-1.3427398904165067:0.9447625422093273,variable=\t]({1.*0.5493336881915749*cos(\t r)+0.*0.5493336881915749*sin(\t r)},{0.*0.5493336881915749*cos(\t r)+1.*0.5493336881915749*sin(\t r)});\begin{scriptsize}
\draw [fill=black] (5.,2.) circle (1.5pt);
\draw [fill=black] (6.52,1.64) circle (1.5pt);
\draw [fill=black] (5.61607378939762,1.2123115552343946) circle (1.5pt);
\draw [fill=black] (5.418395726741486,0.2320445591845851) circle (1.5pt);
\draw [fill=black] (5.101228968403289,-0.7163252249214036) circle (1.5pt);
\draw [fill=black] (6.233642227337841,-0.3470695429277686) circle (1.5pt);
\draw [fill=black] (6.488667149341046,-1.3140040232092711) circle (1.5pt);
\node[above right] at (5.,2.) {$a$};
\node[above right] at (6.52,1.64) {$c$};
\node[left] at (5.101228968403289,-0.7163252249214036) {$b$};
\node[above right] at (6.488667149341046,-1.3140040232092711) {$d$};
\end{scriptsize}
\end{tikzpicture}
  \caption{Assume that a tour $T$ contains the edges $(a,b)$ and $(c,d)$ whose corresponding shortest paths are depicted here as the straight edges for $(a,b)$ and dashed edges for $(c,d)$. If the shortest paths share a common edge, then the 2-move replacing the two edges by $\{a,c\}$ and $\{b,d\}$ would be improving.}
  \label{intersectPath}
\end{figure}

To get a better result we use the same idea from the analysis of the upper bound for \textsc{Metric TSP}: We contract vertices and get a graph with fewer vertices and many edges. Instead of contracting once, we iteratively partition the vertices into sets and contract each set to a single vertex to get a new graph. (We note that we actually just contract the vertices and construct the edges of the new graph in a slightly different way. But let us assume for simplicity that the edges of the new graph are images of the contraction of edges in the old graph.) 
Starting with $G$ in every iteration we ideally want to partition the vertices of the current graph into sets, contract each set to a vertex and delete self-loops such that:
\begin{enumerate}
\item The number of vertices decreases much faster than the number of edges.
\item The subgraphs induced by the sets we contract have small diameter.
\end{enumerate}
The first condition ensures that we get a better bound after every iteration. The second condition builds the connection between the approximation ratio and the number of edges in the contracted graph: It ensures that if the shortest paths corresponding to two edges of $T$ share a directed edge in the contracted graph, then they are also not far away in $G$, so there is an improving 2-move replacing these two edges. This means that a high approximation ratio would imply a high number of edges in the contracted graph.

Unfortunately, it is not easy to ensure both conditions at the same time even if we know that the graph has many edges, as the edges are not equally distributed in the graph. There might be many vertices with very small degree. If we contract them while still ensuring that the subgraphs have small diameter, the number of vertices cannot decrease fast enough. Therefore, we consider a subset of vertices we call \emph{active} vertices and only require that the number of active vertices decreases fast. If an active vertex has small degree, we will not contract it but will consider it as inactive in future iterations. Initially, all vertices are active and we use the following theorem to find a good partition of the active vertices. 

\begin{theorem}[Theorem 6 in \cite{fox2010decompositions}] \label{graph decomposition}
Given $\epsilon>0$ every graph $G$ on $n$ vertices can be edge partitioned $E=E_0 \cup E_1 \cup \dots \cup E_l$ such that $\lvert E_0 \rvert \leq \epsilon n^2, l\leq 16\epsilon^{-1}$ and for $1\leq i\leq l$ the diameter of $E_i$ is at most 4.
\end{theorem}

In every iteration we apply the theorem to the subgraph induced by the currently active vertices. The vertices only incident to edges in $E_0$ become inactive after this iteration. For each of the sets $E_1,\dots, E_l$ we contract the vertices incident to an edge in the set to a single vertex. These are the active vertices in the next iteration. By choosing $\epsilon$ appropriately, we can ensure that the number of vertices decreases significantly and the number of vertices that become inactive in every iteration is small.

After a fixed number of iterations, we have at least one edge and one active vertex remaining. Since the number of active vertices decreased much faster than the edges, we can conclude that $G$ only contains few edges compared to the number of vertices. This implies a bound on the approximation ratio.
\section{Lower Bound for Metric TSP} \label{sec lower bound metric tsp}
In this section, we improve the lower bound of the $k$-Opt algorithm using the following theorem.

\begin{theorem}[Lemma 3.6 in \cite{Chandra}] \label{chandra}
Suppose there exists a Eulerian unweighted graph $G_{k,n,m}$ with $n$ vertices and $m$ edges, having girth at least $2k$. Then, there is a \textsc{Metric TSP} instance with $m$ vertices and a $k$-optimal tour $T$ such that $\frac{c(T)}{c(T^*)}\geq \frac{m}{2n}$, where $T^*$ is the optimal tour of the instance.
\end{theorem}

For the previous lower bound the theorem was applied to regular Eulerian graphs with high girth. Instead, we show that for every graph there is a Eulerian subgraph with similar edge vertex ratio and apply the theorem to the Eulerian subgraphs of dense graphs with high girth to get the new bound.  Before we start, we make the following observation.

\begin{lemma} \label{monoton}
The approximation ratio of the $k$-Opt algorithm for \textsc{Metric TSP} instances with $n$ vertices is monotonically increasing in $n$.
\end{lemma}

\begin{proof}
Given an instance $I$ we can increase the number of vertices of $I$ without decreasing the approximation ratio by constructing an instance $I'$ as follows: Make a copy $v'$ of an arbitrary vertex $v$ and set the costs $c(v,v'):=0$, $c(v',w):=c(v,w) \quad \forall w\neq v$. It is easy to see that $I'$ still satisfies the triangle inequality. To prove that the approximation ratio does not decrease we need to show that the optimal tour of $I$ is at least as long as that of $I'$ and the longest $k$-optimal tour of $I'$ is at least as long as that of $I$. To show this, observe that we can transform a tour of $I$ to a tour of $I'$ by visiting $v'$ directly after visiting $v$ and leaving the order of the other vertices unchanged. The transformed tour has the same cost as the old tour. Given the optimal tour of $I$, the above transformation gives us a tour of $I'$ with the same cost. Thus, the optimal tour of $I$ is at least as long as that of $I'$.

Let $T$ be a $k$-optimal tour of $I$. It remains to show that the transformed tour $T'$ is still $k$-optimal. Assume that there is an improving $k$-move; apply it on $T'$ to get $T'_2$. If the edge $\{v,v'\}$ is contained in $T'_2$, we can contract the vertices $v$ and $v'$ and delete the self-loop at $v$ to get a shorter tour of $I$ than $T$. Observe that this tour arises by performing the same $k$-move to $T$, contradicting the $k$-optimality of $T$. So assume that $\{v,v'\}$ is not contained in $T'_2$. When we contract the vertices $v$ and $v'$ from $T'_2$ we get a connected Eulerian graph $T_2$, where the degree of $v$ is four and the degree of every other vertex is two. Hence, $I$ contains at least two vertices. Now, start at an arbitrary vertex other than $v$ and traverse the graph on a Eulerian walk. Let $\{a_1,v\}, \{v,a_2\}, \{b_1,v\}, \{v,b_2\}$ be the order in which the edges incident to $v$ are traversed (Figure~\ref{tour transformation before fig}). Since there are exactly two edges incident to $v$ in $T$, $T$ contains at most one edge in either $\{a_1,v\}, \{v,a_2\}$ or $\{b_1,v\}, \{v,b_2\}$. W.l.o.g.\ let $T$ contain at most one edge of $\{a_1,v\}$ and $\{v,a_2\}$. We get a tour of $I$ with less or equal length than $T$ by shortcutting $\{a_1,v\}$ and $\{v,a_2\}$ to $\{a_1,a_2\}$ in $T_2$. To obtain this tour from $T$, take the improving $k$-move for $T'$ and replace all occurrences of $v'$ in its edges by $v$. Additionally, introduce two more modifications: First, instead of deleting the self loop $\{v,v\}$ we delete $T\cap\{\{a_1,v\}, \{v,a_2\}\}$, which comprises at most one edge. Second, instead of adding $\{\{a_1,v\}, \{v,a_2\}\} \backslash T$, which comprises at least one edge, we add $\{a_1,a_2\}$ (Figure~\ref{tour transformation after fig}). Hence, this tour arises from $T$ by performing a $k$-move, again contradicting the $k$-optimality of $T$. Therefore, the longest $k$-optimal tour of $I'$ is at least as long as that of $I$.
\end{proof}

\begin{figure}[!htb]
  \centering
  \begin{minipage}[t]{.45\textwidth}
      \centering
\begin{tikzpicture}
  \node[draw, circle, fill=black, inner sep=1.5pt] (a) at (0,0) {};
  \node[draw, circle, fill=black, inner sep=1.5pt] (b) at (0,2) {};
  \node[draw, circle, fill=black, inner sep=1.5pt] (c) at (2,0) {};
  \node[draw, circle, fill=black, inner sep=1.5pt] (d) at (2,2) {};
  \node[draw, circle, fill=black, inner sep=1.5pt] (v) at (1,1) {};

  \draw (a) -- (v) -- (b);
  \draw (c) -- (v) -- (d);

  \draw (b) to [out=180, in=180] (a);
  \draw (c) to [out=0, in=0] (d);
  
  \foreach \label/\name in {a/b_2, b/a_1, c/b_1, d/a_2, v/v}
    \node[above] at (\label) {\(\name\)};
\end{tikzpicture}
\caption{Apply an improving $k$-move to the tour $T'$ of the instance $I'$. After contracting the vertices $v$ and $v'$, the result $T_2$ may visit $v$ multiple times.}
\label{tour transformation before fig}
  \end{minipage}\qquad%
  \begin{minipage}[t]{0.45\textwidth}
      \centering
      \begin{tikzpicture}
    \node[draw, circle, fill=black, inner sep=1.5pt] (a) at (0,0) {};
    \node[draw, circle, fill=black, inner sep=1.5pt] (b) at (0,2) {};
    \node[draw, circle, fill=black, inner sep=1.5pt] (c) at (2,0) {};
    \node[draw, circle, fill=black, inner sep=1.5pt] (d) at (2,2) {};
    \node[draw, circle, fill=black, inner sep=1.5pt] (v) at (1,1) {};
      
    \draw (a) -- (v) -- (c);
    \draw (b) -- (d);
  
    \draw (b) to [out=180, in=180] (a);
    \draw (c) to [out=0, in=0] (d);
    
    \foreach \label/\name in {a/b_2, b/a_1, c/b_1, d/a_2, v/v}
      \node[above] at (\label) {\(\name\)};
      \end{tikzpicture}
      \caption{Shortcut $\{a_1,v\}$ and $\{v,a_2\}$ to $\{a_1,a_2\}$ in $T_2$ to get a tour. To directly construct this tour from $T$ make the following modifications to the $k$-move: Replace all occurrences of $v'$ by $v$. Instead of deleting $\{v,v\}$ we delete $T\cap\{\{a_1,v\}, \{v,a_2\}\}$. Instead of adding $\{\{a_1,v\}, \{v,a_2\}\} \backslash T$ we add $\{a_1,a_2\}$.}
      \label{tour transformation after fig}
  \end{minipage}
\end{figure}

\begin{lemma} \label{subgraph}
For every graph $G$ there exists a Eulerian subgraph $G'$ such that $\frac{\lvert E(G') \rvert}{\lvert V(G') \rvert }\geq \frac{\lvert E(G)\rvert +1}{\lvert V(G) \rvert}-1$.
\end{lemma}

\begin{proof}
We construct a new graph by deleting cycles successively from $G$ and adding them to an empty graph $G_0$ with $V(G_0)=V(G)$ until there are no cycles left. After the deletion of cycles, the remaining graph will be a forest with at most $\lvert V(G) \rvert-1$ edges. Hence, we added at least $\lvert E(G) \rvert -\lvert V(G) \rvert+1$ edges to $G_0$. There is a connected component $G'$ of $G_0$ whose edge-vertex ratio is at least as large as that of $G_0$, which is $\frac{\lvert E(G) \rvert-\lvert V(G) \rvert+1}{\lvert V(G) \rvert}$. By construction, $G'$ is Eulerian.
\end{proof}

\begin{theorem} \label{lower bound}
The approximation ratio of $k$-Opt is $\Omega\left(\frac{n}{\ex^{-1}(n,2k)}\right)$ for \textsc{Metric TSP} where $n$ is the number of vertices.
\end{theorem}

\begin{proof}
Take a graph $G$ with girth $2k$, $ex^{-1}(n,2k)$ vertices and $n$ edges. By Lemma \ref{subgraph}, there is a Eulerian subgraph $G'$ with $\frac{\lvert E(G') \rvert}{\lvert V(G') \rvert}\geq \frac{n +1}{\ex^{-1}(n,2k)}-1$. Clearly, this subgraph has girth at least $2k$.
By Theorem \ref{chandra}, we can construct an instance with $\lvert E(G') \rvert\leq n$ vertices and an approximation ratio of $\Omega\left(\frac{n +1}{\ex^{-1}(n,2k)}-1\right)=\Omega\left(\frac{n}{\ex^{-1}(n,2k)}\right)$ since by Theorem \ref{exlower} $\lim_{n\to \infty}\frac{n}{\ex^{-1}(n,2k)} =\infty$. The statement follows from the fact that the approximation ratio is monotonically increasing by Lemma \ref{monoton}. 
\end{proof}

\begin{theorem}
If $\ex(n,2k) = \Omega(n^c)$ for some $c>0$, then the approximation ratio of $k$-Opt is $\Omega(n^{1-\frac{1}{c}})$ for \textsc{Metric TSP} where $n$ is the number of vertices.
\end{theorem}

\begin{proof}
If $\ex(n,2k) = \Omega(n^c)$, then $\ex^{-1}(n,2k) = O(n^\frac{1}{c})$ and by Theorem \ref{lower bound} we can construct an instance with approximation ratio $\Omega\left(\frac{n}{n^\frac{1}{c}}\right)=\Omega(n^{1-\frac{1}{c}})$.
\end{proof}

Together with Theorem \ref{exlower} and \ref{omega small cases}, we conclude:

\begin{corollary} \label{coro lower}
For \textsc{Metric TSP} the approximation ratio of $k$-Opt is $\Omega(\sqrt[k]{n})$ for $k=3,4,6$ and $\Omega(n^{\frac{2}{3k-4+o}})$ for all other $k$ where $o=0$ if $k$ is even and $o=1$ if $k$ is odd and $n$ is the number of vertices.
\end{corollary}

\section{Upper Bound for Metric TSP} \label{sec upper bound for metric tsp}
In this section we give an upper bound on the approximation ratio of the $k$-Opt and $k$-Lin-Kernighan algorithm. We bound the length of any $k$-optimal or $k$-Lin-Kernighan optimal tour $T$ compared to the optimal tour. To show the bound we divide in Subsection~\ref{subsec construction aux} the edges of $T$ into classes such that the lengths of two edges in the same class differ by at most a constant factor. For each of these classes we construct with the help of the optimal tour a graph containing at least $\frac{1}{4}$ of the edges in the class. In Subsection \ref{subsec bound on girth} we show that this graph has a high girth. Thus, we can use results from extremal graph theory to bound the number of edges in the length class in Subsection~\ref{subsec bound on approx} culminating in a bound on the length of $T$ and the approximation ratio. 

\subsection{Construction of the Auxiliary Graph $G_2$} \label{subsec construction aux}

Fix a $k>2$ and assume that a worst-case instance with $n$ vertices is given. Let $T$ be a $k$-optimal or $k$-Lin-Kernighan optimal tour of this instance. We fix an orientation of the optimal tour and $T$. Moreover, let w.l.o.g.\ the length of the optimal tour be 1. We divide the edges of $T$ into length classes.

\begin{definition} \label{llong}
An edge $e$ is \emph{$l$-long} if $(\frac{4k-5}{4k-4})^{l+1}< c(e)\leq (\frac{4k-5}{4k-4})^{l}$. Let $\{q_l\}_{l \in \N_0}$ be the sequence of the number of $l$-long edges in $T$.
\end{definition}

Note that the shortest path between every pair of vertices has length at most $\frac{1}{2}$ since the optimal tour has length 1. Thus, by the triangle inequality every edge with positive length in $T$ has length at most $\frac{1}{2}$ and is $l$-long for exactly one $l$. For every $l$ we want to bound the number of $l$-long edges. Let us consider from now on a fixed $l$. In the following we define three auxiliary graphs we need for the analysis and show some useful properties of them. Our general aim is to show that the girth of an auxiliary graph containing many $l$-long edges is high since otherwise there would exist an improving $k$-move contradicting the assumption. This would imply a bound on the number of $l$-long edges depending on the number of vertices.

\begin{definition} \label{defnear}
We view the optimal tour as a circle with circumference 1. Let the vertices of the instance lie on that circle in the order of the oriented tour where the arc distance of two consecutive vertices is the length of the edge between them. Partition the optimal tour circle into $4(k-1)\lceil (\frac{4k-4}{4k-5})^{l} \rceil$ consecutive arcs of length $\frac{1}{4(k-1)\lceil (\frac{4k-4}{4k-5})^{l} \rceil}$. Two vertices are called \emph{near} to each other (in respect of the optimal tour) if they lie on the same arc.
\end{definition}

\begin{definition} \label{defG1}
Let the directed graph $G:=(V(K_n),T)$ consist of the vertices of the instance and the oriented edges of $T$ (an example is shown in Figure \ref{tour}; the colors of the edges will be explained later). The directed multigraph $G_1$ arises from $G$ by contracting all vertices near each other to a vertex and deleting self-loops (Figure \ref{g1}). 
\end{definition}

Note that $G_1$ may contain parallel edges. By construction, $G_1$ contains fewer vertices than $G$ and we will later show that the definition of near ensures that $G_1$ contains all the $l$-long edges. Thus, a lower bound on the girth of $G_1$ would give a better upper bound on the number of $l$-long edges than a lower bound on the girth of $G$. Unfortunately, we can not bound the girth of $G_1$ since the existence of a short cycle in $G_1$ would not necessarily lead to an improving $k$-move. For that we need the property that the common vertex of consecutive edges of a cycle in the graph is the head of both or the tail of both edges according to the orientation of $T$. To ensure this, in the next step we further modify $G_1$ to the graph $G_2$.

\begin{lemma} \label{coloring}
There exists a coloring of the vertices of $G_1$ with two colors such that at least $\frac{1}{4}$ of the $l$-long edges in $G_1$ go from a red vertex to a blue vertex according to the fixed orientation of $T$.
\end{lemma}

\begin{proof}
The proof uses the standard probabilistic method developed by Erd\H{o}s (see for example \cite{DBLP:books/wi/AlonS92}).

Color each vertex independently red or blue with equal probability. Each $l$-long edge goes from a red vertex to a blue vertex with probability $\frac{1}{4}$. Hence, the expected number of $l$-long edges satisfying this condition is $\frac{1}{4}$ of the original number. This implies that there is a coloring where at least $\frac{1}{4}$ of the $l$-long edges satisfy the condition.
\end{proof}

\begin{definition} \label{defG2}
We obtain the directed multigraph \emph{$G_2$} by coloring the vertices of $G_1$ red and blue according to Lemma \ref{coloring} and deleting all edges that are not $l$-long edges from a red vertex to a blue vertex according to the fixed orientation of $T$ (Figure \ref{g2}. The colors of the edges will be explained later).
\end{definition}

In the next subsection, we will prove the following central property of $G_2$.

\begin{claim} \label{claim underlying g2 girth}
The underlying undirected graph of $G_2$ has girth at least $2k$.
\end{claim}

Using the claim and that the graph $G_2$ contains at least $\frac{1}{4}$ of the $l$-long edges for the $l$ we have fixed, we get an upper bound on the number of $l$-long edges. 

\begin{corollary} \label{ql bound}
We have $q_l\leq 4\ex(4(k-1)\lceil (\frac{4k-4}{4k-5})^{l} \rceil,2k)$ where $q_l$ is the number of $l$-long edges in $T$.
\end{corollary}

\begin{proof}
By definition, $G$ contains $q_l$ $l$-long edges. By the triangle inequality, any two vertices which are near each other have distance at most $\frac{1}{4(k-1)\lceil (\frac{4k-4}{4k-5})^{l} \rceil}\leq \frac{1}{4(k-1) (\frac{4k-4}{4k-5})^{l}} = \frac{(\frac{4k-5}{4k-4})^{l}}{4(k-1)}$ which is shorter than the length of any $l$-long edge. Hence, $G_1$ has also $q_l$ $l$-long edges. Since we have chosen a coloring according to Lemma \ref{coloring}, $G_2$ has at least $\frac{1}{4}q_l$ edges. By Claim~\ref{claim underlying g2 girth}, $G_2$ has girth at least $2k$ and thus at most $\ex(\lvert V(G_2) \rvert,2k)\leq\ex(4(k-1)\lceil (\frac{4k-4}{4k-5})^{l} \rceil,2k)$ edges. Therefore, $q_l\leq 4\ex(4(k-1)\lceil (\frac{4k-4}{4k-5})^{l} \rceil,2k)$.
\end{proof}

\subsection{The Girth of the Graph $G_2$} \label{subsec bound on girth}
In this subsection we show Claim \ref{claim underlying g2 girth}, i.e.\ the underlying undirected graph of $G_2$ has girth at least $2k$. In particular, it has no parallel edges. Assume the contrary. Then there has to be a cycle $C$ with $2h<2k$ edges since $G_2$ is bipartite by construction. Our strategy is to construct an improving alternating cycle for $T$ with help of $C$. That contradicts $k$-optimality or $k$-Lin-Kernighan optimality of $T$ and thus such a cycle $C$ cannot exist.

We call the preimages in $G$ of the edges in $C$ the \emph{$C$-edges}. Note that the preimages are unique since we do not delete parallel edges after the contraction.

\begin{definition}
Let the \emph{connecting paths} be the connected components of the graph $(V(K_n),T\setminus C\text{-edges})$, i.e.\ the paths in $T$ between consecutive heads and tails of $C$-edges (the colored edges in Figure \ref{tour} and \ref{g2}). Define the \emph{head} and the \emph{tail} of a path $p$ as the head of the last edge and the tail of the first edge of $p$ according to the orientation of $T$, respectively. The head and tail of a connecting path are also called the \emph{endpoints} of the connecting path. 
\end{definition}

Note that the number of connecting paths is equal to that of $C$-edges which is $2h$.

\begin{lemma} \label{connecting path}
The two endpoints of a connecting path are not near each other. In particular, every connecting path contains at least one edge.
\end{lemma}

\begin{proof}
Observe that the head and tail of a connecting path is a tail and head of a $C$-edge, respectively. Hence, the corresponding vertices of the heads and tails of the connecting paths in $G_2$ are colored red and blue, respectively. Therefore, the two endpoints are not near each other. Since the relation near is reflexive, we can conclude that every connecting path contains at least one edge.
\end{proof}

\begin{figure}[h]
    \centering
    \begin{minipage}[t]{.45\textwidth}
        \centering
 \definecolor{ffqqqq}{rgb}{1,0,0}
\definecolor{qqqqff}{rgb}{0,0,1}
\begin{tikzpicture}[line cap=round,line join=round,>=triangle 45,x=0.65cm,y=0.65cm]
\draw [->,shift={(-6,1)},line width=2pt,color=ffqqqq]  plot[domain=3.141592653589793:6.283185307179586,variable=\t]({1*2*cos(\t r)+0*2*sin(\t r)},{0*2*cos(\t r)+1*2*sin(\t r)});
\draw [->,line width=2pt,color=qqqqff] (-4,1)-- (-2,1);
\draw [<-,line width=2pt,color=qqqqff] (-8,1)-- (-4,2);
\draw [<-,line width=2pt,color=qqqqff] (-2,2)-- (2,1);
\draw [<-,shift={(-6,2)},line width=2pt,color=ffqqqq]  plot[domain=1.57:3.141592653589793,variable=\t]({1*2*cos(\t r)+0*2*sin(\t r)},{0*2*cos(\t r)+1*2*sin(\t r)});
\draw [<-,shift={(-6,2)},line width=2pt,color=ffqqqq]  plot[domain=1.0430772443805858:1.57,variable=\t]({1*2*cos(\t r)+0*2*sin(\t r)},{0*2*cos(\t r)+1*2*sin(\t r)});
\draw [->,shift={(-3,2.25)},line width=2pt,color=qqqqff]  plot[domain=-0.049958395721942495:3.1915510493117356,variable=\t]({1*5.006246098625197*cos(\t r)+0*5.006246098625197*sin(\t r)},{0*5.006246098625197*cos(\t r)+1*5.006246098625197*sin(\t r)});
\draw [->,shift={(-6,1)},line width=2pt,color=ffqqqq]  plot[domain=3.141592653589793:4.1832244161110745,variable=\t]({1*2*cos(\t r)+0*2*sin(\t r)},{0*2*cos(\t r)+1*2*sin(\t r)});
\draw [->,shift={(-6,1)},line width=2pt,color=ffqqqq]  plot[domain=4.1832244161110745:5.243619098560426,variable=\t]({1*2*cos(\t r)+0*2*sin(\t r)},{0*2*cos(\t r)+1*2*sin(\t r)});
\draw [<-,shift={(0,2)},line width=2pt,color=ffqqqq]  plot[domain=0:1.069050569333379,variable=\t]({1*2*cos(\t r)+0*2*sin(\t r)},{0*2*cos(\t r)+1*2*sin(\t r)});
\draw [<-,shift={(0,2)},line width=2pt,color=ffqqqq]  plot[domain=1.069050569333379:2.1120654445330964,variable=\t]({1*2*cos(\t r)+0*2*sin(\t r)},{0*2*cos(\t r)+1*2*sin(\t r)});
\draw [->,shift={(0,1)},line width=2pt,color=ffqqqq]  plot[domain=3.141592653589793:4.19198421356179,variable=\t]({1*2*cos(\t r)+0*2*sin(\t r)},{0*2*cos(\t r)+1*2*sin(\t r)});
\draw [->,shift={(0,1)},line width=2pt,color=ffqqqq]  plot[domain=4.19198421356179:5.210114300972504,variable=\t]({1*2*cos(\t r)+0*2*sin(\t r)},{0*2*cos(\t r)+1*2*sin(\t r)});
\draw [<-,shift={(-6,2)},line width=2pt,color=ffqqqq]  plot[domain=0:1.0430772443805858,variable=\t]({1*2*cos(\t r)+0*2*sin(\t r)},{0*2*cos(\t r)+1*2*sin(\t r)});
\draw [<-,shift={(0,2)},line width=2pt,color=ffqqqq]  plot[domain=2.1120654445330964:3.141592653589793,variable=\t]({1*2*cos(\t r)+0*2*sin(\t r)},{0*2*cos(\t r)+1*2*sin(\t r)});
\draw [->,shift={(0,1)},line width=2pt,color=ffqqqq]  plot[domain=-1.0730710062070825:0,variable=\t]({1*2*cos(\t r)+0*2*sin(\t r)},{0*2*cos(\t r)+1*2*sin(\t r)});
\begin{scriptsize}
\draw [fill=black] (-8,2) circle (2.5pt);
\draw [fill=black] (-4,2) circle (2.5pt);
\draw [fill=black] (-4,1) circle (2.5pt);
\draw [fill=black] (-8,1) circle (2.5pt);
\draw [fill=black] (-2,2) circle (2.5pt);
\draw [fill=black] (2,2) circle (2.5pt);
\draw [fill=black] (2,1) circle (2.5pt);
\draw [fill=black] (-2,1) circle (2.5pt);
\draw [fill=black] (-4.992871927896868,3.727915809980865) circle (2.5pt);
\draw [fill=black] (-1.0304482212271813,3.7141109833875223) circle (2.5pt);
\draw [fill=black] (0.9619137071330164,3.7534885286278943) circle (2.5pt);
\draw [fill=black] (-7.009624690035605,-0.7264582199608851) circle (2.5pt);
\draw [fill=black] (-4.986811373774729,-0.7243691042487825) circle (2.5pt);
\draw [fill=black] (-0.9944627230857661,-0.7352359759965914) circle (2.5pt);
\draw [fill=black] (0.95485616205282,-0.7573416599482752) circle (2.5pt);
\draw [fill=black] (-6,4) circle (2.5pt);
\end{scriptsize}
\end{tikzpicture}
  \caption{An example instance with a $k$-optimal tour, i.e.\ the directed graph $G$. The blue and red edges are the $C$-edges and connecting path edges that arise from the chosen cycle in $G_2$ in Figure \ref{g2}, respectively. Note that the optimal tour is not drawn here, so it is not clear from the figure which vertices to contract to construct $G_1$.}
  \label{tour}
    \end{minipage}\qquad%
    \begin{minipage}[t]{0.45\textwidth}
        \centering
 \definecolor{ffqqqq}{rgb}{0,0,0}
\definecolor{qqqqff}{rgb}{0,0,0}
\begin{tikzpicture}[line cap=round,line join=round,>=triangle 45,x=0.65cm,y=0.65cm]
\draw [->,line width=2pt,color=qqqqff] (-3,1)-- (-1,1);
\draw [->,line width=2pt,color=qqqqff] (3,1)-- (-1,1);
\draw [->,line width=2pt,color=qqqqff] (-3,1)-- (-7,1);
\draw [->,shift={(-2,1)},line width=2pt,color=qqqqff]  plot[domain=0:3.141592653589793,variable=\t]({1*5*cos(\t r)+0*5*sin(\t r)},{0*5*cos(\t r)+1*5*sin(\t r)});
\draw [<-,shift={(-5,1)},line width=2pt,color=ffqqqq]  plot[domain=0:1.053436988490858,variable=\t]({1*2*cos(\t r)+0*2*sin(\t r)},{0*2*cos(\t r)+1*2*sin(\t r)});
\draw [<-,shift={(-5,1)},line width=2pt,color=ffqqqq]  plot[domain=1.053436988490858:1.52,variable=\t]({1*2*cos(\t r)+0*2*sin(\t r)},{0*2*cos(\t r)+1*2*sin(\t r)});
\draw [<-,shift={(-5,1)},line width=2pt,color=ffqqqq]  plot[domain=1.52:3.141592653589793,variable=\t]({1*2*cos(\t r)+0*2*sin(\t r)},{0*2*cos(\t r)+1*2*sin(\t r)});
\draw [->,shift={(-5,1)},line width=2pt,color=ffqqqq]  plot[domain=3.141592653589793:4.187593208734581,variable=\t]({1*2*cos(\t r)+0*2*sin(\t r)},{0*2*cos(\t r)+1*2*sin(\t r)});
\draw [->,shift={(-5,1)},line width=2pt,color=ffqqqq]  plot[domain=4.187593208734581:5.245900127691526,variable=\t]({1*2*cos(\t r)+0*2*sin(\t r)},{0*2*cos(\t r)+1*2*sin(\t r)});
\draw [->,shift={(-5,1)},line width=2pt,color=ffqqqq]  plot[domain=-1.0372851794880606:0,variable=\t]({1*2*cos(\t r)+0*2*sin(\t r)},{0*2*cos(\t r)+1*2*sin(\t r)});
\draw [<-,shift={(1,1)},line width=2pt,color=ffqqqq]  plot[domain=2.0807683864963096:3.141592653589793,variable=\t]({1*2*cos(\t r)+0*2*sin(\t r)},{0*2*cos(\t r)+1*2*sin(\t r)});
\draw [->,shift={(1,1)},line width=2pt,color=ffqqqq]  plot[domain=3.141592653589793:4.202416920683277,variable=\t]({1*2*cos(\t r)+0*2*sin(\t r)},{0*2*cos(\t r)+1*2*sin(\t r)});
\draw [->,shift={(1,1)},line width=2pt,color=ffqqqq]  plot[domain=4.202416920683277:5.249401338813848,variable=\t]({1*2*cos(\t r)+0*2*sin(\t r)},{0*2*cos(\t r)+1*2*sin(\t r)});
\draw [->,shift={(1,1)},line width=2pt,color=ffqqqq]  plot[domain=-1.0337839683657384:0,variable=\t]({1*2*cos(\t r)+0*2*sin(\t r)},{0*2*cos(\t r)+1*2*sin(\t r)});
\draw [<-,shift={(1,1)},line width=2pt,color=ffqqqq]  plot[domain=0:1.0423625539140053,variable=\t]({1*2*cos(\t r)+0*2*sin(\t r)},{0*2*cos(\t r)+1*2*sin(\t r)});
\draw [<-,shift={(1,1)},line width=2pt,color=ffqqqq]  plot[domain=1.0423625539140056:2.0807683864963096,variable=\t]({1*2*cos(\t r)+0*2*sin(\t r)},{0*2*cos(\t r)+1*2*sin(\t r)});
\begin{scriptsize}
\draw [fill=black] (-7,1) circle (2.5pt);
\draw [fill=black] (-3,1) circle (2.5pt);
\draw [fill=black] (-1,1) circle (2.5pt);
\draw [fill=black] (3,1) circle (2.5pt);
\draw [fill=black] (-4.0108264175095565,2.73825648961913) circle (2.5pt);
\draw [fill=black] (-6.002072541083411,-0.7308525709622511) circle (2.5pt);
\draw [fill=black] (-3.9828806768827247,-0.7220535074559835) circle (2.5pt);
\draw [fill=black] (0.02369427609954644,-0.7455162942462643) circle (2.5pt);
\draw [fill=black] (2.0231423495308087,-0.7184818103769901) circle (2.5pt);
\draw [fill=black] (0.023694276099546885,2.7455162942462645) circle (2.5pt);
\draw [fill=black] (2.0083627397427817,2.727195583915855) circle (2.5pt);
\draw [fill=black] (-5,3) circle (2.5pt);
\end{scriptsize}
\end{tikzpicture}
  \caption{The directed multigraph $G_1$: We contracted vertices that lie near each other in the optimal tour. We can see there are 4 pairs of vertices in Figure~\ref{tour} that were contracted.}
  \label{g1}
    \end{minipage}
\end{figure}

\begin{figure}[h]
    \centering
    \begin{minipage}[t]{.45\textwidth}
        \centering
 \definecolor{ffqqqq}{rgb}{1,0,0}
\definecolor{qqqqff}{rgb}{0,0,1}
\begin{tikzpicture}[line cap=round,line join=round,>=triangle 45,x=0.65cm,y=0.65cm]
\draw [->,line width=2pt,color=qqqqff] (-3,1)-- (-1,1);
\draw [->,line width=2pt,color=qqqqff] (3,1)-- (-1,1);
\draw [->,line width=2pt,color=qqqqff] (-3,1)-- (-7,1);
\draw [->,shift={(-2,1)},line width=2pt,color=qqqqff]  plot[domain=0:3.141592653589793,variable=\t]({1*5*cos(\t r)+0*5*sin(\t r)},{0*5*cos(\t r)+1*5*sin(\t r)});
\draw [->,shift={(-5,1)},line width=2pt,color=ffqqqq]  plot[domain=4.187593208734581:5.245900127691526,variable=\t]({1*2*cos(\t r)+0*2*sin(\t r)},{0*2*cos(\t r)+1*2*sin(\t r)});
\draw [->,shift={(1,1)},line width=2pt,color=ffqqqq]  plot[domain=4.202416920683277:5.249401338813848,variable=\t]({1*2*cos(\t r)+0*2*sin(\t r)},{0*2*cos(\t r)+1*2*sin(\t r)});
\draw [<-,shift={(1,1)},line width=2pt,color=ffqqqq]  plot[domain=1.0423625539140056:2.0807683864963096,variable=\t]({1*2*cos(\t r)+0*2*sin(\t r)},{0*2*cos(\t r)+1*2*sin(\t r)});
\begin{scriptsize}
\draw [fill=qqqqff] (-7,1) circle (2.5pt);
\draw [fill=ffqqqq] (-3,1) circle (2.5pt);
\draw [fill=qqqqff] (-1,1) circle (2.5pt);
\draw [fill=ffqqqq] (3,1) circle (2.5pt);
\draw [fill=qqqqff] (-4.0108264175095565,2.73825648961913) circle (2.5pt);
\draw [fill=ffqqqq] (-6.002072541083411,-0.7308525709622511) circle (2.5pt);
\draw [fill=qqqqff] (-3.9828806768827247,-0.7220535074559835) circle (2.5pt);
\draw [fill=ffqqqq] (0.02369427609954644,-0.7455162942462643) circle (2.5pt);
\draw [fill=qqqqff] (2.0231423495308087,-0.7184818103769901) circle (2.5pt);
\draw [fill=ffqqqq] (0.023694276099546885,2.7455162942462645) circle (2.5pt);
\draw [fill=qqqqff] (2.0083627397427817,2.727195583915855) circle (2.5pt);
\draw [fill=ffqqqq] (-5,3) circle (2.5pt);
\end{scriptsize}
\end{tikzpicture}
  \caption{The directed multigraph $G_2$: Coloring the vertices and only considering the $l$-long edges from red to blue. In this example the upper left edge is not $l$-long and hence not drawn. The blue edges form the undirected cycle $C$. The red edges are the remaining edges of the connecting paths corresponding to this cycle.}
  \label{g2}
    \end{minipage}\qquad%
    \begin{minipage}[t]{0.45\textwidth}
        \centering
 \definecolor{qqffqq}{rgb}{0,1,0}
\definecolor{ffqqqq}{rgb}{1,0,0}
\begin{tikzpicture}[line cap=round,line join=round,>=triangle 45,x=1.43cm,y=1.43cm]
\draw [shift={(-5,2)},line width=2pt,color=ffqqqq]  plot[domain=0:3.141592653589793,variable=\t]({1*1*cos(\t r)+0*1*sin(\t r)},{0*1*cos(\t r)+1*1*sin(\t r)});
\draw [shift={(-5,1.5)},line width=2pt,color=ffqqqq]  plot[domain=3.141592653589793:6.283185307179586,variable=\t]({1*1*cos(\t r)+0*1*sin(\t r)},{0*1*cos(\t r)+1*1*sin(\t r)});
\draw [shift={(-2,2)},line width=2pt,color=ffqqqq]  plot[domain=0:3.141592653589793,variable=\t]({1*1*cos(\t r)+0*1*sin(\t r)},{0*1*cos(\t r)+1*1*sin(\t r)});
\draw [shift={(-2,1.5)},line width=2pt,color=ffqqqq]  plot[domain=3.141592653589793:6.283185307179586,variable=\t]({1*1*cos(\t r)+0*1*sin(\t r)},{0*1*cos(\t r)+1*1*sin(\t r)});
\draw [line width=2pt,color=qqffqq] (-6,2)-- (-6,1.5);
\draw [line width=2pt,color=qqffqq] (-4,2)-- (-4,1.5);
\draw [line width=2pt,color=qqffqq] (-3,2)-- (-3,1.5);
\draw [line width=2pt,color=qqffqq] (-1,2)-- (-1,1.5);
\begin{scriptsize}
\draw [fill=black] (-6,2) circle (2.5pt);
\draw [fill=black] (-4,2) circle (2.5pt);
\draw [fill=black] (-4,1.5) circle (2.5pt);
\draw [fill=black] (-6,1.5) circle (2.5pt);
\draw [fill=black] (-3,2) circle (2.5pt);
\draw [fill=black] (-1,2) circle (2.5pt);
\draw [fill=black] (-1,1.5) circle (2.5pt);
\draw [fill=black] (-3,1.5) circle (2.5pt);
\draw [fill=black] (-4.4986326803889485,2.865234540934439) circle (2.5pt);
\draw [fill=black] (-5.511366354611534,0.6406371829248653) circle (2.5pt);
\draw [fill=black] (-4.492859135063964,0.6381368187978275) circle (2.5pt);
\draw [fill=black] (-2.498293794761824,0.6329917566125098) circle (2.5pt);
\draw [fill=black] (-1.506912539986257,0.6300202549615334) circle (2.5pt);
\draw [fill=black] (-1.5104150267938854,2.871955591765296) circle (2.5pt);
\draw [fill=black] (-2.5082509895346465,2.8612089941686936) circle (2.5pt);
\draw [fill=black] (-5,3) circle (2.5pt);
\end{scriptsize}
\end{tikzpicture}
  \caption{The graph $G_3^C$: The green edges are the short edges, the red edges are the connecting paths.} 
  \label{g3}
    \end{minipage}
\end{figure}

\begin{definition} \label{defShortEdge}
For any two endpoints $v_1,v_2$ of $C$-edges in $G$ which are near each other we call the edge $\{v_1,v_2\}$ a \emph{short edge}.
\end{definition}

\begin{lemma}\label{short edges}
There are exactly $2h$ short edges forming an alternating cycle with the $C$-edges. Moreover, every short edge connects either two heads or two tails of connecting paths.
\end{lemma}

\begin{proof}
For any endpoint of a $C$-edge in $G$ there is exactly one other endpoint of a $C$-edge which is near to it since the $C$-edges in $G$ are the preimage of a cycle in $G_2$. By definition, every near pair of such endpoints is connected by a short edge and no other short edges exist. Note that there are $2h$ $C$-edges, so we get $2h$ short edges which form a set of alternating cycles with the $C$-edges.  Again using the fact that after the contraction we get a single cycle $C$ in $G_2$, we see that the $C$-edges form with the short edges a single alternating cycle.
Since the vertices of $C$ are colored either red or blue in $G_2$, the short edges connect two heads or two tails of $C$-edges and hence also two heads or two tails of connecting paths.
\end{proof}

\begin{definition}
We construct the graph $G_3^C$ as follows: The vertex set of $G_3^C$ is that of $G$ and the edge set consists of the connecting paths and the short edges (Figure \ref{g3}).
\end{definition}

\begin{lemma} \label{connected comp}
$E(G_3^C)$ is the union of at most $h$ disjoint cycles.
\end{lemma}

\begin{proof}
By the definition of connecting path, every endpoint of a connecting path is an endpoint of a $C$-edge and vice versa. By Lemma \ref{short edges}, every endpoint of a $C$-edge is an endpoint of a short edge and vice versa. Hence, every vertex in $G$ is either an endpoint of a connecting path and a short edge or none of them. Thus, the edges of $G_3^C$ form disjoint cycles. 
Note that every connected component in $G_3^C$ contains at least two connecting paths since the two endpoints of a connecting path are not near each other by Lemma~\ref{connecting path} and hence they cannot be connected by a short edge. Thus, there are at most $h$ connected components.
\end{proof}

Before we start with the actual analysis we show that the total length of all short edges is smaller than that of any $C$-edge. The first step is to bound the total length of all short edges.

\begin{lemma} \label{SetOfShortEdges}
Let $S$ be the set of short edges. We have
\begin{align*}
\sum_{e\in S} c(e)\leq \frac{1}{2}\left(\frac{4k-5}{4k-4}\right)^l.
\end{align*}
\end{lemma}

\begin{proof}
By Lemma \ref{short edges}, there are $2h\leq 2(k-1)$ short edges. Each of them connects two vertices which are near each other. By the triangle inequality, each of the short edges has length at most $\frac{1}{4(k-1)\lceil (\frac{4k-4}{4k-5})^{l} \rceil}$. Hence the total length of short edges is at most $2h\frac{1}{4(k-1)\lceil (\frac{4k-4}{4k-5})^{l} \rceil}\leq 2(k-1)\frac{1}{4(k-1) (\frac{4k-4}{4k-5})^{l}}=\frac{1}{2}(\frac{4k-5}{4k-4})^l$.
\end{proof}

In the subsequent step, we leverage the fact that the ratio between the upper and lower bounds of the length of $l$-long edges is by definition reasonably small. By the choice of the ratio the combined length of $x$ ``long'' $l$-long edges, along with all the short edges, remains shorter than that of $x+1$ ``short'' $l$-long edges.

\begin{lemma} \label{number of blue edges}
Let $S$ be the set of short edges. Let $B_1$ and $B_2$ be sets of $C$-edges with $\lvert B_1 \rvert < \lvert B_2 \rvert\leq 2h$. Then
\begin{align*}
\sum_{e\in B_1}c(e)+\sum_{e\in S}c(e)\leq \sum_{e\in B_2}c(e).
\end{align*}
\end{lemma}

\begin{proof}
Let $\beta_1:=\lvert B_1 \rvert, \beta_2:=\lvert B_2 \rvert$. We have
\begin{align*}
\frac{\beta_1+\frac{1}{2}}{\beta_1+1}=\frac{2\beta_1+1}{2\beta_1+2}\leq \frac{2\cdot (2h-1)+1}{2\cdot (2h-1)+2}\leq \frac{4h-1}{4h}\leq \frac{4(k-1)-1}{4(k-1)}=\frac{4k-5}{4k-4}. \end{align*}
Combined with Lemma \ref{SetOfShortEdges} and the fact that $C$-edges are edges of $G_2$ and hence $l$-long we get
\begin{align*}
\sum_{e\in B_1}c(e)+\sum_{e\in S}c(e)&\leq \beta_1\left(\frac{4k-5}{4k-4}\right)^l+\frac{1}{2}\left(\frac{4k-5}{4k-4}\right)^l\\
&= \left(\beta_1+\frac{1}{2}\right)\left(\frac{4k-5}{4k-4}\right)^l \leq (\beta_1+1)\left(\frac{4k-5}{4k-4}\right)^{l+1}\\
&\leq \beta_2\left(\frac{4k-5}{4k-4}\right)^{l+1} < \sum_{e\in B_2}c(e). \qedhere
\end{align*}
\end{proof}

\begin{figure}[h]
  \begin{minipage}[t]{.45\textwidth}
      \centering
         \definecolor{xfqqff}{rgb}{1,1,0}
\definecolor{qqqqff}{rgb}{0,0,1}
\definecolor{qqffqq}{rgb}{0,1,0}
\definecolor{ffqqqq}{rgb}{1,0,0}
\begin{tikzpicture}[line cap=round,line join=round,>=triangle 45,x=0.7cm,y=0.7cm]
\draw [shift={(4,3)},line width=2pt,color=ffqqqq]  plot[domain=0.7853981633974483:3.9269908169872414,variable=\t]({1*1.4142135623730951*cos(\t r)+0*1.4142135623730951*sin(\t r)},{0*1.4142135623730951*cos(\t r)+1*1.4142135623730951*sin(\t r)});
\draw [shift={(7,3)},line width=2pt,color=ffqqqq]  plot[domain=-0.7853981633974483:2.356194490192345,variable=\t]({1*1.4142135623730951*cos(\t r)+0*1.4142135623730951*sin(\t r)},{0*1.4142135623730951*cos(\t r)+1*1.4142135623730951*sin(\t r)});
\draw [shift={(4,0)},line width=2pt,color=ffqqqq]  plot[domain=2.356194490192345:5.497787143782138,variable=\t]({1*1.4142135623730951*cos(\t r)+0*1.4142135623730951*sin(\t r)},{0*1.4142135623730951*cos(\t r)+1*1.4142135623730951*sin(\t r)});
\draw [shift={(7,0)},line width=2pt,color=ffqqqq]  plot[domain=-2.356194490192345:0.7853981633974483,variable=\t]({1*1.4142135623730951*cos(\t r)+0*1.4142135623730951*sin(\t r)},{0*1.4142135623730951*cos(\t r)+1*1.4142135623730951*sin(\t r)});
\draw [shift={(12,-0.75)},line width=2pt,color=ffqqqq]  plot[domain=-1.5707963267948966:1.5707963267948966,variable=\t]({1*0.75*cos(\t r)+0*0.75*sin(\t r)},{0*0.75*cos(\t r)+1*0.75*sin(\t r)});
\draw [shift={(11,-0.75)},line width=2pt,color=ffqqqq]  plot[domain=1.5707963267948966:4.71238898038469,variable=\t]({1*0.75*cos(\t r)+0*0.75*sin(\t r)},{0*0.75*cos(\t r)+1*0.75*sin(\t r)});
\draw [line width=2pt,color=qqffqq] (3,2)-- (3,1);
\draw [line width=2pt,color=qqffqq] (5,-1)-- (6,-1);
\draw [line width=2pt,color=qqffqq] (8,2)-- (8,1);
\draw [line width=2pt,color=qqffqq] (5,4)-- (6,4);
\draw [line width=2pt,color=qqffqq] (11,-1.5)-- (12,-1.5);
\draw [line width=2pt,color=qqffqq] (11,0)-- (12,0);
\draw [line width=2pt,color=qqffqq] (11,-5)-- (11,-4);
\draw [line width=2pt,color=qqffqq] (9.5,-5)-- (9.5,-4);
\draw [line width=2pt,color=qqqqff] (5,-1)-- (9.5,-5);
\draw [line width=2pt,color=qqqqff] (11,-5)-- (12,-1.5);
\draw [shift={(13.689024390243905,4.243902439024391)},line width=2pt,color=xfqqff]  plot[domain=3.684588139919361:4.286905139486566,variable=\t]({1*10.148776164851574*cos(\t r)+0*10.148776164851574*sin(\t r)},{0*10.148776164851574*cos(\t r)+1*10.148776164851574*sin(\t r)});
\draw [shift={(8.336,-2.346)},line width=2pt,color=xfqqff]  plot[domain=-0.7835177594460383:0.22691844143581524,variable=\t]({1*3.7604005105839455*cos(\t r)+0*3.7604005105839455*sin(\t r)},{0*3.7604005105839455*cos(\t r)+1*3.7604005105839455*sin(\t r)});
\draw [shift={(10.25,-4)},line width=2pt,color=ffqqqq]  plot[domain=0:3.141592653589793,variable=\t]({1*0.75*cos(\t r)+0*0.75*sin(\t r)},{0*0.75*cos(\t r)+1*0.75*sin(\t r)});
\draw [shift={(10.25,-5)},line width=2pt,color=ffqqqq]  plot[domain=3.141592653589793:6.283185307179586,variable=\t]({1*0.75*cos(\t r)+0*0.75*sin(\t r)},{0*0.75*cos(\t r)+1*0.75*sin(\t r)});
\begin{scriptsize}
\draw [fill=black] (3,2) circle (2.5pt);
\draw [fill=black] (5,4) circle (2.5pt);
\draw [fill=black] (6,4) circle (2.5pt);
\draw [fill=black] (8,2) circle (2.5pt);
\draw [fill=black] (5,-1) circle (2.5pt);
\draw [fill=black] (3,1) circle (2.5pt);
\draw [fill=black] (8,1) circle (2.5pt);
\draw [fill=black] (6,-1) circle (2.5pt);
\draw [fill=black] (12,0) circle (2.5pt);
\draw [fill=black] (12,-1.5) circle (2.5pt);
\draw [fill=black] (11,-1.5) circle (2.5pt);
\draw [fill=black] (11,0) circle (2.5pt);
\draw [fill=black] (9.5,-5) circle (2.5pt);
\draw [fill=black] (11,-5) circle (2.5pt);
\draw [fill=black] (11,-4) circle (2.5pt);
\draw [fill=black] (9.5,-4) circle (2.5pt);
\end{scriptsize}
\end{tikzpicture}
         \caption{Sketch of the graph $G'$ constructed in Lemma \ref{2k-opt}. The red curves represent the connecting paths. The green edges are the short edges, the blue edges are the fixed $C$-edges and the yellow edges are the copies of the fixed $C$-edges.}
         \label{2k-opt fig}
    \end{minipage}\qquad%
  \begin{minipage}[t]{0.45\textwidth}
      \centering
      \definecolor{xfqqff}{rgb}{1,1,0}
\definecolor{qqqqff}{rgb}{0,0,1}
\definecolor{qqffqq}{rgb}{0,1,0}
\definecolor{ffqqqq}{rgb}{1,0,0}
\begin{tikzpicture}[line cap=round,line join=round,>=triangle 45,x=0.7cm,y=0.7cm]
\draw [shift={(4,3)},line width=2pt,color=ffqqqq]  plot[domain=0.7853981633974483:3.9269908169872414,variable=\t]({1*1.4142135623730951*cos(\t r)+0*1.4142135623730951*sin(\t r)},{0*1.4142135623730951*cos(\t r)+1*1.4142135623730951*sin(\t r)});
\draw [shift={(7,3)},line width=2pt,color=ffqqqq]  plot[domain=-0.7853981633974483:2.356194490192345,variable=\t]({1*1.4142135623730951*cos(\t r)+0*1.4142135623730951*sin(\t r)},{0*1.4142135623730951*cos(\t r)+1*1.4142135623730951*sin(\t r)});
\draw [shift={(4,0)},line width=2pt,color=ffqqqq]  plot[domain=2.356194490192345:5.497787143782138,variable=\t]({1*1.4142135623730951*cos(\t r)+0*1.4142135623730951*sin(\t r)},{0*1.4142135623730951*cos(\t r)+1*1.4142135623730951*sin(\t r)});
\draw [shift={(7,0)},line width=2pt,color=ffqqqq]  plot[domain=-2.356194490192345:0.7853981633974483,variable=\t]({1*1.4142135623730951*cos(\t r)+0*1.4142135623730951*sin(\t r)},{0*1.4142135623730951*cos(\t r)+1*1.4142135623730951*sin(\t r)});
\draw [shift={(12,-0.75)},line width=2pt,color=ffqqqq]  plot[domain=-1.5707963267948966:1.5707963267948966,variable=\t]({1*0.75*cos(\t r)+0*0.75*sin(\t r)},{0*0.75*cos(\t r)+1*0.75*sin(\t r)});
\draw [shift={(11,-0.75)},line width=2pt,color=ffqqqq]  plot[domain=1.5707963267948966:4.71238898038469,variable=\t]({1*0.75*cos(\t r)+0*0.75*sin(\t r)},{0*0.75*cos(\t r)+1*0.75*sin(\t r)});
\draw [line width=2pt,color=qqffqq] (3,2)-- (3,1);
\draw [line width=2pt,color=qqffqq] (8,2)-- (8,1);
\draw [line width=2pt,color=qqffqq] (5,4)-- (6,4);
\draw [line width=2pt,color=qqffqq] (11,-1.5)-- (12,-1.5);
\draw [line width=2pt,color=qqffqq] (11,0)-- (12,0);
\draw [line width=2pt,color=qqffqq] (11,-5)-- (11,-4);
\draw [line width=2pt,color=qqffqq] (9.5,-5)-- (9.5,-4);
\draw [line width=2pt,color=qqqqff] (5,-1)-- (9.5,-5);
\draw [line width=2pt,color=qqqqff] (11,-5)-- (12,-1.5);
\draw [shift={(8.336,-2.346)},line width=2pt,color=xfqqff]  plot[domain=-0.7835177594460383:0.22691844143581524,variable=\t]({1*3.7604005105839455*cos(\t r)+0*3.7604005105839455*sin(\t r)},{0*3.7604005105839455*cos(\t r)+1*3.7604005105839455*sin(\t r)});
\draw [line width=2pt] (9.5,-5)-- (6,-1);
\draw [shift={(10.25,-4)},line width=2pt,color=ffqqqq]  plot[domain=0:3.141592653589793,variable=\t]({1*0.75*cos(\t r)+0*0.75*sin(\t r)},{0*0.75*cos(\t r)+1*0.75*sin(\t r)});
\draw [shift={(10.25,-5)},line width=2pt,color=ffqqqq]  plot[domain=3.141592653589793:6.283185307179586,variable=\t]({1*0.75*cos(\t r)+0*0.75*sin(\t r)},{0*0.75*cos(\t r)+1*0.75*sin(\t r)});
\begin{scriptsize}
\draw [fill=black] (3,2) circle (2.5pt);
\draw [fill=black] (5,4) circle (2.5pt);
\draw [fill=black] (6,4) circle (2.5pt);
\draw [fill=black] (8,2) circle (2.5pt);
\draw [fill=black] (5,-1) circle (2.5pt);
\draw [fill=black] (3,1) circle (2.5pt);
\draw [fill=black] (8,1) circle (2.5pt);
\draw [fill=black] (6,-1) circle (2.5pt);
\draw [fill=black] (12,0) circle (2.5pt);
\draw [fill=black] (12,-1.5) circle (2.5pt);
\draw [fill=black] (11,-1.5) circle (2.5pt);
\draw [fill=black] (11,0) circle (2.5pt);
\draw [fill=black] (9.5,-5) circle (2.5pt);
\draw [fill=black] (11,-5) circle (2.5pt);
\draw [fill=black] (11,-4) circle (2.5pt);
\draw [fill=black] (9.5,-4) circle (2.5pt);
\end{scriptsize}
\end{tikzpicture}
      \caption{Shortcutting the copies of the $C$-edges (yellow) in the graph $G'$ from Figure~\ref{2k-opt fig} while leaving the connecting paths (red) and fixed $C$-edges (blue) fixed.}
  \label{2k-opt-shortcut fig}
  \end{minipage}
\end{figure}

\begin{figure}[h]
  \begin{minipage}[t]{.45\textwidth}
      \centering
         \definecolor{xfqqff}{rgb}{1,1,0}
\definecolor{qqqqff}{rgb}{0,0,1}
\definecolor{qqffqq}{rgb}{0,1,0}
\definecolor{ffqqqq}{rgb}{1,0,0}
\begin{tikzpicture}[line cap=round,line join=round,>=triangle 45,x=0.7cm,y=0.7cm]
\draw [shift={(4,3)},line width=2pt,color=ffqqqq]  plot[domain=0.7853981633974483:3.9269908169872414,variable=\t]({1*1.4142135623730951*cos(\t r)+0*1.4142135623730951*sin(\t r)},{0*1.4142135623730951*cos(\t r)+1*1.4142135623730951*sin(\t r)});
\draw [shift={(7,3)},line width=2pt,color=ffqqqq]  plot[domain=-0.7853981633974483:2.356194490192345,variable=\t]({1*1.4142135623730951*cos(\t r)+0*1.4142135623730951*sin(\t r)},{0*1.4142135623730951*cos(\t r)+1*1.4142135623730951*sin(\t r)});
\draw [shift={(4,0)},line width=2pt,color=ffqqqq]  plot[domain=2.356194490192345:5.497787143782138,variable=\t]({1*1.4142135623730951*cos(\t r)+0*1.4142135623730951*sin(\t r)},{0*1.4142135623730951*cos(\t r)+1*1.4142135623730951*sin(\t r)});
\draw [shift={(7,0)},line width=2pt,color=ffqqqq]  plot[domain=-2.356194490192345:0.7853981633974483,variable=\t]({1*1.4142135623730951*cos(\t r)+0*1.4142135623730951*sin(\t r)},{0*1.4142135623730951*cos(\t r)+1*1.4142135623730951*sin(\t r)});
\draw [shift={(12,-0.75)},line width=2pt,color=ffqqqq]  plot[domain=-1.5707963267948966:1.5707963267948966,variable=\t]({1*0.75*cos(\t r)+0*0.75*sin(\t r)},{0*0.75*cos(\t r)+1*0.75*sin(\t r)});
\draw [shift={(11,-0.75)},line width=2pt,color=ffqqqq]  plot[domain=1.5707963267948966:4.71238898038469,variable=\t]({1*0.75*cos(\t r)+0*0.75*sin(\t r)},{0*0.75*cos(\t r)+1*0.75*sin(\t r)});
\draw [line width=2pt,color=qqffqq] (3,2)-- (3,1);
\draw [line width=2pt,color=qqffqq] (8,2)-- (8,1);
\draw [line width=2pt,color=qqffqq] (5,4)-- (6,4);
\draw [line width=2pt,color=qqffqq] (11,0)-- (12,0);
\draw [line width=2pt,color=qqqqff] (5,-1)-- (9.5,-5);
\draw [line width=2pt,color=qqqqff] (11,-5)-- (12,-1.5);
\draw [line width=2pt] (11,-4)-- (11,-1.5);
\draw [line width=2pt] (9.5,-4)-- (6,-1);
\draw [shift={(10.25,-4)},line width=2pt,color=ffqqqq]  plot[domain=0:3.141592653589793,variable=\t]({1*0.75*cos(\t r)+0*0.75*sin(\t r)},{0*0.75*cos(\t r)+1*0.75*sin(\t r)});
\draw [shift={(10.25,-5)},line width=2pt,color=ffqqqq]  plot[domain=3.141592653589793:6.283185307179586,variable=\t]({1*0.75*cos(\t r)+0*0.75*sin(\t r)},{0*0.75*cos(\t r)+1*0.75*sin(\t r)});
\begin{scriptsize}
\draw [fill=black] (3,2) circle (2.5pt);
\draw [fill=black] (5,4) circle (2.5pt);
\draw [fill=black] (6,4) circle (2.5pt);
\draw [fill=black] (8,2) circle (2.5pt);
\draw [fill=black] (5,-1) circle (2.5pt);
\draw [fill=black] (3,1) circle (2.5pt);
\draw [fill=black] (8,1) circle (2.5pt);
\draw [fill=black] (6,-1) circle (2.5pt);
\draw [fill=black] (12,0) circle (2.5pt);
\draw [fill=black] (12,-1.5) circle (2.5pt);
\draw [fill=black] (11,-1.5) circle (2.5pt);
\draw [fill=black] (11,0) circle (2.5pt);
\draw [fill=black] (9.5,-5) circle (2.5pt);
\draw [fill=black] (11,-5) circle (2.5pt);
\draw [fill=black] (11,-4) circle (2.5pt);
\draw [fill=black] (9.5,-4) circle (2.5pt);
\end{scriptsize}
\end{tikzpicture}
         \caption{After finishing shortcutting Figure~\ref{2k-opt-shortcut fig} we get the tour $T'$.}
         \label{2k-opt-shortcut-end fig}
    \end{minipage}\qquad%
  \begin{minipage}[t]{0.45\textwidth}
      \centering
      \definecolor{xfqqff}{rgb}{1,1,0}
\definecolor{qqqqff}{rgb}{0,0,1}
\definecolor{qqffqq}{rgb}{0,1,0}
\definecolor{ffqqqq}{rgb}{1,0,0}
\begin{tikzpicture}[line cap=round,line join=round,>=triangle 45,x=0.7cm,y=0.7cm]
\draw [shift={(4,3)},line width=2pt,color=ffqqqq]  plot[domain=0.7853981633974483:3.9269908169872414,variable=\t]({1*1.4142135623730951*cos(\t r)+0*1.4142135623730951*sin(\t r)},{0*1.4142135623730951*cos(\t r)+1*1.4142135623730951*sin(\t r)});
\draw [shift={(7,3)},line width=2pt,color=ffqqqq]  plot[domain=-0.7853981633974483:2.356194490192345,variable=\t]({1*1.4142135623730951*cos(\t r)+0*1.4142135623730951*sin(\t r)},{0*1.4142135623730951*cos(\t r)+1*1.4142135623730951*sin(\t r)});
\draw [shift={(4,0)},line width=2pt,color=ffqqqq]  plot[domain=2.356194490192345:5.497787143782138,variable=\t]({1*1.4142135623730951*cos(\t r)+0*1.4142135623730951*sin(\t r)},{0*1.4142135623730951*cos(\t r)+1*1.4142135623730951*sin(\t r)});
\draw [shift={(7,0)},line width=2pt,color=ffqqqq]  plot[domain=-2.356194490192345:0.7853981633974483,variable=\t]({1*1.4142135623730951*cos(\t r)+0*1.4142135623730951*sin(\t r)},{0*1.4142135623730951*cos(\t r)+1*1.4142135623730951*sin(\t r)});
\draw [shift={(12,-0.75)},line width=2pt,color=ffqqqq]  plot[domain=-1.5707963267948966:1.5707963267948966,variable=\t]({1*0.75*cos(\t r)+0*0.75*sin(\t r)},{0*0.75*cos(\t r)+1*0.75*sin(\t r)});
\draw [shift={(11,-0.75)},line width=2pt,color=ffqqqq]  plot[domain=1.5707963267948966:4.71238898038469,variable=\t]({1*0.75*cos(\t r)+0*0.75*sin(\t r)},{0*0.75*cos(\t r)+1*0.75*sin(\t r)});
\draw [line width=2pt,color=qqffqq] (3,2)-- (3,1);
\draw [line width=2pt,color=qqffqq] (5,4)-- (6,4);
\draw [line width=2pt,color=qqqqff] (5,-1)-- (9.5,-5);
\draw [line width=2pt,color=qqqqff] (11,-5)-- (12,-1.5);
\draw [line width=2pt] (11,-4)-- (11,-1.5);
\draw [line width=2pt] (9.5,-4)-- (6,-1);
\draw [shift={(10.25,-4)},line width=2pt,color=ffqqqq]  plot[domain=0:3.141592653589793,variable=\t]({1*0.75*cos(\t r)+0*0.75*sin(\t r)},{0*0.75*cos(\t r)+1*0.75*sin(\t r)});
\draw [shift={(10.25,-5)},line width=2pt,color=ffqqqq]  plot[domain=3.141592653589793:6.283185307179586,variable=\t]({1*0.75*cos(\t r)+0*0.75*sin(\t r)},{0*0.75*cos(\t r)+1*0.75*sin(\t r)});
\draw [line width=2pt, color=qqqqff] (8,2)-- (11,0);
\draw [line width=2pt] (8,1)-- (12,0);
\begin{scriptsize}
\draw [fill=black] (3,2) circle (2.5pt);
\draw [fill=black] (5,4) circle (2.5pt);
\draw [fill=black] (6,4) circle (2.5pt);
\draw [fill=black] (8,2) circle (2.5pt);
\draw [fill=black] (5,-1) circle (2.5pt);
\draw [fill=black] (3,1) circle (2.5pt);
\draw [fill=black] (8,1) circle (2.5pt);
\draw [fill=black] (6,-1) circle (2.5pt);
\draw [fill=black] (12,0) circle (2.5pt);
\draw [fill=black] (12,-1.5) circle (2.5pt);
\draw [fill=black] (11,-1.5) circle (2.5pt);
\draw [fill=black] (11,0) circle (2.5pt);
\draw [fill=black] (9.5,-5) circle (2.5pt);
\draw [fill=black] (11,-5) circle (2.5pt);
\draw [fill=black] (11,-4) circle (2.5pt);
\draw [fill=black] (9.5,-4) circle (2.5pt);
\end{scriptsize}
\end{tikzpicture}
      \caption{Performing an ambivalent 2-move to $T'$ that decreases the number of short edges (green) by 2 and increases the number of $C$-edges (blue).}
  \label{2k-opt-ambivalent fig}
  \end{minipage}
\end{figure}

Now, we show that the existence of $C$ implies that there is an improving $k$-move or improving alternating cycle of length at most $2k$ contradicting the $k$-optimality or $k$-Lin-Kernighan optimality of $T$.

\begin{lemma} \label{2k-opt}
There is a tour $T'$ containing the connecting paths, $u-1$ $C$-edges and at least $2h-2u+2$ short edges, where $u$ is the number of connected components of $G_3^C$. Moreover, $T\triangle T'$ is an alternating cycle of $T$.
\end{lemma}

\begin{proof}
We construct such a tour $T'$. Let $S$ be a minimal set of $C$-edges that makes $E(G_3^C)$ connected. There is such a set since $T$ consists of the $C$-edges and connecting paths and is connected. Call $S$ the set of \emph{fixed $C$-edges}, colored blue in Figure \ref{2k-opt fig}. Next, let $S'$ be a copy of the fixed $C$-edges, colored yellow in Figure \ref{2k-opt fig}. (The yellow edges are not fixed.) We will call the connecting path edges and the fixed $C$-edges the \emph{fixed} edges. Define the multigraph $G'\coloneq(V(G),E(G_3^C)\sqcup S \sqcup S')$. $G'$ is by construction connected.

We decompose $E(G')$ into cycles: for every connected component in $G_3^C$ we get a cycle by Lemma \ref{connected comp} and for every fixed $C$-edge and the copy of it a cycle with two edges. Moreover, every vertex $b$ with degree greater two has degree four and is the intersection point of two cycles $C_1$ and $C_2$. Note that there are two fixed edges incident to $b$, a connecting path edge and a fixed $C$-edge, one lying on $C_1$ and the other on $C_2$. This implies that there are also two non-fixed edges incident to $b$, one lying on $C_1$ and the other on $C_2$. We call this property the \emph{transverse} property. Now, we can iteratively shortcut $E(G')$ to a tour: In every step we shortcut two cycles intersecting at vertex $b$ to one cycle by shortcutting the two non-fixed edges $\{a,b\}$ and $\{b,c\}$ to $\{a,c\}$ and decrease the number of vertices with degree greater two (Figure~\ref{2k-opt-shortcut fig}). Note that each shortcut does not affect the transverse property at other intersection points.
When this procedure is not possible anymore, every vertex has degree two and since $G'$ was connected we get a tour $T'$ that contains all the fixed edges (Figure~\ref{2k-opt-shortcut-end fig}).

By construction the final tour $T'$ contains $u-1$ fixed $C$-edges and shortcuts of their copies. To connect the $2h$ connecting paths, $T'$ also contains $2h-2u + 2$ short edges.

It remains to prove that $T\triangle T'$ is an alternating cycle. 
By Lemma \ref{short edges}, we know that $A \coloneq T\triangle E(G_3^C)$ is an alternating cycle. Note that the $C$-edges of $E(G_3^C)$ are the tour edges in $A$. The first step in the construction of $T'$ from $E(G_3^C)$ was adding the fixed $C$-edges and a copy of them. By applying this step to $A$, it adds these edges instead of removing them, i.e.\ the fixed $C$-edges change from tour edges to non-tour edges. So $A$ stays a cycle (although not necessarily alternating). The second step was shortcutting consecutive short edges and copies of fixed $C$-edges. As the edges we shortcut are non-tour edges, applying a shortcut on $A$ replaces consecutive non-tour edges with a non-tour edge preserving the property that $A$ is a cycle. Thus, after appying all steps to $A$ it is still a cycle and $T\triangle T'$ is an alternating cycle.
\end{proof}

\begin{remark} \label{rem2h}
The last lemma already gives us a bound on the girth of $G_2$: The length of $T'$ can be bounded by the length of the connecting paths plus $2(u-1)<2h$ $C$-edges and all short edges. Thus, by Lemma \ref{number of blue edges} $T'$ is shorter than $T$. The alternating cycle $T\triangle T'$ consists of $2h-(u-1)\leq 2h$ tour edges, which are the $C$-edges we remove. If $2h\leq k$, this would contradict the $k$-optimality or $k$-Lin-Kernighan optimality of $T$, hence $G_2$ has girth at least $k+1$. 
\end{remark}
Next, we use $T'$ to show Claim~\ref{claim underlying g2 girth}: $G_2$ has girth at least $2k$. 

\begin{definition} \label{def ambivalent 2-move}
Let $T'$ be a tour containing the connecting paths. An \emph{ambivalent 2-move} replaces two edges not belonging to the connecting paths of $T'$ to obtain a new tour containing at least one more $C$-edge.
\end{definition}

\begin{definition}
Fix an orientation of $T'$. A connecting path $p$ is \emph{wrongly oriented} if the orientation of $p$ in $T'$ is opposite to the orientation in $T$. Otherwise, it is \emph{correctly oriented}.
\end{definition}

\begin{figure}[h]
\centering
 \definecolor{qqffqq}{rgb}{0,1,0}
\definecolor{ffqqqq}{rgb}{1,0,0}
\definecolor{qqqqff}{rgb}{0,0,1}
\definecolor{ududff}{rgb}{0,0,0}
\begin{tikzpicture}[line cap=round,line join=round,>=triangle 45,x=1cm,y=1cm]
\draw [line width=2pt,color=qqqqff] (-7,3)-- (-5,3);
\draw [->,shift={(-4,3)},line width=2pt,color=ffqqqq]  plot[domain=0:3.141592653589793,variable=\t]({1*1*cos(\t r)+0*1*sin(\t r)},{0*1*cos(\t r)+1*1*sin(\t r)});
\draw [<-,shift={(-8,3)},line width=2pt,color=ffqqqq]  plot[domain=0:3.141592653589793,variable=\t]({1*1*cos(\t r)+0*1*sin(\t r)},{0*1*cos(\t r)+1*1*sin(\t r)});
\draw [line width=2pt] (-7,1)-- (-5,1);
\draw [->,line width=2pt,color=qqffqq] (-7,3) -- (-7,1);
\draw [->,line width=2pt,color=qqffqq] (-5,3) -- (-5,1);
\begin{scriptsize}
\draw [fill=ududff] (-7,3) circle (2.5pt);
\draw [fill=ududff] (-5,3) circle (2.5pt);
\draw[color=qqqqff] (-5.972681566356373,3.2) node {$e_1$};
\draw [fill=ududff] (-3,3) circle (2.5pt);
\draw [fill=ududff] (-9,3) circle (2.5pt);
\draw [fill=ududff] (-7,1) circle (2.5pt);
\draw [fill=ududff] (-5,1) circle (2.5pt);
\draw[color=black] (-5.972681566356373,1.2) node {$e_2$};
\draw[color=qqffqq] (-6.7,2.0670426685018586) node {$f_1$};
\draw[color=qqffqq] (-4.7,2.0670426685018586) node {$f_2$};
\end{scriptsize}
\end{tikzpicture}
  \caption{Sketch for Lemma \ref{ambivalent-2-move}. The drawn orientation is that of $T'$. The red curves represent oppositely oriented connecting paths connected by a $C$-edge $e_1$. The green edges $f_1$ and $f_2$ are the edges not belonging to the connecting paths of $T'$ incident to $e_1$. The edge $e_2$ connects the other two endpoints of $f_1$ and $f_2$ not incident to $e_1$.}
  \label{replace}
\end{figure}

\begin{lemma} \label{ambivalent-2-move}
If a tour $T'$ contains a short edge and all connecting paths, then there is an ambivalent 2-move that increases the length of the tour by at most two $C$-edges.
\end{lemma}

\begin{proof}
By Lemma \ref{short edges}, every short edge $e$ always connects either two heads or two tails of connecting paths. If in addition $e\in T'$, one of them is correctly oriented and the other one is wrongly oriented. Thus, as long as there is a short edge in $T'$, there has to be at least one correctly oriented and one wrongly oriented connecting path. In this case there has to be a $C$-edge $e_1$ connecting two oppositely oriented connecting paths since the $C$-edges connect the connecting paths to the tour $T$. By definition, every $C$-edge connects a head and a tail of two connecting paths. If $e_1\in T'$, the incident connecting paths would be both correctly or both wrongly oriented. Thus, $e_1$ is not contained in $T'$. Let the two edges not belonging to the connecting paths of $T'$ that share an endpoint with $e_1$ be $f_1$ and $f_2$. 
Because the two connecting paths have opposite orientations, either both tails of $f_1$ and $f_2$, following the orientation of $T'$, serve as endpoints of $e1$ or both heads do. Assume w.l.o.g.\ that they share their tails with $e_1$; let $e_2$ be the edge connecting the heads of $f_1$ and $f_2$ (Figure \ref{replace}). Now, we can make a 2-move replacing $f_1, f_2$ by $e_1$ and $e_2$ to obtain a new tour with the additional $C$-edge $e_1$. The tour stays connected since $e_1$ and $e_2$ connect the tails and heads of $f_1$ and $f_2$, respectively. By Lemma \ref{connecting path}, every connecting path contains at least one edge, hence there are no two adjacent $C$-edges. Thus, $f_1$ and $f_2$ are not $C$-edges and the new tour contains at least one more $C$-edge. 

Moreover, by the triangle inequality we have $c(e_2)\leq c(f_1)+c(e_1)+c(f_2)$ and thus each of the 2-moves increases the length of the tour by at most two $C$-edges.
\end{proof}

\begin{lemma} \label{k-opt}
The given tour $T$ is not $h+1$-optimal and not $h+1$-Lin-Kernighan optimal.
\end{lemma}

\begin{proof}
Let $u$ be the number of connected components of $G_3^C$. By Lemma \ref{2k-opt}, we can construct a tour $T'$ using the connecting paths, $u-1$ $C$-edges and $2h-2u+2$ short edges. We modify $T'$ iteratively to get a tour containing at least $h+1$ $C$-edges. We start with $T'_0=T'$. In the $i$th iteration we perform an ambivalent 2-move on $T'_{i-1}$ to get $T'_i$ (Figure~\ref{2k-opt-ambivalent fig}). Note that there are $2h-2u+2$ short edges in $T'_0$ and with each of these 2-moves, we replaced at most two short edges. 
Therefore, we can perform $s\coloneq \max\{h-u,0\}$ iterations by Lemma \ref{ambivalent-2-move}. As every ambivalent 2-move increases the number of $C$-edges by at least one, $T'_{s}$ has at least $h-1$ $C$-edges and all connecting paths. Thus, $T'_{s}$ arises by an $h+1$-move from $T$. In the beginning the length of $T'_0$ can be bounded by the length of the connecting paths, $2(u-1)$ $C$-edges and copies of $C$-edges and the short edges. In every iteration the cost increases by at most two $C$-edges. Hence, in the end the cost of $T'_{s}$ is bounded from above by the cost of the connecting paths, $2h-2$ $C$-edges and the cost of the short edges. By Lemma \ref{number of blue edges}, $T'_{s}$ is shorter than $T$ which contains $2h$ $C$-edges. 

It remains to show that the $h+1$-move can be performed by augmenting a closed alternating walk. We prove by induction over the iteration number $i$ that $T\triangle T'_i$ is an alternating cycle. In the beginning, by Lemma \ref{2k-opt} $T\triangle T'_0$ is an alternating cycle of $T$. Assume that $T\triangle T'_i$ is an alternating cycle. Let $f_1$ and $f_2$ be replaced by $e_1$ and $e_2$ during the iteration, where $e_1$ is a $C$-edge. Note that $f_1$, $e_1$ and $f_2$ share endpoints on the alternating cycle $T\triangle T'_i$. Moreover, a cycle visits every vertex by definition at most once, hence $f_1$, $e_1$ and $f_2$ are consecutive edges of $T\triangle T'_i$. With the 2-move we shortcut the three consecutive non-tour, tour and non-tour edges of the cycle by the non-tour edge $e_2$, hence $T\triangle T'_{i+1}$ remains an alternating cycle. This completes the proof.
\end{proof}

As we assumed that $h<k$, we conclude that $T$ is also not $k$-optimal and not $k$-Lin-Kernighan optimal. This is a contradiction to the assumption in the beginning that $T$ is $k$-optimal or $k$-Lin-Kernighan optimal and proves Claim~\ref{claim underlying g2 girth}.

\subsection{Bound on the Length of $T$} \label{subsec bound on approx}
We have shown in the last subsection that the girth of $G_2$ is at least $2k$. By leveraging this result, Corollary~\ref{ql bound} provides an upper bound on the number $l$-long edges for a fixed $l$. 

In this subsection we derive an upper bound on the length of the $k$-optimal tour $T$ by combining all upper bounds from different values of $l$. This implies a bound on the approximation ratio as we assumed that the optimal tour has length 1.

\begin{lemma} \label{upper bound qi}
If the number of $l$-long edges $q_l\leq f(l)$ for some function $f$ and all $l$ and $l^*:=\min\{j\in \N\mid \sum_{l=0}^{j} f(l) \geq n\}$, then
\begin{align*}
c(T)\leq \sum_{l=0}^{l^*}f(l)\left(\frac{4k-5}{4k-4}\right)^l.
\end{align*}
\end{lemma}

\begin{proof}
By the definition of $l$-long edges, we have 
\begin{align*}
c(T)\leq \sum_{l=0}^{\infty}q_l\left(\frac{4k-5}{4k-4}\right)^l.
\end{align*}
Since every edge with positive cost is $l$-long for some $l$, we have $\sum_{l=0}^\infty q_l\leq n$. Moreover, $(\frac{4k-5}{4k-4})^l$ is monotonically decreasing in $l$, hence the right hand side is maximized if $q_l$ is maximal for small $l$. Thus, we get an upper bound by assuming that $q_l =f(l)$ for $l\leq l^*$ and $q_l=0$ for $l>l^*$, where $l^*:=\min \{j\in \N\mid \sum_{l=0}^{j} f(l) \geq n\}$.
\end{proof}

\begin{corollary} \label{ex upper bound}
For $l^*:=\min \{j\in \N\mid \sum_{l=0}^{j} 4\ex(4(k-1)\lceil (\frac{4k-4}{4k-5})^{l} \rceil, 2k) \geq n\}$ we have 
\begin{align*}
c(T)\leq \sum_{l=0}^{l^*}\frac{4\ex(4(k-1)\lceil (\frac{4k-4}{4k-5})^{l} \rceil, 2k)}{ (\frac{4k-4}{4k-5})^{l}}.
\end{align*}
\end{corollary}

\begin{proof}
By Corollary \ref{ql bound} and Lemma \ref{upper bound qi}, we get an upper bound by assuming $q_l =4\ex(4(k-1)\lceil (\frac{4k-4}{4k-5})^{l} \rceil, 2k)$ for $l\leq l^*$ and $q_l=0$ otherwise.
\end{proof}

\begin{theorem} \label{upper bound}
If $\ex(x,2k) = O(x^c)$ for some $c>1$, the approximation ratios of the $k$-Opt and $k$-Lin-Kernighan algorithm are $O(n^{1-\frac{1}{c}})$ for \textsc{Metric TSP} where $n$ is the number of vertices.
\end{theorem}

\begin{proof}
Let $d$ be a constant such that $\ex(x,2k)\leq d x^c$. 
By Corollary \ref{ql bound}, we have $q_l \leq 4\ex(4(k-1)\lceil (\frac{4k-4}{4k-5})^{l} \rceil, 2k) \leq 4d\left(8(k-1)\left(\frac{4k-4}{4k-5}\right)^l\right)^{c}$. 
Applying Lemma \ref{upper bound qi} for $f(l)=4d\left(8(k-1)\left(\frac{4k-4}{4k-5}\right)^l\right)^{c}$, we get the following upper bound:
\begin{align*}
c(T)&\leq \sum_{l=0}^{l^*}\frac{4d\left(8(k-1)\left(\frac{4k-4}{4k-5}\right)^l\right)^{c}}{(\frac{4k-4}{4k-5})^l}\\
&= 4d\left(8(k-1)\right)^{c}\sum_{l=0}^{l^*}\left(\frac{4k-4}{4k-5}\right)^{(c-1)l}=4d\left(8(k-1)\right)^{c}\frac{\left(\frac{4k-4}{4k-5}\right)^{(c-1)(l^*+1)}-1}{\left(\frac{4k-4}{4k-5}\right)^{(c-1)}-1}.
\end{align*}
By definition, 
\begin{align*}
\sum_{l=0}^{l^*-1} q_l \leq \sum_{l=0}^{l^*-1}4d\left(8(k-1)\left(\frac{4k-4}{4k-5}\right)^l\right)^{c}= 4d\left(8(k-1)\right)^{c}\frac{\left(\frac{4k-4}{4k-5}\right)^{cl^*}-1}{\left(\frac{4k-4}{4k-5}\right)^{c}-1}<n.
\end{align*}
Thus, $\left(\frac{4k-4}{4k-5}\right)^{cl^*}<\frac{\left(\left(\frac{4k-4}{4k-5}\right)^{c}-1\right)n}{4d\left(8(k-1)\right)^{c}}+1$ and we get
\begin{align*}
c(T)&\leq 4d\left(8(k-1)\right)^{c}\frac{\left(\frac{4k-4}{4k-5}\right)^{(c-1)(l^*+1)}-1}{\left(\frac{4k-4}{4k-5}\right)^{(c-1)}-1}\\
&= \frac{4d\left(8(k-1)\right)^{c}}{\left(\frac{4k-4}{4k-5}\right)^{(c-1)}-1}\left(\left(\frac{4k-4}{4k-5}\right)^{(c-1)}\left(\left(\frac{4k-4}{4k-5}\right)^{cl^*}\right)^{\frac{c-1}{c}}-1\right)\\
&<\frac{4d\left(8(k-1)\right)^{c}}{\left(\frac{4k-4}{4k-5}\right)^{(c-1)}-1}\left(\left(\frac{4k-4}{4k-5}\right)^{(c-1)}\left(\frac{\left(\left(\frac{4k-4}{4k-5}\right)^{c}-1\right)n}{4d\left(8(k-1)\right)^{c}}+1\right)^{\frac{c-1}{c}}-1\right)= O(n^{1-\frac{1}{c}}).
\end{align*}
Since we assumed that the length of the optimal tour is 1, we get the result.
\end{proof}

Combined with Theorem \ref{exupper} we conclude:

\begin{corollary} \label{coro upper}
The approximation ratios of the $k$-Opt and $k$-Lin-Kernighan algorithm are $O(\sqrt[k]{n})$ for \textsc{Metric TSP} where $n$ is the number of vertices.
\end{corollary}

\begin{remark}
When we do not consider $k$ as a constant the above analysis gives us an upper bound of $O(k\sqrt[k]{n})$.
\end{remark}

\section{Comparing the Lower and Upper Bound} \label{sec comparing lower and upper bound}
In this section we compare the lower and upper bound we got from the previous sections for the $k$-Opt algorithm. From Corollary \ref{coro lower} and Corollary \ref{coro upper} we can directly conclude that

\begin{theorem}
The approximation ratio of the $k$-Opt algorithm is $\Theta(\sqrt[k]{n})$ for $k=3,4,6$ where $n$ is the number of vertices.
\end{theorem}

Now, we want to compare the bounds for other values of $k$ where the exact behavior of $\ex(n,2k)$ is still unknown.

\begin{lemma}\label{ex2x}
For all $x\geq 2$ we have $\ex(2x,2k)\leq 6\ex(x,2k)$.
\end{lemma}

\begin{proof}
The proof uses the standard probabilistic method developed by Erd\H{o}s (see for example \cite{DBLP:books/wi/AlonS92}).

By definition there exists a graph $H$ with $2x$ vertices, girth at least $2k$, and $\ex(2x, 2k)$ edges. We will split $V(H)$ into two subsets of size $x$ such that their induced subgraphs together have at least $\frac{1}{3}$ of the edges of $H$. Hence at least one induced subgraph $H'$ has at least $\frac{1}{6}$ of the edges of $H$. Being a subgraph of $H$, it has girth at least $2k$. Hence $\ex(x, 2k) \geq \frac{1}{6} \ex(2x, 2k)$. It remains to show that such a split to $V(H)$ exists.

Consider randomly splitting $V(H)$ into two sets. For each edge the probability is $\frac{x-1}{2x-1}$ that the endpoints are in the same set. So the expected number of edges whose endpoints are in the same set is $\frac{x-1}{2x-1}\ex(2x,2k)\geq \frac{1}{3}\ex(2x,2k)$. Hence, there exists a split satisfying this condition.
\end{proof}

\begin{lemma} \label{sample lower bound}
For real numbers $p_1,\dots, p_n$ with $0 \leq p_j\leq 1$ for all $j\in\{1,\dots,n\}$ and $\sum_{j=1}^n  p_j=1$ there exists an instance with $n$ vertices and an approximation ratio of $k$-Opt of $\Omega\left(\sum_{j=1}^n p_j\frac{j}{\ex^{-1}(j,2k)}\right)$.
\end{lemma}

\begin{proof}
By Theorem \ref{lower bound}, there exists for any $1\leq j \leq n$ an instance $I_j$ with at most $n$ vertices and approximation ratio $\Omega\left(\frac{j}{\ex^{-1}(j,2k)}\right)$. We can extend the number of vertices of these instances to $n$ as described in Lemma \ref{monoton}. Now, construct a random instance which is equal to $I_j$ with probability $p_j$ for all $j\in\{1,\dots,n\}$. This instance has the expected approximation ratio of $\Omega\left(\sum_{j=1}^n p_j\frac{j}{\ex^{-1}(j,2k)}\right)$. Hence, there is a deterministic instance with an approximation ratio of this value.
\end{proof}

Next, we show that the upper bound from Corollary \ref{ex upper bound} on the approximation ratio of the $k$-Opt algorithm is tight up to a factor of $O(\log(n))$.

\begin{theorem}
The approximation ratio of the k-Opt algorithm is between 
\begin{align*}
x\coloneq\sum_{l=0}^{l^*}\frac{4\ex(4(k-1)\lceil (\frac{4k-4}{4k-5})^{l} \rceil, 2k)}{ (\frac{4k-4}{4k-5})^{l}}
\end{align*}
and 
\begin{align*}
  \Omega\left(\frac{x}{\log(n)}\right),
\end{align*}
where $l^*:=\min \{j\in \N\mid \sum_{l=0}^{j} 4\ex(4(k-1)\lceil (\frac{4k-4}{4k-5})^{l} \rceil, 2k) \geq n\}$ and $n$ is the number of vertices.
\end{theorem}

\begin{proof}
By Corollary \ref{ex upper bound} and Lemma \ref{ex2x}, we get an upper bound for the approximation ratio of the $k$-Opt algorithm of
\begin{align*}
&\sum_{l=0}^{l^*}\frac{4\ex(4(k-1)\lceil (\frac{4k-4}{4k-5})^{l} \rceil, 2k)}{ (\frac{4k-4}{4k-5})^{l}}\\
\leq &\sum_{l=0}^{l^*-1}\frac{4\ex(4(k-1)\lceil (\frac{4k-4}{4k-5})^{l} \rceil, 2k)}{(\frac{4k-4}{4k-5})^{l}}+\frac{4\ex(4(k-1)\lceil (\frac{4k-4}{4k-5}) \rceil\lceil (\frac{4k-4}{4k-5})^{l^*-1} \rceil, 2k)}{(\frac{4k-4}{4k-5})^{l^*}}\\
\leq &\sum_{l=0}^{l^*-1}\frac{4\ex(4(k-1)\lceil (\frac{4k-4}{4k-5})^{l} \rceil, 2k)}{(\frac{4k-4}{4k-5})^{l}}+\frac{24\ex(4(k-1)\lceil (\frac{4k-4}{4k-5})^{l^*-1} \rceil, 2k)}{(\frac{4k-4}{4k-5})^{l^*}}\\
\leq &28\sum_{l=0}^{l^*-1}\frac{\ex(4(k-1)\lceil (\frac{4k-4}{4k-5})^{l} \rceil, 2k)}{(\frac{4k-4}{4k-5})^{l}}.
\end{align*} 
By the definition of $l^*$, we have $\ex(4(k-1)\lceil (\frac{4k-4}{4k-5})^{l} \rceil,2k)<n$ for all $l<l^*$. Hence, we can use Lemma~\ref{sample lower bound} with
\begin{align*}
p_i=
\begin{cases}
  \frac{1}{l^*}, & \text{for } i\in \{\ex(4(k-1)\lceil (\frac{4k-4}{4k-5})^{l} \rceil, 2k) \mid 0\leq l < l^* \}\\
  0, & \text{otherwise}
\end{cases}
\end{align*}
to get a lower bound of
\begin{align*}
\Omega\left(\frac{1}{l^*}\sum_{l=0}^{l^*-1}\frac{\ex(4(k-1)\lceil (\frac{4k-4}{4k-5})^{l} \rceil, 2k)}{4(k-1)\lceil (\frac{4k-4}{4k-5})^{l} \rceil}\right).
\end{align*}
The upper and lower bound differ by a factor of $\Theta(l^*)$. By the definition of $l^*$ and applying Theorem \ref{exlower} with $4(k-1)\left\lceil \left(\frac{4k-4}{4k-5}\right)^{l} \right\rceil$ vertices, there is a constant $C$ such that
\begin{align*}
n&>\sum_{l=0}^{l^*-1} \ex\left(4(k-1)\left\lceil \left(\frac{4k-4}{4k-5}\right)^{l} \right\rceil, 2k\right)\geq \sum_{l=0}^{l^*-1} C \left(4(k-1)\left\lceil \left(\frac{4k-4}{4k-5}\right)^{l} \right\rceil\right)^{1+\frac{2}{3k-5}}\\
&\geq C (4(k-1))^{1+\frac{2}{3k-5}} \sum_{l=0}^{l^*-1}\left(\frac{4k-4}{4k-5}\right)^{l(1+\frac{2}{3k-5})}=C (4(k-1))^{1+\frac{2}{3k-5}}\frac{\left(\frac{4k-4}{4k-5}\right)^{l^*(1+\frac{2}{3k-5})}-1}{\left(\frac{4k-4}{4k-5}\right)^{(1+\frac{2}{3k-5})}-1}.
\end{align*}
Thus, $l^* = \Theta(\log(n))$ and the upper bound is tight up to a factor of $O(\log(n))$.
\end{proof}

\section{Lower Bound for Graph TSP} \label{sec lower bound graph tsp}
In this section we show a lower bound of $\Omega\left(\frac{\log(n)}{\log\log(n)}\right)$ on the approximation ratio of the $k$-Opt algorithm for \textsc{Graph TSP}. For all positive integers $f$ we first construct an instance with at most $4(2f)^{2kf}$ vertices and a $k$-optimal tour $T$ with an approximation ratio of at least $\frac{f}{4}$. 

For the construction, note that we have $\frac{(2f-1)^{2kf-1}-1}{2f-2}\leq (2f)^{2kf}$, hence by Theorem~\ref{regular high girth} there exists a $2f$-regular graph with girth at least $2kf$ and $2(2f)^{2kf}$ vertices. Let $G$ be a connected component of this graph. By construction, we know that $G$ is Eulerian. Now, we construct a $k$-optimal tour $T$ of a graph similar to $G$.

\begin{definition} \label{Graph k opt def T}
Let $W=(v_0,v_1,\dots, v_{\lvert E(G) \rvert-1})$ be a Eulerian walk of $G$. Traverse through $G$ according to $W$ starting at $v_0$ and mark every $f$th vertex both in $G$ and in $W$. Whenever we would mark an already marked vertex $v$ in $G$, we add a new copy $v'$ of $v$ adjacent exactly to the neighbors of $v$ and mark $v'$ instead. Moreover, we replace this occurrence of $v$ in $W$ by $v'$ and mark $v'$. Let $G'$ be the graph containing $G$ and all the copies of the vertices we made. After the traversal of $W$, we mark for every unmarked vertex in $G'$ one arbitrary occurrence of it in $W$. The tour $T$ consists of the edges connecting consecutive marked vertices in $W$.
\end{definition}

We only need the property that every vertex of $G'$ is marked somewhere in $W$, hence it does not matter which occurrence we mark in $W$ for the unmarked vertices in $G'$. Note that the number of edges in $W$ is $f\lvert V(G) \rvert$ since $G$ is $2f$ regular. Hence, we added at most $\lvert V(G) \rvert-1$ copies of vertices to $G$ to obtain $G'$. Therefore, we have $\lvert V(G') \rvert< 2\lvert V(G) \rvert$. Next, we show that $T$ is a tour with length $f\lvert V(G) \rvert$ and it is $k$-optimal. This will conclude the lower bound on the approximation ratio.

\begin{lemma} \label{length graph}
$T$, as defined in Definition \ref{Graph k opt def T}, is a tour of $G'$ with length $f\lvert V(G) \rvert$.
\end{lemma}

\begin{proof}
By construction, we marked every vertex of $G'$ exactly once. Hence, $T$ visits every vertex of $G'$ exactly once in $W$ and is a tour. It remains to show that the length of $T$ is $f\lvert V(G) \rvert$. For that, we show that every edge of $T$ has the same length as the shorter of the two walks in $W$ between the two consecutive marked endpoints. This implies the statement since $W$ consists of $\lvert E(G) \rvert=f\lvert V(G) \rvert$ edges. First, note that two consecutive marked vertices of $W$ have distance at most $f$ in $G'$ since we marked every $f$th vertex at the beginning of the construction and two consecutive vertices of the Eulerian walk are connected by an edge in $G'$. Now, assume that the distance of two consecutive marked vertices $u$ and $v$ is not equal to the length of the shorter walk between these vertices in $W$. Then, the walk between $u$ and $v$ in $W$ is not the shortest path between them. Hence, there are at least two distinct walks in $G'$ between $u$ and $v$ that are together shorter than $2f$. Now, transfer the two walks to $G$ by mapping the copies of the vertices to the original vertex. The transferred $u$-$v$ walk in $W$ uses every edge at most once since $W$ is an Eulerian walk of $G$. Thus, there has to be an edge of the transferred $u$-$v$ walk that does not occur in the transferred shortest $u$-$v$ path, otherwise the transferred shortest path between $u$ and $v$ cannot be shorter. Hence, the union of the two has to contain at least one cycle with length less than $2f$ contradicting the girth of the graph $G$.
\end{proof}

\begin{lemma} \label{k optimal graph}
The tour $T$, as defined in Definition \ref{Graph k opt def T}, is $k$-optimal.
\end{lemma}

\begin{proof}
This proof is similar to the proof of Theorem 3.5 in \cite{Chandra}.

Assume that there is an improving $k$-move. Then, this $k$-move can be decomposed into alternating cycles. Since the $k$-move is improving, at least one alternating cycle has positive gain. Choose such a cycle $C$; it consists of at most $k$ tour edges. By construction all tour edges of the cycle have length at most $f$, so the total length of the tour edges is at most $kf$. Since $C$ has positive gain, the non-tour edges have a total length of less than $kf$. Recall that we showed in the proof of Lemma \ref{length graph} that the shorter of the two walks in $W$ between consecutive marked vertices is a shortest path between them. Now, consider for all tour edges in $C$ the corresponding walk in $W$ between the endpoints and call these edges in the walk \emph{tour-$W$-edges}. For all non-tour edges in $C$ consider the shortest path in $G'$ and call them \emph{non-tour-$G'$-edges}. The union of the tour-$W$-edges and non-tour-$G'$-edges is a closed walk of length less than $2kf$ in $G'$. We map the closed walk to $G$ by mapping the copies of a vertex to the original vertex. Note that every tour-$W$-edge occurs at most once in this closed walk since $W$ is a Eulerian walk of $G$. Thus, there has to be a tour-$W$-edge that does not occur a second time as a non-tour-$G'$-edge, otherwise the cost of the non-tour-$G'$-edges is not strictly less than that of the tour-$W$-edges. Hence, the closed walk contains a cycle with length less than $2kf$ contradicting the girth of $G$.
\end{proof}

\begin{lemma} \label{given f high ratio graph}
For all positive integers $f$ there exists an instance of \textsc{Graph TSP} with at most $4(2f)^{2kf}$ vertices and approximation ratio of at least $\frac{f}{4}$ for the $k$-Opt algorithm.
\end{lemma}

\begin{proof}
By construction (Definition \ref{Graph k opt def T}), $G'$ has at most $2\lvert V(G) \rvert\leq 4(2f)^{2kf}$ vertices and by Lemmas \ref{length graph} and \ref{k optimal graph}, $T$ is a $k$-optimal tour of $G'$ with length $f\lvert V(G) \rvert$. By the double tree algorithm (see for example \cite{Korte:2007:COT:1564997}), we can bound the length of the optimal tour by twice the cost of the minimum spanning tree. In the special case of \textsc{Graph TSP} this is at most $2 (\lvert V(G')\rvert -1)< 2(2\lvert V(G) \rvert-1)< 4\lvert V(G) \rvert$ since the minimum spanning tree consists only of edges of cost 1. Hence, the approximation ratio is at least $\frac{f}{4}$. 
\end{proof}

\begin{lemma} \label{increase vertices graph}
For all positive integers $f$ and $n\geq 4(2f)^{2kf}$ there exists an instance of \textsc{Graph TSP} with $n$ vertices and approximation ratio of at least $\frac{f}{8}$ for the $k$-Opt algorithm.
\end{lemma}

\begin{proof}
Let $G'$ and $T$ be constructed as above. For nonnegative integers $a,b$ we construct a graph $G'_{a,b}$ from $G'$. Choose an arbitrary vertex $v\in V(G')$ and let $G'_1,\dots, G'_a$ be $a$ copies of $G'$ and $v_1\dots, v_a$ be the corresponding vertices of $v$ in these copies. Let $V(G'_{a,b})$ be the union of the vertices in $V(G'_i)$, $1\leq i\leq a$, and $b$ extra vertices $v_{a+1},\dots, v_{a+b}$, and $E(G'_{a,b})$ be the union of the edges in $E(G'_i)$ together with the edges $\{v_i,v_{i+1}\}$ for $i\in \{1,\dots, a+b-1\}$. We call the edges of the form $\{v_i,v_{i+1}\}$ the \emph{connecting edges}. Consider the copies of the tour $T$ for each of the graphs $G'_1,\dots, G'_a$. Assemble the tour copies together with two copies of the connecting edges and shortcut to a tour $T'$ of $G'_{a,b}$. The length of $T'$ is $a\lvert V(G) \rvert f+2(a+b-1)$. 

Next, we show that $T'$ is still $k$-optimal. Assume that there is an improving $k$-move and applying it results in a shorter tour $T''$. For every tour edge in the $k$-move we replace it by the corresponding walk according to $W$ possibly connected by connecting edges between different copies. For every non-tour edge in the $k$-move we replace it by the shortest path in $G'_{a,b}$, respectively. We call these edges \emph{tour}-$W$-edges and \emph{non-tour}-$G'_{a,b}$-edges, respectively. Note that by construction $T'$ crosses the cut beween $v_i$ and $v_{i+1}$ twice for all $i$. As the new tour $T''$ has to cross these cuts an even number of times, it has to cross it at least twice. That means that the multiset of tour-$W$-edges does not contain any connecting edge and the multiset of non-tour-$G'_{a,b}$-edges contains each connecting edge an even number of times. Thus, we can split the union of the tour-$W$-edges and the non-tour-$G'_{a,b}$-edges into a union of cycles (not necessarily alternating) such that each of them either only contains edges in one copy of $G'_{j}$ or only connecting edges. As the multiset of tour-$W$-edges does not contain any connecting edge, there is a cycle with positive gain lying completely in a single copy $G'_j$. We get a contradiction to the girth of $G$ by transforming this cycle to $G$ similar to the proof of Lemma \ref{k optimal graph}.

We choose $a,b$ such that $a\lvert V(G') \rvert+b=n$, $a\geq 1$ and $0\leq b<\lvert V(G') \rvert$ since $\lvert V(G') \rvert<2 \lvert V(G) \rvert \leq 4(2f)^{2kf}\leq n$. In this case $G'_{a,b}$ has $n$ vertices and the approximation ratio is at least 
\begin{align*}
\frac{a\lvert V(G) \rvert f+2(a+b-1)}{2(a\lvert V(G') \rvert+b)}&>\frac{\frac{1}{2}a\lvert V(G') \rvert f+2(a+b-1)}{2(a\lvert V(G') \rvert+b)}> \frac{a\lvert V(G') \rvert f}{4(a\lvert V(G') \rvert+\lvert V(G') \rvert)}\\
&= \frac{af}{4(a+1)}\geq \frac{f}{8}.
\end{align*}
\end{proof}

After constructing the instances in the previous lemma, our next task is to determine the asymptotic relationship between the approximation ratio $f$ and the number of vertices $n$. To facilitate this, we introduce an auxiliary lemma in a more general form, which will be also useful in the analysis of the upper bound.

\begin{lemma} \label{ff asymp}
If $(c_1f)^{c_2f^{c_3}}\leq n$ for constants $c_1,c_2,c_3>0$, then $f = O\left(\left(\frac{\log(n)}{\log\log(n)}\right)^{\frac{1}{c_3}}\right)$. Similarly, if $(c_1f)^{c_2f^{c_3}}\geq n$, then $f = \Omega\left(\left(\frac{\log(n)}{\log\log(n)}\right)^{\frac{1}{c_3}}\right)$.
\end{lemma}

\begin{proof}
By taking the logarithm on both sides of $(c_1f)^{c_2f^{c_3}}\leq n$, we obtain $\log(n) \geq c_2f^{c_3}\log(c_1) + c_2f^{c_3} \log(f) = \Theta(f^{c_3} \log(f))$.
Hence $f^{c_3}\log(f) = O(\log n)$, which implies $\log(f) = O(\log\log(n))$ by taking the logarithm again. The first statement follows from 
\begin{align*}  
f^{c_3}= O\left(\frac{f^{c_3}\log(f)}{\log(f)}\right) = O\left(\frac{\log(n)}{\log\log(n)}\right).
\end{align*}
The second statement can be shown analogously.
\end{proof}

\begin{theorem}
The approximation ratio of $k$-Opt for \textsc{Graph TSP} is $\Omega\left(\frac{\log(n)}{\log\log(n)}\right)$ where $n$ is the number of vertices.
\end{theorem}

\begin{proof}
For all positive integers $f$ and $n$ with $4(2f)^{2kf}\leq n < 4(2(f+1))^{2k(f+1)}$ we get an instance with $n$ vertices and approximation ratio at least $\frac{f}{8}$ by Lemma \ref{increase vertices graph}. Applying Lemma~\ref{ff asymp} we obtain the result $f=\Theta\left(\frac{\log(n)}{\log\log(n)}\right)$.
\end{proof}

\section{Upper Bound for Graph TSP} \label{sec upper bound graph tsp}
In this section we show an upper bound of $O\left(\left(\frac{\log(n)}{\log\log(n)}\right)^{\log_2(9)+\epsilon}\right)$ for all $\epsilon>0$ on the approximation ratio for the 2-Opt algorithm for \textsc{Graph TSP} instances. This implies the same upper bound also for the general $k$-Opt and Lin-Kernighan algorithm since they also produce 2-optimal tours. To show the bound, we assume that a worst-case instance together with a 2-optimal tour is given and bound the length of the tour compared to the length of the optimal tour. Starting with the given instance we iteratively contract a subset of vertices. We show that the cardinality of a subset of the vertices, the so-called \emph{active vertices}, decreases by a factor exponential in the approximation factor after a certain number of iterations. In the end we know that by construction at least one active vertex is remaining. Hence, we can bound the approximation ratio by the number of active vertices at the beginning which is upper bounded by the total number of vertices.

Let an instance $(K_n,c)$ of \textsc{Graph TSP} and a graph $G=(V(K_n),E(G))$ be given such that $c(u,v)$ is the shortest distance between $u$ and $v$ in $G$. Moreover, let $T$ be a 2-optimal TSP tour of this instance. Fix an orientation of $T$ and define $f:=\frac{\sum_{e\in T}c(e)}{n}$. Note that $f$ does not have to be an integer. We may assume that $f>1$ since otherwise $T$ has length $n$ and is optimal.

\begin{definition}
For every edge $(u,v)\in T$ fix a shortest path between $u$ and $v$ in $G$. We call $(u',v')$ a \emph{subedge} of $(u,v)\in T$ if $u'$ and $v'$ lie on the fixed shortest path between $u$ and $v$ in $G$ and $c(u,u')<c(u,v')$.
\end{definition}

Next, we construct a directed multigraph $G_0$ with $V(G_0)=V(G)$. Starting from $G_0$, we iteratively construct the directed multigraph $G_{i+1}$ from $G_i$. We get $V(G_{i+1})$ by contracting subsets of $V(G_i)$ using Theorem~\ref{graph decomposition}. The edge sets $E(G_{i})$ do not depend on the previous graph and can be constructed directly. Throughout the construction, we keep track of functions $p_i: V(G) \to V(G_i)$ that map vertices from graph $G$ to their corresponding images in graph $G_i$. Let $n_i$ be the number of active vertices in $G_i$.

\begin{figure}[h]
  \centering
\definecolor{ffffqq}{rgb}{1.,1.,0.}
\definecolor{qqzzqq}{rgb}{0.,0.6,0.}
\definecolor{ffqqqq}{rgb}{1.,0.,0.}
\definecolor{qqqqff}{rgb}{0.,0.,1.}
\definecolor{ffxfqq}{rgb}{1.,0.5,0.}
\begin{tikzpicture}[line cap=round,line join=round,>=triangle 45,x=0.65cm,y=0.65cm]
\draw [line width=2.pt,color=qqqqff] (-6.18,3.39)-- (-4.14,3.29);
\draw [line width=2.pt,color=qqqqff] (-6.18,3.39)-- (-6.3,1.37);
\draw [line width=2.pt,color=qqqqff] (-4.14,3.29)-- (-4.,1.13);
\draw [line width=2.pt,color=qqqqff] (-6.18,3.39)-- (-4.,1.13);
\draw [line width=2.pt,color=ffqqqq] (-4.,1.13)-- (1.38,-2.53);
\draw [line width=2.pt,color=ffqqqq] (1.38,-2.53)-- (4.22,-2.75);
\draw [line width=2.pt,color=ffqqqq] (1.38,-2.53)-- (1.04,-5.17);
\draw [line width=2.pt,color=qqzzqq] (2.,5.)-- (1.62,3.25);
\draw [line width=2.pt,color=qqzzqq] (1.62,3.25)-- (5.4,5.01);
\draw [line width=2.pt,color=qqzzqq] (2.,5.)-- (3.84,5.15);
\draw [line width=2.pt,color=qqzzqq] (3.84,5.15)-- (5.4,5.01);
\draw [line width=2.pt,color=qqzzqq] (5.4,5.01)-- (4.46,3.09);
\draw [line width=2.pt,color=ffqqqq] (1.62,3.25)-- (1.38,-2.53);
\draw [line width=2.pt,color=ffqqqq] (4.46,3.09)-- (4.22,-2.75);
\draw [line width=2.pt] (-4.88,2.43) circle (1.5049149477628299cm);
\draw [line width=2.pt] (3.44,3.97) circle (1.7373856796923361cm);
\draw [line width=2.pt] (2.16,-3.45) circle (1.7577956245428663cm);
\draw [line width=2.pt,color=ffffqq] (-6.,-2.)-- (-5.6,-3.78);
\draw [line width=2.pt,color=ffffqq] (-5.6,-3.78)-- (-2.96,-3.61);
\draw (-5.116,4.6108) node[anchor=north west] {$V_1^i$};
\draw (3.004,6.4308) node[anchor=north west] {$V_2^i$};
\draw (2.224,-1.3892) node[anchor=north west] {$V_3^i$};
\draw [line width=2.pt] (-4.656,-2.8892) circle (1.43271560332119cm);
\draw [line width=2.pt,color=ffffqq] (-2.96,-3.61)-- (-3.756,-2.1492);
\draw [line width=2.pt,color=ffffqq] (-3.756,-2.1492)-- (-4.876,-2.1692);
\draw [line width=2.pt,color=ffffqq] (-4.876,-2.1692)-- (-4.876,-1.2892);
\draw [line width=2.pt] (-8.816,0.1108)-- (-9.676,-0.8092);
\draw [line width=2.pt] (-9.676,-0.8092)-- (-8.196,-1.7492);
\draw [line width=2.pt] (-8.796,-0.9092) circle (1.1259080779530806cm);
\draw (-4.796,-2.5492) node[anchor=north west] {$X_i$};
\draw (-9.096,-0.2492) node[anchor=north west] {$X_{i-1}$};
\draw [line width=2.pt] (-8.816,0.1108)-- (-6.3,1.37);
\draw [line width=2.pt] (-9.076,2.5908) circle (0.5992703897240373cm);
\draw (-9.736,3.4708) node[anchor=north west] {$X_{i-2}$};
\begin{scriptsize}
\draw [fill=ffxfqq] (-6.18,3.39) circle (2.5pt);
\draw [fill=ffxfqq] (-4.14,3.29) circle (2.5pt);
\draw [fill=ffxfqq] (-6.3,1.37) circle (2.5pt);
\draw [fill=ffxfqq] (-4.,1.13) circle (2.5pt);
\draw [fill=ffxfqq] (1.38,-2.53) circle (2.5pt);
\draw [fill=ffxfqq] (1.04,-5.17) circle (2.5pt);
\draw [fill=ffxfqq] (4.22,-2.75) circle (2.5pt);
\draw [fill=ffxfqq] (2.,5.) circle (2.5pt);
\draw [fill=ffxfqq] (3.84,5.15) circle (2.5pt);
\draw [fill=ffxfqq] (5.4,5.01) circle (2.5pt);
\draw [fill=ffxfqq] (1.62,3.25) circle (2.5pt);
\draw [fill=ffxfqq] (4.46,3.09) circle (2.5pt);
\draw [fill=ffxfqq] (-6.,-2.) circle (2.5pt);
\draw [fill=ffxfqq] (-5.6,-3.78) circle (2.5pt);
\draw [fill=ffxfqq] (-2.96,-3.61) circle (2.5pt);
\draw [fill=ffxfqq] (-3.756,-2.1492) circle (2.5pt);
\draw[color=ffffqq] (-2.986,-2.4892) node {$f_1$};
\draw [fill=ffxfqq] (-4.876,-2.1692) circle (2.5pt);
\draw [fill=ffxfqq] (-4.876,-1.2892) circle (2.5pt);
\draw [fill=black] (-8.816,0.1108) circle (2.5pt);
\draw [fill=black] (-9.676,-0.8092) circle (2.5pt);
\draw [fill=black] (-8.196,-1.7492) circle (2.5pt);
\draw [fill=black] (-9.076,2.5908) circle (2.5pt);
\end{scriptsize}
\end{tikzpicture}
    \caption{Construction of $V(G_{i+1})$: The orange and black vertices are the active and inactive vertices in $G_i$, respectively. The yellow, blue, green and red edges are the edges of $E_0^i$, $E_1^i$, $E_2^i$ and $E_3^i$, respectively. The black edges have at least one inactive vertex in $G_i$ as endpoint and are hence unassigned. Each of the sets $V_1^i$, $V_2^i$ and $V_3^i$ will be contracted to a single vertex in $G_{i+1}$; they will be the active vertices of $G_{i+1}$. All vertices in $X_i$ will become inactive in $G_{i+1}$.}
    \label{construction}
  \end{figure}

\begin{definition}
Fix some $0<\delta<1$ and set $s:=n(f^\delta-1)>0$. Starting with $G_0$ we iteratively construct the directed multigraph $G_{i+1}$ from $G_i$ (Figure \ref{construction}):
\begin{itemize}
\item Let $V(G_0):=V(G)$; we call all vertices of $G_0$ \emph{active}, in particular $n_0\coloneq n$. Moreover, let $p_0(v):=v$ for all $v\in V(G)$.
\item To construct $E(G_{i})$ for all $i\geq 0$ we start with $E(G_{i})=\emptyset$. For every subedge $(u',v')$ of $(u,v)\in T$ with $c(u',v')=9^i$ and such that $9^i$ divides $c(u,u')$ we add the edge $(p_{i}(u'), p_{i}(v'))$ to $G_{i}$ (Figure \ref{giedges}). 
\item To construct $V(G_{i+1})$ from $G_i$ consider the underlying undirected graph of $G_i$ and delete parallel edges. We call the resulting graph $G'_i$. The set of active vertices in $G_i'$ is the same as in $G_i$.
\item By Theorem~\ref{graph decomposition} there is an $n_{i+1}$ with an edge partition $E_0^i,\dots, E_{n_{i+1}}^i$ of the subgraph induced by the active vertices in $G'_i$ satisfying the following conditions: $\epsilon_i:= \frac{s}{8n_i^22^i}$, $\lvert E_0 \rvert \leq \epsilon_i n_i^2$, $n_{i+1}\leq \frac{16}{\epsilon_i}$ and the diameter of $E_j^i$ is at most 4 for all $j>0$. 
\item Define iteratively the sets $V_1^i,\dots, V_{n_{i+1}}^i$ as follows: $V_j^i:=\{v\in V(G_{i})\mid \exists e \in E_j^i, v\in e\} \backslash \cup_{h=1}^{j-1} V_{h}^i$. 
\item We contract the vertices in each of the sets $V_j^i$ to a single vertex, which together with the vertices in $V(G_i)\backslash \left( V_1^i \cup \dots \cup V_{n_{i+1}}^i \right)$ form the vertex set of $G_{i+1}$. 
\item We call the contracted vertices of $V_1^i, \dots, V_{n_{i+1}}^i$ the \emph{active} vertices of $G_{i+1}$, all other vertices of $G_{i+1}$ are called \emph{inactive}. 
\item Note that if a vertex is inactive in $G_{i}$ it is also inactive in $G_{i+1}$. Let $X_{i}:=V(G_{i})\backslash \left(X_1\cup \dots \cup X_{i-1} \cup V_1^i \cup \dots \cup V_{n_{i+1}}^i \right)$ be the set of vertices that become inactive the first time in $G_{i+1}$.
\item Let $p_{i+1}(v)\in V(G_{i+1})$ for all $v\in V(G)$ be the image of $p_{i}(v)$ in $G_{i+1}$.  
\end{itemize}

\end{definition}

\begin{figure}[h]
\centering
 \definecolor{ccqqqq}{rgb}{1,0.,0.}
\begin{tikzpicture}[line cap=round,line join=round,>=triangle 45,x=0.9cm,y=0.9cm]
\draw [line width=2.pt] (-8.,4.)-- (-7,4.);
\draw [line width=2.pt] (-6.,4.)-- (-4,4.);
\draw [->,line width=2.pt] (-3.,4.)-- (0.,4.);
\draw [<-,shift={(-6.5,2.77)},line width=2.pt,color=ccqqqq]  plot[domain=0.6868176497586452:2.454775003831148,variable=\t]({1.*1.939819579239265*cos(\t r)+0.*1.939819579239265*sin(\t r)},{0.*1.939819579239265*cos(\t r)+1.*1.939819579239265*sin(\t r)});
\draw [<-,shift={(-3.5,2.77)},line width=2.pt,color=ccqqqq]  plot[domain=0.6868176497586452:2.454775003831148,variable=\t]({1.*1.939819579239265*cos(\t r)+0.*1.939819579239265*sin(\t r)},{0.*1.939819579239265*cos(\t r)+1.*1.939819579239265*sin(\t r)});
\begin{scriptsize}
\draw [fill=black] (-8.,4.) circle (2.5pt);
\draw [fill=black] (-7.,4.) circle (2.5pt);
\draw [fill=black] (-6.3,4.) circle (0.5pt);
\draw [fill=black] (-6.5,4.) circle (0.5pt);
\draw [fill=black] (-6.7,4.) circle (0.5pt);
\draw [fill=black] (-6.,4.) circle (2.5pt);
\draw [fill=black] (-5.,4.) circle (2.5pt);
\draw [fill=black] (-4.,4.) circle (2.5pt);
\draw [fill=black] (-3.3,4.) circle (0.5pt);
\draw [fill=black] (-3.5,4.) circle (0.5pt);
\draw [fill=black] (-3.7,4.) circle (0.5pt);
\draw [fill=black] (-3.,4.) circle (2.5pt);
\draw [fill=black] (-2.,4.) circle (2.5pt);
\draw [fill=black] (-1.,4.) circle (2.5pt);
\draw [fill=black] (0.,4.) circle (2.5pt);
\node[below=2] at (-8.,4.) {$u_0$};
\node[below=2] at (-7.,4.) {$u_1$};
\node[below=2] at (-6.,4.) {$u_8$};
\node[below=2] at (-5.,4.) {$u_9$};
\node[below=2] at (-4.,4.) {$u_{10}$};
\node[below=2] at (-3.,4.) {$u_{17}$};
\node[below=2] at (-2.,4.) {$u_{18}$};
\node[below=2] at (-1.,4.) {$u_{19}$};
\node[below=2] at (0.,4.) {$u_{20}$};
\end{scriptsize}
\end{tikzpicture}
  \caption{Construction of $E(G_{i})$: Let the fixed shortest path of the edge $(u_0, u_{20})\in T$ be $u_0, u_1, \dots, u_{20}$. For the edge $(u_0, u_{20})$, we add the edges $\{(u_{9^ij}, u_{9^i(j+1)}) \mid j\in \N_0 \land 9^i(j+1)\leq 20 \}$ to $E(G_i)$. The red edges illustrate the edges we add to $E(G_{1})$.}
  \label{giedges}
\end{figure}

In the following we will show that $G_i$ is a simple directed graph and give a lower bound on the number of edges of $G_i$ depending on the constant $\delta$ we fixed above.

\begin{figure}[h]
  \centering
   \begin{tikzpicture}[line cap=round,line join=round,>=triangle 45,x=0.7cm,y=0.7cm]
\draw [line width=2pt] (3.303902439024386,2.213170731707321)-- (5.278536585365849,2.4492682926829303);
\draw [line width=2pt] (6.823902439024385,2.4278048780487835)-- (8.841463414634141,2.27756097560976);
\draw [line width=2pt] (10.472682926829263,1.9341463414634184)-- (12.36146341463414,1.333170731707321);
\draw [line width=2pt] (0.1058536585365818,1.2473170731707355)-- (2,2);
\draw [line width=2pt] (2.6385365853658493,2.0843902439024427) circle (1cm);
\draw [line width=2pt] (6.029756097560971,2.5136585365853694) circle (1cm);
\draw [line width=2pt] (9.614146341463409,2.1273170731707354) circle (1cm);
\draw [line width=2pt] (12.983902439024384,1.0112195121951257) circle (1cm);
\draw [line width=2pt] (-0.49512195121951563,0.9897560975609794) circle (1cm);
\begin{scriptsize}
\draw [fill=black] (3.303902439024386,2.213170731707321) circle (2.5pt);
\node[label=above:$x_2$] at (3.303902439024386,2.213170731707321) {};
\draw [fill=black] (5.278536585365849,2.4492682926829303) circle (2.5pt);
\node[label=above:$y_2$] at (5.278536585365849,2.4492682926829303) {};
\draw [fill=black] (6.823902439024385,2.4278048780487835) circle (2.5pt);
\node[label=above:$x_3$] at (6.823902439024385,2.4278048780487835) {};
\draw [fill=black] (8.841463414634141,2.27756097560976) circle (2.5pt);
\node[label=above:$y_3$] at (8.841463414634141,2.27756097560976) {};
\draw [fill=black] (10.472682926829263,1.9341463414634184) circle (2.5pt);
\node[label=above:$x_4$] at (10.472682926829263,1.9341463414634184) {};
\draw [fill=black] (12.36146341463414,1.333170731707321) circle (2.5pt);
\node[label=above:$y_4$] at (12.36146341463414,1.333170731707321) {};
\draw [fill=black] (0.1058536585365818,1.2473170731707355) circle (2.5pt);
\node[label=above:$x_1$] at (0.1058536585365818,1.2473170731707355) {};
\draw [fill=black] (2,2) circle (2.5pt);
\node[label=above:$y_1$] at (2,2) {};
\draw [fill=black] (-1.0102439024390277,0.903902439024394) circle (2.5pt);
\node[label=above:$u$] at (-1.0102439024390277,0.903902439024394) {};
\draw [fill=black] (13.778048780487797,1.0756097560975648) circle (2.5pt);
\node[label=above:$v$] at (13.778048780487797,1.0756097560975648) {};
\end{scriptsize}
\end{tikzpicture}
    \caption{Sketch for the proof of Lemma \ref{dist vertices}.}
    \label{liinduction}
\end{figure}  

\begin{lemma} \label{dist vertices}
If $p_i(u)=p_i(v)$, then $c(u,v)< 9^i$  for all $u,v\in V(G)$.
\end{lemma}

\begin{proof}
We prove this statement by induction on $i$. For $i=0$ the two vertices $u$ and $v$ have to be identical, hence $c(u,v)=0<1=9^0$. Now, consider the case $i>0$. By construction, either $p_{i-1}(u)=p_{i-1}(v)$ or $p_{i-1}(u), p_{i-1}(v) \in V^{i-1}_j$ for some $j>0$. In the first case we can simply apply the induction hypothesis. In the second case recall that by construction the diameter of $E_j^{i-1}$ is at most 4. Hence, there exists a path of length at most 4 in $G_{i-1}$ connecting $p_{i-1}(u)$ and $p_{i-1}(v)$. W.l.o.g.\ assume the worst case that the path has length 4. Let $(p_{i-1}(x_j),p_{i-1}(y_j))\in E(G_{i-1})$ for $j\in \{1,2,3,4\}$ such that $p_{i-1}(y_j)=p_{i-1}(x_{j+1})$ for $j\in \{1,2,3\}$, $p_{i-1}(x_1)=p_{i-1}(u)$ and $p_{i-1}(y_4)=p_{i-1}(v)$, i.e.\ $(p_{i-1}(x_j),p_{i-1}(y_j))$ are the edges of the path (Figure \ref{liinduction}). We can use the induction hypothesis five times for $c(u,x_1), c(y_4,v)$ and $c(y_j,x_{j+1})$, $j\in \{1,2,3\}$ to bound the distance:
\begin{align*}
c(u,v)&\leq c(u,x_1)+\sum_{j=1}^4c(x_j,y_j)+\sum_{j=1}^3c(y_j,x_{j+1})+c(y_4,v)< 9\cdot 9^{i-1}=9^{i}
\end{align*}
\end{proof}

\begin{lemma} \label{subedge non optimal}
If there are two subedges $(a',b')$ and $(u',v')$ of different edges $(a,b)$ and $(u,v)$ in $T$ with $c(a',b')+c(u',v')> c(a',u')+c(b',v')$, then $T$ is not 2-optimal.
\end{lemma}

\begin{proof}
We have by the triangle inequality
\begin{align*}
c(a,b)+c(u,v)&= c(a,a')+c(a',b')+c(b',b)+c(u,u')+c(u',v')+c(v',v)\\
&> c(a,a')+ c(a',u') +c(u',u)+c(b,b')+ c(b',v')+c(v',v)\\
&\geq c(a,u)+c(b,v).
\end{align*}
Hence, replacing $(a,b)$ and $(u,v)$ by $(a,u)$ and $(b,v)$ is an improving 2-move.
\end{proof}

\begin{lemma} \label{gi simple}
$G_i$ is a simple directed graph with at least $s$ edges for all $i\leq \log_9(f^{1-\delta})$.
\end{lemma}

\begin{proof}
Assume that there are parallel edges $(p_i(a'),p_i(b')), (p_i(u'),p_i(v')) \in E(G_i)$, where $p_i(a')=p_i(u')$ and $p_i(b')=p_i(v')$ for some $a',b',u',v' \in V(G)$. Then, by Lemma \ref{dist vertices} $c(a',u')+ c(b',v')< 9^i+9^i= c(a',b')+c(u',v')$. If $(a',b')$ and $(u',v')$ are subedges of different edges, there is an improving 2-move by Lemma \ref{subedge non optimal} contradicting the 2-optimality of $T$. Otherwise, assume that $(a',b')$ and $(u',v')$ are subedges of an edge $e\in T$. By construction, the fixed shortest paths of $(a',b')$ and $(u',v')$ are disjoint except for possibly one of the endpoints $a',b', u',v'$. So we can w.l.o.g.\ assume that $a',b', u',v'$ lie in this order on the fixed shortest path between the endpoints of $e$ according to the orientation of $T$ (with possibly $b'=u'$). Thus,
\begin{align*}
c(a',v') = c(a',u')+c(u',v')\leq c(a',u')+c(b',v') < c(a',b')+c(u',v')\leq c(a',v'),
\end{align*}
where the strict inequality arises from Lemma \ref{subedge non optimal}. Contradiction. 

Assume that there is a self-loop $(p_i(u),p_i(u'))$ with $p_i(u)=p_i(u')$ for some $u,u' \in V(G)$. By Lemma \ref{dist vertices}, we have $c(u,u')< 9^i=c(u,u')$, contradiction. Hence, $G_i$ is simple.

Note that every edge $e\in T$ produces at least $\lfloor \frac{c(e)}{9^i}\rfloor$ edges in $G_i$. Hence, $G_i$ has in total at least $\sum_{e\in T} \lfloor \frac{c(e)}{9^i} \rfloor \geq \sum_{e\in T} (\frac{c(e)}{9^i}-1)=n(\frac{f}{9^i}-1)$ edges. For $i\leq \log_9(f^{1-\delta})$ we have $9^i\leq 9^{\log_9(f^{1-\delta})}=f^{1-\delta}$. Therefore, we have at least $n(\frac{f}{9^i}-1)\geq n(f^\delta-1)=s$ edges.
\end{proof}

Recall that $n_i$ is the number of active vertices and let $m_i$ be the number of edges where both endpoints are active vertices in $G_i$. Our next aim is to get a lower bound on $m_i$ and an upper bound on $n_i$. 

\begin{lemma} \label{m_i bound}
We have $m_i\geq \frac{s}{2^i}$ for $i\leq \log_9(f^{1-\delta})$.
\end{lemma}

\begin{proof}
Let $\delta_j(v)$ for $v\in V(G_j)$ be the sum of the indegree and outdegree of $v$ in $G_j$. Similarly, let $\delta'_j(v)$ for $v\in V(G'_j)$ be the degree of $v$ in $G_j'$. Since by Theorem \ref{graph decomposition} $\lvert  E_0^j \rvert \leq \epsilon_j n_j^2=\frac{s}{8\cdot 2^j}$, we know that $\sum_{x\in X_j}\delta'_j(x) \leq \frac{s}{4 \cdot 2^j}$. By Lemma \ref{gi simple}, we know that $G_j$ is a simple directed graph, hence we delete at most one parallel edge between every pair of vertices while constructing the graph $G_j'$. This gives us $\sum_{x\in X_j}\delta_{j}(x) \leq \frac{s}{2\cdot 2^j}$ for all $j< i$. The vertices in $X_j$ won't be contracted in future iterations because they will have become inactive. Moreover, $9^i$ is divisible by $9^j$ for all $j<i$. Thus, by construction $\delta_i(x)\leq \delta_{j}(x)$ for all $x\in X_j$ with $j<i$ and hence $\sum_{x\in X_j}\delta_i(x)\leq \sum_{x\in X_j} \delta_{j}(x) \leq \frac{s}{2\cdot 2^j}$.
By Lemma \ref{gi simple}, we have $s\leq \lvert E(G_i) \rvert \leq \sum_{j=0}^{i-1}\sum_{x\in X_j}\delta_i(x)+m_i$. Therefore, 
\begin{align*}
m_i\geq s-\sum_{j=0}^{i-1} \sum_{x\in X_j} \delta_i(x)\geq s- \sum_{j=0}^{i-1} \frac{s}{2^{j+1}}= \frac{s}{2^i}.
\end{align*}
\end{proof}

\begin{lemma} \label{n_i bound}
There is a constant $d>0$ such that $n_i\leq \frac{n}{\left(d(f^\delta-1)\right)^{2^{i}-1}}$.
\end{lemma}

\begin{proof}
By Theorem \ref{graph decomposition}, we can bound the number of active vertices by $n_{i+1}\leq 16\frac{1}{\epsilon_i}=\frac{16\cdot 8 \cdot 2^i \cdot n_i^2}{s}=\frac{2^{i+7} n_i^2}{s}$. Now, we show by induction that $n_i\leq \frac{n^{2^i}2^{2^{i+3}-i-8}}{s^{2^{i}-1}}$. For $i=0$ we have $n_0=n=\frac{n^{2^0}2^{2^{3}-8}}{s^{2^{0}-1}}$. The inductive step is
\begin{align*}
n_{i+1} \leq \frac{2^{i+7}n_i^2}{s}\leq \frac{2^{i+7}}{s} \cdot \frac{n^{2^{i+1}}2^{2^{i+4}-2i-16}}{s^{2^{i+1}-2}}=\frac{n^{2^{i+1}}2^{2^{i+4}-(i+1)-8}}{s^{2^{i+1}-1}}.
\end{align*}
Hence,
\begin{align*}
n_{i} \leq \frac{n^{2^{i}}2^{2^{i+3}-i-8}}{s^{2^{i}-1}} = \frac{n^{2^{i}}2^{2^{i+3}-i-8}}{(n(f^\delta-1))^{2^{i}-1}} \leq \frac{n2^{2^{i+3}-8}}{(f^\delta-1)^{2^{i}-1}}= \frac{n}{(\frac{1}{2^8}(f^\delta-1))^{2^{i}-1}}.
\end{align*}
\end{proof}

\begin{theorem}
The approximation ratio of the 2-Opt algorithm for \textsc{Graph TSP} is $O\left(\left(\frac{\log(n)}{\log\log(n)}\right)^{\log_2(9)+\epsilon}\right)$ for all $\epsilon>0$ where $n$ is the number of vertices.
\end{theorem}

\begin{proof}
By the definition of $f$, we have $\sum_{e\in T}c(e)=nf$. The cost of the optimal tour is at least $n$ since it consists of $n$ edges. Hence, the approximation ratio is at most $f$ and it is enough to get an upper bound on $f$.

Consider the graph $G_{\lfloor\log_9(f^{1-\delta}) \rfloor}$. On the one hand, by Lemma \ref{m_i bound} $m_{\lfloor\log_9(f^{1-\delta}) \rfloor}\geq \frac{s}{2^{\lfloor\log_9(f^{1-\delta}) \rfloor}}=\frac{n(f^\delta-1)}{2^{\lfloor\log_9(f^{1-\delta}) \rfloor}}>0$ and hence $n_{\lfloor\log_9(f^{1-\delta}) \rfloor}\geq 1$. On the other hand, we have by Lemma \ref{n_i bound} $n_{\lfloor \log_9(f^{1-\delta}) \rfloor}\leq \frac{n}{(d(f^\delta-1))^{2^{\lfloor \log_9(f^{1-\delta}) \rfloor}-1}}$ for some constant $d$.
Thus, for all $f\geq 2^{\frac{1}{\delta}}$ there exists a constant $d_1$ such that
\begin{align*}
n&\geq (d(f^\delta-1))^{2^{\lfloor \log_9(f^{1-\delta}) \rfloor}-1}\geq (d_1f^\delta)^{2^{(1-\delta)\log_9(f)-1}-1}= (d_1f^\delta)^{2^{\frac{(1-\delta)\log_2(f)}{\log_2(9)} - 1}-1}\\
&=(d_1f^\delta)^{f^{\frac{1-\delta}{\log_2(9)} -\frac{1}{\log_2(f)}}-1}.
\end{align*}
For a given $\epsilon>0$ we can choose constants $\delta, d_2$ such that for all $f\geq d_2$ we have
\begin{align*}
\delta \left( f^{\frac{1-\delta}{\log_2(9)} -\frac{1}{\log_2(f)}}-1 \right) \geq \frac{\delta}{2} f^{\frac{1}{\log_2(9)+\epsilon}}. 
\end{align*}
By Lemma \ref{ff asymp}, we conclude $f= O\left(\left(\frac{\log(n)}{\log\log(n)}\right)^{\log_2(9)+\epsilon}\right)$. 
\end{proof}

\section{Lower Bound on the Approximation Ratio of the $k$-Improv and $k$-Opt Algorithm} \label{sec 1,2 TSP lower}
In this section we show that the approximation ratio of the $k$-improv algorithm (see Subsection~\ref{subsec k-improv}) is at least $\frac{11}{10}$ for arbitrary fixed $k$. For given fixed $k\geq 2$ and $\epsilon>0$, we construct a $k$-improv-optimal instance $I_{k,\epsilon}$ together with a 2-matching $\widetilde{T}$ with approximation ratio at least $\frac{11}{10}-\epsilon$. Moreover, we show that by connecting a $2k$-improv-optimal 2-matching to a tour, we get a $k$-optimal tour. Thus, this lower bound on the approximation ratio also carries over to the $k$-Opt algorithm. 

\subsection{Construction of the Instance $I_{k,\epsilon}$ and the 2-Matching $\widetilde{T}$} \label{subsec construction instance 2 matching}
We first construct some auxiliary graphs before the construction of the instance. Let $S$ be a graph with 10 vertices $w_0,\dots, w_9$ and the edges $\{w_{0}, w_{1}\}$, $\{w_{0},w_{4}\}$, $\{w_{2},w_{3}\}$, $\{w_{3},w_{4}\}$, $\{w_{5},w_{9}\}$, $\{w_{5},w_{6}\}$, $\{w_{6},w_{7}\}$ and $\{w_{8},w_{9}\}$ (Figure \ref{localGraphS}).

\begin{figure}[ht]
\centering
\begin{tikzpicture}[line cap=round,line join=round,>=triangle 45,x=1cm,y=1cm]
\draw [shift={(-8,4)},line width=2pt]  plot[domain=0:0,variable=\t]({1*3*cos(\t r)+0*3*sin(\t r)},{0*3*cos(\t r)+1*3*sin(\t r)});
\draw [line width=2pt] (0,4)-- (1,4);
\draw [shift={(-6,1.9487223436021441)},line width=2pt]  plot[domain=0.7980546443205417:2.3435380092692517,variable=\t]({1*2.864915360641058*cos(\t r)+0*2.864915360641058*sin(\t r)},{0*2.864915360641058*cos(\t r)+1*2.864915360641058*sin(\t r)});
\draw [shift={(-1,1.9487223436021441)},line width=2pt]  plot[domain=0.7980546443205417:2.3435380092692517,variable=\t]({1*2.864915360641058*cos(\t r)+0*2.864915360641058*sin(\t r)},{0*2.864915360641058*cos(\t r)+1*2.864915360641058*sin(\t r)});
\draw [line width=2pt] (-3,4)-- (-2,4);
\draw [line width=2pt] (-2,4)-- (-1,4);
\draw [line width=2pt] (-8,4)-- (-7,4);
\draw [line width=2pt] (-6,4)-- (-4,4);
\begin{scriptsize}
\draw [fill=black] (-8,4) circle (2.5pt);
\draw[color=black] (-7.8594230151522675,3.7785780973902116) node {$w_0$};
\draw [fill=black] (0,4) circle (2.5pt);
\draw[color=black] (0.1609826461566822,3.7785780973902116) node {$w_8$};
\draw [fill=black] (-7,4) circle (2.5pt);
\draw[color=black] (-6.814239835694997,3.7785780973902116) node {$w_1$};
\draw [fill=black] (-6,4) circle (2.5pt);
\draw[color=black] (-5.835068225466607,3.7785780973902116) node {$w_2$};
\draw [fill=black] (-5,4) circle (2.5pt);
\draw[color=black] (-4.8448946870334035,3.7785780973902116) node {$w_3$};
\draw [fill=black] (-4,4) circle (2.5pt);
\draw[color=black] (-3.8547211486001993,3.7785780973902116) node {$w_4$};
\draw [fill=black] (-3,4) circle (2.5pt);
\draw[color=black] (-2.8315418255525553,3.7785780973902116) node {$w_5$};
\draw [fill=black] (-2,4) circle (2.5pt);
\draw[color=black] (-1.8303663589145385,3.7785780973902116) node {$w_6$};
\draw [fill=black] (-1,4) circle (2.5pt);
\draw[color=black] (-0.8401928204813349,3.7785780973902116) node {$w_7$};
\draw [fill=black] (1,4) circle (2pt);
\draw[color=black] (1.1401542563850724,3.7785780973902116) node {$w_9$};
\end{scriptsize}
\end{tikzpicture}
  \caption{The graph $S$ consists of 10 vertices and the drawn edges.}
  \label{localGraphS}
\end{figure}

Fix $\epsilon > 0$ and an integer $k$. Apply Theorem \ref{regular high girth} for $\delta=4$, $g:=\max\{2k+1,\left\lceil \frac{1}{\frac{11}{11-10\epsilon}-1} \right \rceil\}$ and an integer $m\geq \frac{3^{g-1}-1}{2}$. The theorem guarantees the existence of a 4-regular graph $G$ with $s\coloneq 2m$ vertices and girth at least $g$. 

Next, we construct a graph $G_S$ with the vertex set $V(G_S)=\{v_0,\dots,v_{10s-1}\}$. For simplicity, we consider in the following all indices modulo $10s$. To construct $G_S$ we replace every vertex of $G$ by a copy of $S$. In each copy of $S$, the vertices $w_0,\dots, w_9$ are mapped to $v_{10h}, \dots, v_{10h+9}$ in this order for some unique $h\in \Z$, which we refer to as a \emph{block}. For every edge $\{u,v\}\in E(G)$, we connect two vertices with degree 1 in each of the corresponding blocks (Figure \ref{12g1construction} and Figure \ref{12gSconstruction}). The 4-regularity of $G$ ensures that this procedure can be carried out so that every vertex in $G_S$ has degree 2.

\begin{figure}[!htb]
    \centering
 \definecolor{qqffqq}{rgb}{0,1,0}
\definecolor{ffqqqq}{rgb}{1,0,0}
\definecolor{ffxfqq}{rgb}{1,0.4980392156862745,0}
\definecolor{qqqqff}{rgb}{0,0,1}
\begin{tikzpicture}[line cap=round,line join=round,>=triangle 45,x=1cm,y=1cm]
\draw [line width=2pt,color=qqqqff] (2,7)-- (8,7);
\draw [line width=2pt,color=ffxfqq] (2,7)-- (2,5);
\draw [line width=2pt,color=cyan] (8,7)-- (8,5);
\draw [line width=2pt,color=ffqqqq] (0.845299461620749,5)-- (2,7);
\draw [line width=2pt,color=qqffqq] (2,7)-- (3.1547005383792524,5);
\draw [line width=2pt,color=brown] (6.8452994616207485,5)-- (8,7);
\draw [line width=2pt,color=yellow] (8,7)-- (9.154700538379249,5);
\begin{scriptsize}
\draw [fill=black] (2,7) circle (2.5pt);
\draw [fill=black] (8,7) circle (2.5pt);
\end{scriptsize}
\end{tikzpicture}
  \caption{A part of a graph $G$ from which we will construct a part of the graph $G_S$. The colored edges correspond to those with the same color in Figure \ref{12gSconstruction}.}
  \label{12g1construction}
\end{figure}
\begin{figure}[!htb]
\definecolor{qqffqq}{rgb}{0,1,0}
\definecolor{ffxfqq}{rgb}{1,0.4980392156862745,0}
\definecolor{ffqqqq}{rgb}{1,0,0}
\definecolor{qqqqff}{rgb}{0,0,1}
\begin{tikzpicture}[line cap=round,line join=round,>=triangle 45,x=0.7cm,y=0.7cm]
\draw [shift={(-8,4)},line width=2pt]  plot[domain=0:0,variable=\t]({1*3*cos(\t r)+0*3*sin(\t r)},{0*3*cos(\t r)+1*3*sin(\t r)});
\draw [line width=2pt] (0,4)-- (1,4);
\draw [shift={(-6,1.9487223436021441)},line width=2pt]  plot[domain=0.7980546443205417:2.3435380092692517,variable=\t]({1*2.864915360641058*cos(\t r)+0*2.864915360641058*sin(\t r)},{0*2.864915360641058*cos(\t r)+1*2.864915360641058*sin(\t r)});
\draw [shift={(-1,1.9487223436021441)},line width=2pt]  plot[domain=0.7980546443205417:2.3435380092692517,variable=\t]({1*2.864915360641058*cos(\t r)+0*2.864915360641058*sin(\t r)},{0*2.864915360641058*cos(\t r)+1*2.864915360641058*sin(\t r)});
\draw [line width=2pt] (-3,4)-- (-2,4);
\draw [line width=2pt] (-2,4)-- (-1,4);
\draw [line width=2pt] (-8,4)-- (-7,4);
\draw [line width=2pt] (-6,4)-- (-4,4);
\draw [shift={(5.764114133233071,1.9487223436021441)},line width=2pt]  plot[domain=0.7980546443205417:2.343538009269252,variable=\t]({1*2.864915360641058*cos(\t r)+0*2.864915360641058*sin(\t r)},{0*2.864915360641058*cos(\t r)+1*2.864915360641058*sin(\t r)});
\draw [shift={(10.764114133233072,1.9487223436021441)},line width=2pt]  plot[domain=0.7980546443205471:2.3435380092692517,variable=\t]({1*2.8649153606410427*cos(\t r)+0*2.8649153606410427*sin(\t r)},{0*2.8649153606410427*cos(\t r)+1*2.8649153606410427*sin(\t r)});
\draw [line width=2pt] (3.7641141332330683,4)-- (4.764114133233069,4);
\draw [line width=2pt] (5.764114133233071,4)-- (7.764114133233071,4);
\draw [line width=2pt] (8.764114133233072,4)-- (10.764114133233072,4);
\draw [line width=2pt] (11.764114133233065,4)-- (12.76411413323305,4);
\draw [shift={(5.882057066616532,16.137902002887127)},line width=2pt,color=qqqqff]  plot[domain=4.261135112770937:5.163642847998442,variable=\t]({1*13.488041383634052*cos(\t r)+0*13.488041383634052*sin(\t r)},{0*13.488041383634052*cos(\t r)+1*13.488041383634052*sin(\t r)});
\draw [line width=2pt,color=cyan] (5.764114133233071,4)-- (5.764114133233071,1.9487223436021441);
\draw [line width=2pt,color=ffqqqq] (-7,4)-- (-8.184305707103968,1.9487223436021441);
\draw [line width=2pt,color=qqffqq] (-1,4)-- (0.184305707103966,1.9487223436021441);
\draw [line width=2pt,color=brown] (4.764114133233069,4)-- (3.579808426129101,1.9487223436021441);
\draw [line width=2pt,color=yellow] (10.764114133233072,4)-- (11.948419840337039,1.9487223436021441);
\draw [line width=2pt,color=ffxfqq] (-6,4)-- (-6,1.9487223436021441);
\begin{scriptsize}
\draw [fill=black] (-8,4) circle (2.5pt);
\draw [fill=black] (0,4) circle (2.5pt);
\draw [fill=black] (-7,4) circle (2.5pt);
\draw [fill=black] (-6,4) circle (2.5pt);
\draw [fill=black] (-5,4) circle (2.5pt);
\draw [fill=black] (-4,4) circle (2.5pt);
\draw [fill=black] (-3,4) circle (2.5pt);
\draw [fill=black] (-2,4) circle (2.5pt);
\draw [fill=black] (-1,4) circle (2.5pt);
\draw [fill=black] (1,4) circle (2.5pt);
\draw [fill=black] (3.7641141332330683,4) circle (2.5pt);
\draw [fill=black] (4.764114133233069,4) circle (2.5pt);
\draw [fill=black] (5.764114133233071,4) circle (2.5pt);
\draw [fill=black] (6.764114133233071,4) circle (2.5pt);
\draw [fill=black] (7.764114133233071,4) circle (2.5pt);
\draw [fill=black] (8.764114133233072,4) circle (2.5pt);
\draw [fill=black] (9.764114133233072,4) circle (2.5pt);
\draw [fill=black] (10.764114133233072,4) circle (2.5pt);
\draw [fill=black] (11.764114133233065,4) circle (2.5pt);
\draw [fill=black] (12.76411413323305,4) circle (2.5pt);
\end{scriptsize}
\end{tikzpicture}
  \caption{A part of the graph $G_S$ we constructed from the part of $G$ shown in Figure~\ref{12g1construction}. For its construction replace every vertex of $G$ with a copy of $S$. For every edge in $G$ connect two vertices with degree 1 in the corresponding} blocks.
  \label{12gSconstruction}
\end{figure}

The vertex set of the instance $I_{k,\epsilon}$ is $V(G_S)$. The edge set of $I_{k,\epsilon}$ with cost 1 is the union of $E(G_S)$ with the edges $\{ \{v_{10h+1},v_{10h+2} \},\{v_{10h+4},v_{10h+5}\}, \{v_{10h+7},v_{10h+8}\} \mid h\in \Z \}$. All other edges have cost 2.

The tour $T$ consists of the edges $\{\{v_i,v_{i+1}\} \mid i\in \Z \}$. It has cost $11s$ since all of its edges except $\{\{v_{10h+9},v_{10h+10}\} \mid h\in\Z \}$ have cost 1. Let $\widetilde{T}$ be the 2-matching we get by removing all edges with cost 2 in $T$.

\subsection{Proof of the $k$-improv-optimality of $\widetilde{T}$}
In this subsection, we prove the $k$-improv-optimality of $\widetilde{T}$. Assume that there is an improving $k$-improv-move for $\widetilde{T}$ and $\widetilde{T'}$ is the result after performing it. The $k$-improv-move can be decomposed into edge-disjoint alternating cycles and paths. We choose such a decomposition that contains a minimal number of alternating paths, i.e.\ we cannot merge two alternating paths to a longer alternating path.

\begin{lemma} \label{no alternating cycles}
There are no alternating cycles. Moreover, any augmenting path does not visit any block twice.
\end{lemma}
\begin{proof}
Suppose there exists an alternating cycle. If this cycle traverses at least two blocks, it would imply the existence of a cycle in $G$ with at most $2k<g$ edges, contradicting the girth of $G$. Likewise, any augmenting path does not revisit any block. Now, assume that the cycle only visits one block. Since we chose a decomposition of alternating paths and cycles that is edge-disjoint, the alternating cycle cannot add an edge that is already in $\widetilde{T}$. As the edges $\{w_0,w_4\}$ and $\{w_5,w_9\}$ are the only edges with cost 1 in the block that are not contained in $\widetilde{T}$, they have to be the non-tour edges of the alternating cycle. But then, the edges $\{w_0,w_9\}$ and $\{w_4,w_5\}$ or $\{w_0,w_5\}$ and $\{w_4,w_9\}$ have to be tour edges of the alternating cycle. Both cases are impossible because at least one edge has a cost of 2 and cannot be a tour edge of the alternating cycle.
\end{proof}

Whenever an alternating path can be decomposed into $p_1, p=(w_{0},w_{4},w_{5},w_{9}),p_2$, where $p_1$ and $p_2$ (possibly empty) are subpaths and the vertices of $p$ belong to a single block, we split this alternating path into three alternating paths: $p_1$, $p$ and $p_2$. Since augmenting $p$ creates two cycles within that block by adding the edges $\{w_0,w_4\}$ and $\{w_5,w_9\}$ and removing $\{w_4, w_5\}$, we refer to $p$ as \emph{cycle-creating}.

Starting with $\widetilde{T}$ we augment all cycle-creating alternating paths and let the result be $\widetilde{T_1}$. They create $2q$ cycles in $q$ blocks where $q$ is the number of these paths. We will call these blocks \emph{cycle-containing}. Augmenting $\widetilde{T_1}$ by the remaining alternating paths results in $\widetilde{T'}$. Since $\widetilde{T'}$ is obtained by applying an improving $k$-improv-move from $\widetilde{T}$, the remaining alternating paths must remove an edge from enough of the $2q$ cycles to reduce the number of connected components. Our strategy is to demonstrate that this is impossible.

We distinguish three possible types for the remaining alternating paths:
\begin{enumerate}
\item Alternating paths starting and ending with non-tour edges.
\item Alternating paths starting and ending with one tour edge and one non-tour edge.
\item Alternating paths starting and ending with tour edges.
\end{enumerate}

\begin{lemma}
There are no alternating paths starting and ending with non-tour edges.
\end{lemma}

\begin{proof}
To maintain the property that every vertex in $\widetilde{T'}$ has degree at most 2, these alternating paths have to start and end at a vertex with degree at most 1 in $\widetilde{T_1}$, i.e.\ in $w_0$ or $w_9$ in a block. Then, such an alternating path has to contain the subpath $(w_0,w_4,w_5,w_9)$ from one block. This is impossible because such a subpath would be a cycle-creating alternating path that has already been augmented.
\end{proof}

\begin{lemma}
Alternating paths that start and end with one tour edge and one non-tour edge do not remove any cycle in $\widetilde{T_1}$. 
\end{lemma}

\begin{proof}
Note that there are only four possibilities for such alternating paths, namely such visiting the following vertices of a single block: $(w_0,w_4,w_5)$; $(w_0,w_4,w_3)$; $(w_9,w_5,w_4)$ and $(w_9,w_5,w_6)$. Each of the possibilities adds either $\{w_0,w_4\}$ or $\{w_5,w_9\}$. Thus, the block visited by such a path is not cycle-containing since the decomposition of alternating paths is edge-disjoint. Hence, these alternating paths do not remove any cycle in $\widetilde{T_1}$. 
\end{proof}

Last, alternating paths starting and ending with tour edges decrease the number of edges in the 2-matching by one. As the alternating paths of the previous type do not remove cycles, either these paths remove a cycle created by the cycle-creating alternating paths or the cycle remains in $\widetilde{T'}$. The alternating paths of this type are called \emph{cycle-removing}.

We construct an auxiliary multigraph $G_P$ with $q$ vertices. Each of its vertices corresponds to one cycle-containing block in $\widetilde{T_1}$. For every cycle-removing alternating path $p$ we construct a path $p'$ in $G_P$ we call the \emph{block-path} of $p$ as follows: The vertices of $p'$ are the vertices in $G_P$ corresponding to the cycle-containing blocks visited by $p$. The edges of $p'$ connect the vertices in the order in which the corresponding blocks are visited by $p$ (Figure \ref{Spath Construction}). Note that $p'$ may consist of only one vertex or may be the empty set. Since by Lemma~\ref{no alternating cycles} every alternating path visits any block at most once, $p'$ is indeed a path. The edge multiset $E(G_P)$ consists of the disjoint union of all block-paths.

\begin{figure}[!htb]
    \centering
    \begin{minipage}[b]{.45\textwidth}
        \centering
 \definecolor{ffqqqq}{rgb}{1,0,0}
\begin{tikzpicture}[line cap=round,line join=round,>=triangle 45,x=1cm,y=1cm]
\draw [line width=2pt] (5,5)-- (1,1);
\draw [line width=2pt] (1,1)-- (3,0);
\draw [line width=2pt] (3,0)-- (2,3);
\draw [line width=2pt] (2,3)-- (4,2);
\draw [line width=2pt] (4,2)-- (6,3);
\begin{scriptsize}
\draw [fill=ffqqqq] (1,1) circle (2.5pt);
\draw [fill=black] (2,3) circle (2.5pt);
\draw [fill=black] (4,2) circle (2.5pt);
\draw [fill=ffqqqq] (3,0) circle (2.5pt);
\draw [fill=black] (5,5) circle (2.5pt);
\draw [fill=ffqqqq] (6,3) circle (2.5pt);
\draw [fill=ffqqqq] (5,4) circle (2.5pt);
\end{scriptsize}
\end{tikzpicture}
    \end{minipage}\qquad%
    \begin{minipage}[b]{0.45\textwidth}
        \centering
 \definecolor{ffqqqq}{rgb}{1,0,0}
\begin{tikzpicture}[line cap=round,line join=round,>=triangle 45,x=1cm,y=1cm]
\draw [line width=2pt] (1,1)-- (3,0);
\draw [line width=2pt] (3,0)-- (6,3);
\begin{scriptsize}
\draw [fill=ffqqqq] (1,1) circle (2.5pt);
\draw [fill=ffqqqq] (3,0) circle (2.5pt);
\draw [fill=ffqqqq] (6,3) circle (2.5pt);
\draw [fill=ffqqqq] (5,4) circle (2.5pt);
\end{scriptsize}
\end{tikzpicture}
    \end{minipage}
    \caption{The construction of the blocks-paths. Let $p$ be a cycle-removing alternating path. Left: The edges of $p$ in $G$ after contracting the blocks and removing self-loops. The red vertices are cycle-containing. Right: the corresponding block-path $p'$ in $G_P$ of $p$. It visits the red vertices in the same order as $p$ in $G$.}
    \label{Spath Construction}
\end{figure}

\begin{lemma} \label{GP acyclic}
The graph $G_P$ is acyclic. In particular, it has no parallel edges.
\end{lemma}

\begin{proof}
If $G_P$ contains a cycle, the cycle corresponds to a closed walk in $G$ by considering for every edge in the cycle the corresponding block-path $p'$ it belongs to and the corresponding subpath in $p$. The length of this closed walk is bounded by the number of edges in the cycle-removing alternating paths, which is at most $2k$. As $g>2k$ and $G$ has girth at least $g$ we get a contradiction. 
\end{proof}

\begin{lemma} \label{cycle-removing alternating path remove cycles}
A cycle-removing alternating path whose block-path has $h$ edges can remove at most $h+1$ cycles in $\widetilde{T_1}$.
\end{lemma}

\begin{proof}
It is enough to show that every visit to a vertex along the block-path can eliminate at most one cycle. Assume that a cycle-removing alternating path eliminates two cycles within a block during a single visit. Then, this alternating path has to visit vertices from both sets $\{w_0,w_1,w_2,w_3,w_4\}$ and $\{w_5,w_6,w_7,w_8,w_9\}$ in the block without leaving the block. Thus, it has to remove $\{w_4,w_5\}$, the only edge with cost 1 connecting the two sets. In this case this block is not cycle-containing by the edge-disjoint property of the decomposition as cycle creating alternating paths also have to remove the edge $\{w_4,w_5\}$, contradiction. 
\end{proof}

\begin{lemma} \label{12 k improv optimality T}
The 2-matching $\widetilde{T}$ for the instance $I_{k,\epsilon}$ is $k$-improv-optimal.
\end{lemma}

\begin{proof}
Recall that $\widetilde{T'}$ results from performing an improving $k$-improv move from $\widetilde{T}$. Since $\widetilde{T}$ does not contain singletons, either $\widetilde{T'}$ contains fewer connected components than $\widetilde{T}$ or the same number but more cycles. In both cases $\widetilde{T'}$ contains more edges than $\widetilde{T}$. Note that cycle-creating and cycle-removing alternating paths increases and decreases the number of edges by one, respectively. All other alternating paths do not change the number of edges. Thus, there are at most $r\leq q-1$ cycle-removing alternating paths. Note that $G_P$ can have at most $q-1$ edges since it is a simple acyclic graph with $q$ vertices by Lemma~\ref{GP acyclic}. Thus, by definition the union of the block-paths contains at most $q-1$ edges. By Lemma~\ref{cycle-removing alternating path remove cycles}, we conclude that at most $r+ (q-1)$ cycles can be removed and at least $2q-r-(q-1)=q+1-r$ cycles are remaining. 
Therefore, $\widetilde{T'}$ must contain at least $q+1-r$ more edges than $\widetilde{T}$ to maintain at least the same number of connected components as $\widetilde{T'}$. Since $\widetilde{T'}$ contains $q-r$ more edges than $\widetilde{T}$, we have $q-r\geq q+1-r$, which is a contradiction.
\end{proof}

\subsection{Analyzing the Approximation Ratio}
In this subsection, we examine the approximation ratio of the $k$-improv algorithm by utilizing the tour $T$ and the $k$-improv-optimal 2-matching $\widetilde{T}$, as constructed in Subsection~\ref{subsec construction instance 2 matching}. Additionally, we extend our analysis to include the $k$-Opt algorithm.

\begin{lemma} \label{k-improv does not change T}
If the $k$-improv algorithm starts with the tour $T$, it outputs a tour with the same cost as $T$.
\end{lemma}

\begin{proof}
The $k$-improv algorithm first computes the 2-matching $\widetilde{T}$ from $T$ by removing all edges with cost 2. By Lemma~\ref{12 k improv optimality T}, $\widetilde{T}$ is $k$-improv-optimal, hence the algorithm cannot make any improvements. By construction, $\widetilde{T}$ is cycle-free. Consequently, the algorithm arbitrarily adds edges to connect the paths in $\widetilde{T}$, forming a tour. Suppose an edge with a cost of 1 was added during this step. Adding this edge to $\widetilde{T}$ would decrease the number of connected components, contradicting the $k$-improv-optimality of $\widetilde{T}$. Therefore, all added edges must have a cost of 2 resulting in a tour with the same cost as $T$.
\end{proof}

Now, we need an upper bound on the length of the optimal tour before we can conclude the lower bound on the approximation ratio.

\begin{lemma} \label{12 k length opt}
The optimal tour $T^*$ of the instance $I_{k,\epsilon}$ has cost at most $10s+\frac{10s}{g}$.
\end{lemma}

\begin{proof}
Note that the edges in $G_S$ form disjoint cycles since the degree of every vertex is exactly 2. Moreover, every cycle in $G_S$ corresponds to a closed walk in $G$ if we contract the blocks. Since the girth of the graph $G$ is at least $g$ and $S$ is acyclic, the girth of $G_S$ is also at least $g$. Thus, each of the disjoint cycles in $G_S$ has at least $g$ edges. We can get a tour by removing an arbitrary edge from each cycle and arbitrarily add edges to complete the paths to a tour. Recall that all edges in $G_S$ have cost 1. Hence, we introduced at most one edge with cost 2 for every path with at least $g-1$ edges with cost 1. The constructed tour has a maximum length of $\lvert V(G_S)\rvert +\frac{\lvert V(G_S)\rvert}{g} =10s+\frac{10s}{g}$.
\end{proof}

\begin{theorem} \label{12 lower apx ratio}
The approximation ratio of the $k$-improv algorithm with arbitrarily fixed $k$ for \textsc{(1,2)-TSP} is at least $\frac{11}{10}$.
\end{theorem}

\begin{proof}
The constructed tour $T$ has length $11s$ and by Lemma \ref{k-improv does not change T} the $k$-improv algorithm outputs a tour of this length if starting with $T$. The optimal tour $T^*$ has cost at most $10s+\frac{10s}{g}$ by Lemma \ref{12 k length opt}. Recall that for any fixed $\epsilon>0$ we chose $g\geq \frac{1}{\frac{11}{11-10\epsilon}-1}$ which implies $\frac{11}{10} \cdot \frac{1}{1+\frac{1}{g}} \geq \frac{11}{10}-\epsilon$. Hence, for every $\epsilon>0$ there exists an instance with approximation ratio at least 
\begin{align*}
\frac{c(T)}{c(T^*)}=\frac{11s}{10s+\frac{10s}{g}}=\frac{11s}{10s(1+\frac{1}{g})}=\frac{11}{10} \cdot \frac{1}{1+\frac{1}{g}}\geq \frac{11}{10}-\epsilon.
\end{align*}
\end{proof}

Moreover, we show that we can carry over the result to the $k$-Opt algorithm.

\begin{lemma} \label{12 tour k optimal}
The constructed tour $T$ for the instance $I_{2k,\epsilon}$ is $k$-optimal.
\end{lemma}

\begin{proof}
Assume that there is an improving $k$-move after which augmentation we get a shorter tour $T'$. Let $\widetilde{T'}$ be the 2-matching we obtain by removing all edges with cost 2 from $T'$. Then, $\widetilde{T'}$ must contain fewer connected component than $\widetilde{T}$ and we can perform a $2k$-improv-move to obtain $\widetilde{T'}$ from $\widetilde{T}$. This is a contradiction to the $2k$-improv-optimality of $\widetilde{T}$ by Lemma \ref{12 k improv optimality T}.
\end{proof}

\begin{remark}
The fact that we need a $2k$-improv-optimal (instead of a $k$-improv-optimal) 2-matching to ensure that every corresponding tour is $k$-optimal is caused by the different definitions of the two algorithms. In contrast to a $k$-move where at most $k$ edges can be removed and added a $k$-improv-move is defined such that at most $k$ edges can be removed and added in total.
\end{remark}

Therefore, we can carry over the result to the $k$-Opt algorithm:

\begin{theorem}
The approximation ratio of the $k$-Opt algorithm with arbitrarily fixed $k$ for \textsc{(1,2)-TSP} is at least $\frac{11}{10}$.
\end{theorem}

\subsubsection*{Acknowledgements}
I want to thank Fabian Henneke, Stefan Hougardy, Yvonne Omlor, Heiko Röglin and anonymous reviewers for reading this paper and making helpful remarks. I was supported by the Bonn International Graduate School.


\bibliographystyle{plain}
\bibliography{k-opt-lin-kernighan.bbl}
\end{document}